\documentclass[twocolumn,journal]{IEEEtran}
\IEEEoverridecommandlockouts
\usepackage{cite}
\usepackage{amsmath,amssymb,amsfonts}
\usepackage{algorithmic}
\usepackage{graphicx}
\usepackage{textcomp}
\usepackage{xcolor}
\def\BibTeX{{\rm B\kern-.05em{\sc i\kern-.025em b}\kern-.08em
    T\kern-.1667em\lower.7ex\hbox{E}\kern-.125emX}}

\usepackage{mathrsfs}
\usepackage{dsfont}
\newcommand{\mc}{\mathcal}
\newcommand{\ms}{\mathscr}
\newcommand{\ox}{\otimes}
\newcommand{\NE}{\operatorname{NE}}
\newcommand{\err}{\varepsilon}
\newcommand{\nc}{\newcommand}
\nc{\tr}{\operatorname{Tr}}
\nc{\NN}{{{\mathbb N}}}
\nc{\rf}{\operatorname{rf}}
\nc{\sce}{\operatorname{sc}}
\nc{\DD}{{{\mathbb D}}}
\nc{\density}{\mathscr{D}}
\nc{\PSD}{\ms{P}}
\nc{\RR}{{{\mathbb R}}}
\nc{\supp}{{\operatorname{supp}}}
\nc{\Renyi}{R\'{e}nyi }
\nc{\Sand}{{\scriptscriptstyle  \rm S}}
\nc{\Petz}{{\scriptscriptstyle  \rm P}}
\nc{\cH}{{\cal H}}
\nc{\sA}{{{\mathscr{A}}}}
\nc{\sB}{{{\mathscr{B}}}}
\nc{\reg}{\infty}
\nc{\<}{\langle}
\nc{\rnc}{\renewcommand}
\rnc{\>}{\rangle}
\nc{\fid}{\mathfrak{f}}
\nc{\F}{\mathrm{OP}}
\nc{\1}{\mathds{1}}

\newcommand{\smallotimes}{\mathbin{\text{\scalebox{0.75}{$\otimes$}}}}

\usepackage{amsmath,amsthm}
\newtheoremstyle{boldhead}
  {\topsep}      
  {\topsep}      
  {\normalfont}     
  {}             
  {\bfseries}    
  {. }           
  {.5em}         
  {}             

\theoremstyle{boldhead}
\newtheorem{theorem}{Theorem}
\newtheorem{lemma}[theorem]{Lemma}        
\newtheorem{proposition}[theorem]{Proposition}
\newtheorem{corollary}[theorem]{Corollary}
\newtheorem{definition}[theorem]{Definition}
\newtheorem{remark}[theorem]{Remark}

\usepackage{booktabs}
\usepackage{tabularx}
\usepackage{array}
\usepackage{multirow}
\usepackage{makecell}

\begin{document}

\title{Finite-Blocklength and Exponential Analyses for Universal Entanglement and Coherence Distillation under Various Free Operations}

\author{
\IEEEauthorblockN{Zhiwen Lin}
\IEEEauthorblockA{\textit{Institute for Advanced Study in Mathematics, Harbin Institute of Technology,\\
Harbin, Heilongjiang, 150001, China} \\
21b912029@stu.hit.edu.cn\\}
\and
\IEEEauthorblockN{Ke Li}
\IEEEauthorblockA{\textit{Institute for Advanced Study in Mathematics, Harbin Institute of Technology,\\
Harbin, Heilongjiang, 150001, China} \\
carl.ke.lee@gmail.com\\}
\and
\IEEEauthorblockN{Kun Fang}
\IEEEauthorblockA{\textit{School of Data Science, The Chinese University of Hong Kong, \\Shenzhen, Guangdong, 518172, China}\\
kunfang@cuhk.edu.cn}
}

\maketitle

\begin{abstract}
Historically, entanglement and coherence distillation focused on the distillable rate under full knowledge of the initial state. In this paper, we develop a unified finite-blocklength and asymptotic theory for universal distillation from convex compact sequences of possible states stable under tensor products. We first resolve an open question by showing that two regularized variants of the generalized quantum Hoeffding divergence need not coincide, and establish fundamental properties of the Hoeffding divergence and anti-divergence. 
For universal entanglement distillation, we connect finite-blocklength performance under non-entangling, PPT-preserving, and dually free operations to composite correlated hypothesis testing, yielding reliability functions for non-entangling and PPT-preserving operations characterized by regularized quantum Hoeffding divergences and strong converse bounds via the anti-divergence. 
We recover the entanglement concentration error exponent, zero-rate exponents under a continuity condition, and distillable entanglements, and characterize the dually free cases by regularized relative entropies under restricted measurements. 
Additionally, five concrete noise models justify the black-box assumptions for both resources. For coherence distillation under maximally incoherent and dephasing-covariant incoherent operations, we extend the one-shot semidefinite program characterization to universal inputs; relating a relaxed optimization to hypothesis testing yields the high-rate reliability function. A coherent-seed approach extends the single-state formula to the entire direct regime when the zero-error rate is positive. 
Finally, we completely characterize the states with positive zero-error coherence distillation rate, show that the rate equals the order-zero Petz R\'enyi relative entropy of coherence, and derive dimension-independent one-shot bounds that recover the asymptotic coherence distillation rate without invoking operational inclusions.

\end{abstract}

\section{Introduction}
Entanglement was initially seen as a peculiar physical phenomenon~\cite{EPR, schrodinger1935, schrodinger1936,trimmer1980present,cat}, but is now regarded as a resource~\cite{horodecki2009quantum}. It fuels quantum teleportation~\cite{teleportation} and superdense coding~\cite{superdensecoding}, plays a critical role in quantum key distribution~\cite{BB84,Ekert,RennerPhD}, and enhances quantum~\cite{SmithYard,hastings2009superadditivity} and private~\cite{SmithSmolin,li2009private} communications on noisy channels. 
Coherence is another important quantum feature that marks the departure of quantum mechanics from the classical realm~\cite{streltsov2017colloquium}. Coherence also serves as a key resource for numerous quantum information tasks, such as cryptography~\cite{coles2016numerical}, metrology~\cite{giovannetti2011advances}, randomness generation~\cite{yuan2015intrinsic,cao2016source}, thermodynamics~\cite{aaberg2014catalytic,cwiklinski2015limitations,lostaglio2015quantum,lostaglio2015description,narasimhachar2015low} and biological systems~\cite{plenio2008dephasing,rebentrost2009role}.

Distillation is a central task in quantum resource theories: many noisy resource states are converted into fewer nearly ideal units that can be used in quantum information processing~\cite{chitambar2019quantum}. Entanglement distillation aims to generate maximally entangled states~\cite{bennett1996concentrating,bennett1996purification,bennett1996mixed}, whereas coherence distillation extracts maximally coherent states with respect to a fixed reference basis~\cite{yuan2015intrinsic}.
Over the past two decades, significant attention has been focused on the distillation rate, leading to breakthroughs for entanglement under one-way local operations and classical communication~\cite{onewayLOCC}, (asymptotically) non-entangling operations~\cite{Reversible2008,Reversible2010}, and dually non-entangling operations~\cite{Dually}, while for coherence, the asymptotic distillation rate under incoherent operations (IO)~\cite{winter2016operational}, dephasing-covariant incoherent operations (DIO)~\cite{regula2018one,chitambar2018dephasing}, and maximally incoherent operations (MIO)~\cite{regula2018one} is uniformly characterized by the relative entropy of coherence~\cite{zhao2019one}.

Although their asymptotic rates are well understood, rates alone do not capture the behavior of finite-blocklength experiments or the speed of asymptotic convergence. A complete operational analysis must also quantify the finite-blocklength or one-shot error, the reliability function below the distillable rate, the strong converse exponent above the distillable rate, as well as the zero-rate error exponent and the zero-error distillation rate.
Another practical challenge stems from the incomplete characterization of the initial state.
The standard independent and identically distributed (i.i.d.) setting assumes complete knowledge of a single state, whereas unknown noise or adversarial effects naturally specify only a set of possible states and may introduce correlations across copies.
We therefore study universal distillation from
sequences of convex compact state sets that are stable
under tensor products. The protocol is independent of
a specific initial state and is evaluated in the worst case over
the entire set. 

Before proceeding to the formal analysis of distillation, we first resolve an open question raised by Fang and Hayashi~\cite{fang2025errorexponentsquantumstate} concerning the reliability function: they introduced two versions of the generalized quantum Hoeffding divergence without determining whether their regularized forms coincide. By constructing explicit counterexamples, we provide a definitive negative answer — the two regularized divergences are inequivalent~(Theorem~\ref{them: regularized inequivalent}), thereby clarifying the correct quantity for the reliability function. Furthermore, we show the basic properties of two central quantities~(Proposition~\ref{prop: Hoeffding} and~\ref{prop: Hoeffding anti}) --- the quantum Hoeffding divergence and the quantum Hoeffding anti-divergence --- including the data processing inequality, joint convexity with respect to the two component states, continuity, convexity, and monotonicity in the parameter $r$, subadditivity with respect to product states and stable sequences, positivity and faithfulness, as well as invariance under isometries.

For entanglement distillation, earlier work established one-shot bounds under non-entangling operations~\cite{brandao2011one1}, the error exponent for pure-state concentration under local operations and classical communication (LOCC)~\cite{hayashi2002error}, the zero-rate error exponent under non-entangling operations~\cite{lami2024asymptotic}, and distillable entanglement under non-entangling operations~\cite{Reversible2008} and dually non-entangling operations~\cite{Dually}. We improve and extend these results within a single framework. 
For non-entangling and PPT-preserving operations, we derive an exact equivalence between universal entanglement distillation and composite correlated hypothesis testing in the finite-blocklength regime~(Theorem~\ref{thm: error in finite blocklenth}). To the best of our knowledge, this may be the first result that requires no additional correction terms. As for the dually free operations, we provide upper and lower bounds on the distillation error in terms of the type-II error of hypothesis testing under restricted measurements~(Theorem~\ref{thm: error in finite blocklenth under DNE and DPPT}). These bounds are sufficiently tight for analyzing the error exponent. 
Thus, for all four free operations, the reliability function of universal entanglement distillation can be connected to the error exponent of composite correlated hypothesis testing in the corresponding setting~(Theorem~\ref{them: exponents}). 
In particular, under non-entangling and PPT-preserving operations, the reliability function of universal entanglement distillation is precisely characterized by the regularized quantum Hoeffding divergence (Theorem~\ref{them:rf,NE,PPT}). Concretely, we obtain the following explicit expressions:
\begin{align}
    E^{\rf}(\ms{R},\mc{F}_{\NE},r)&=H_r^\infty(\ms{R}\|\mc{S}),\nonumber\\
    E^{\rf}(\ms{R},\mc{F}_{\mathrm{PPT}},r)&=H_r^\infty(\ms{R}\|\mc{PPT}),\nonumber
\end{align}
where $H_r^\infty(\cdot\|\cdot)$ denotes the regularized quantum Hoeffding divergence, and $\mc{S}$ and $\mc{PPT}$ denote the sets of separable states and PPT states, respectively.
The two expressions endow the regularized quantum Hoeffding divergence with a new operational interpretation. For maximally correlated initial states, our framework yields a single-letter expression for the reliability function~(Corollary~\ref{cor: maximally correlated state}); in the special case of a pure initial state, this reduces to the same error exponent formula 
for entanglement concentration established by Hayashi et al.~\cite{hayashi2002error}.
A complementary recent work~\cite{berta2026tight} combines decoupling-based LOCC achievability with our non-entangling converse to establish the optimal high-rate exponent for maximally correlated states already under one-way LOCC.
Furthermore, our analysis holds for any fixed entanglement distillation rate below the distillable entanglement. For cases where the regularized Petz R\'{e}nyi relative entropy with respect to the set of free states is continuous at $\alpha=1$, our result provides an alternative approach to obtaining the result of Lami et al.~\cite{lami2024asymptotic} by taking the distillation rate to zero~(Section~\ref{sec:Zero-Rate Error Exponent under NE and PPT}). 
By exploiting the symmetry between the two types of errors in Hoeffding's hypothesis testing framework and letting the type-II error exponent tend to zero, we recover the distillable entanglement under non-entangling operations established by Brand\~ao and Plenio~\cite{Reversible2008}, and additionally obtain the distillable entanglement under PPT-preserving operations~(Section~\ref{sec:distillable entanglement under NE and PPT}).
We also establish that the zero-error universal entanglement distillation rate is characterized by the regularized order-zero Petz R\'enyi relative entropy~(Section~\ref{sec:Zero-Error Distillation Rate under NE and PPT}) and that the strong converse exponent is lower-bounded by the regularized quantum Hoeffding anti-divergence for non-entangling and PPT-preserving operations~(Proposition~\ref{pro: strong converse}). 
When the free operations are dually free operations, we characterize both the zero‑rate error exponent and the distillable entanglement in terms of the regularized relative entropy under restricted measurements~(Section~\ref{sec:Exact Zero‑Rate Error Exponent and Distillable Entanglement under Dually Free Operations}). In the special case where the initial state is fully known, our framework reproduces the result of Lami and Regula~\cite{Dually} on the distillable entanglement for the exact state under dually non-entangling operations.

In addition, to further ground the black‑box setting in physical reality, we present five concrete noise models that naturally give rise to the required properties of the black‑box sequence $\ms{R}$~(Section~\ref{sec: Concrete Physical Illustrations of Noise}) --- convexity, compactness, and stability under tensor products, as well as closedness under permutations of the systems and compatibility with separable and PPT measurements --- thereby establishing a clear link between abstract theoretical assumptions and realistic noise processes and providing a physical justification for the abstract formalism adopted throughout this work.

For coherence distillation under MIO and DIO, we extend the equivalence between the distillation fidelity and a semidefinite program, originally established by Regula et al.~\cite{regula2018one}, to the black-box setting~(Theorem~\ref{them:op1}). However, the exact constant-diagonal constraint prevents an immediate identification with hypothesis testing. We resolve this by relaxing the exact optimization problem to a relaxed one, which yields precisely a composite hypothesis testing problem between the set $\mc I$ of incoherent states and the set $\ms{C}$  of initial states~(Lemma~\ref{lem:op2}). We establish quantitative bounds between the exact and relaxed optimization problems~(Lemma~\ref{lem: lower and up bound coherence}), which yield that the reliability function for high distillation rates is given by the regularized quantum Hoeffding divergence~(Theorem~\ref{them:rf for high rate}).  
In particular, when full knowledge of the initial state $\omega$ is available, we obtain a single-letter expression above an explicit critical rate. More concretely, the reliability function takes the form
\[
E^{\rf}(\omega,\mc O,r)
=
\sup_{\alpha\in(0,1)}
\frac{\alpha-1}{\alpha}
\bigl(r-D_{\mathrm{P},\alpha}(\omega\Vert\mc I)\bigr),
\]
where $D_{\mathrm{P},\alpha}(\cdot\Vert\mc I)$ denotes the Petz R\'enyi relative entropy of coherence, and $\mc O$ stands for either MIO or DIO.
We further introduce a seed-assisted approach that connects relaxed tests to exact distillation protocols~(Theorem~\ref{them:relax seed}). A maximally coherent seed gives an exact identity between the two optimization problems~(Theorem~\ref{them:exact seed maximally coherent state}). When the initial state has a positive zero-error distillation rate, such a seed can be distilled from an asymptotically negligible fraction of the input, extending the above reliability formula to the entire direct regime~(Theorem~\ref{them: coherence rf full regime}).
Finally, we provide a complete characterization of the states for which the zero-error distillation rate is positive~(Theorem~\ref{them:coherence structure of state}). Furthermore, again by means of the relaxed optimization program, the zero-error distillation rate is shown to be given by the order-zero Petz R\'enyi  relative entropy of coherence~(Theorem~\ref{them:coherence zero-error rate}). 
We also derive dimension-independent bounds on the one-shot distillation rate in terms of the hypothesis testing relative entropy of coherence~(Theorem~\ref{them:one-shot-coherence}), which directly recover the asymptotic distillation rate under both MIO and DIO without resorting to the operational inclusion relation with IO~(Corollary~\ref{cor:asymptotic-coherence-distillation}).

 Table~\ref{tab:resource-distillation-summary} summarizes the main characterization results established in this work, distinguishing new results, generalizations of previous results, and recovered known results.

\begin{table*}[t]
\centering
\caption{
Summary of resource distillation characterizations.\\ The table distinguishes new results (shown in orange), \\generalizations of previous results (in violet), and recovered known results (in light blue).
}
\label{tab:resource-distillation-summary}

\renewcommand{\arraystretch}{2.6}
\setlength{\tabcolsep}{5pt}

\footnotesize
\begin{tabular}{c|c|c|c|c|c|c}
\hline
\makecell[c]{\textbf{Resource}}
&
\makecell[c]{\textbf{Free}\\\textbf{Operation}}
&
\makecell[c]{\textbf{Finite-Blocklength}\\\textbf{Analysis}}
&
\makecell[c]{\textbf{Reliability}\\\textbf{Function}}
&
\makecell[c]{\textbf{Zero-Rate}\\\textbf{Error Exponent}}
&
\makecell[c]{\textbf{Distillable}\\\textbf{Rate}}
&
\makecell[c]{\textbf{Zero-Error}\\\textbf{Distillation Rate}}
\\
\hline

\multirow{4}{*}{\makecell[c]{Entanglement}}
&
$\mc F_{\NE}$
&
\makecell[c]{
\textcolor{orange}{Exact hypothesis-testing}\\
\textcolor{orange}{equivalence}}
&
\makecell[c]{
\textcolor{orange}{$H_r^\infty(\ms R\Vert\mc S)$}}
&
\makecell[c]{
\textcolor{violet}{$D^\infty(\mc S\Vert\ms R)$}\\
\textcolor{violet}{under endpoint condition}}
&
\makecell[c]{
\textcolor{blue!40!white}{Brand\~ao--Plenio}\\
\textcolor{blue!40!white}{single-state result}}
&
\makecell[c]{
\textcolor{orange}{$D_{\Petz,0}^{\infty}(\ms R\|\mc S)$}}
\\
\cline{2-7}

&
$\mc F_{\mathrm{PPT}}$
&
\makecell[c]{
\textcolor{orange}{Exact hypothesis-testing}\\
\textcolor{orange}{equivalence}}
&
\makecell[c]{
\textcolor{orange}{$H_r^\infty(\ms R\Vert\mc{PPT})$}}
&
\makecell[c]{
\textcolor{violet}{$D^\infty(\mc{PPT}\Vert\ms R)$}\\
\textcolor{violet}{under endpoint condition}}
&
\makecell[c]{
\textcolor{orange}{PPT analogue}\\
\textcolor{orange}{single-state result}}
&
\makecell[c]{
\textcolor{orange}{$D_{\Petz,0}^{\infty}(\ms R\|\mc{PPT})$}}
\\
\cline{2-7}

&
$\mc F_{\mathrm{DNE}}$
&
\makecell[c]{
\textcolor{orange}{$\mathds{SEP}$-measurement}\\
\textcolor{orange}{hypothesis testing bounds}}
&
\makecell[c]{
\textcolor{orange}{$\mathds{SEP}$-measurement}\\
\textcolor{orange}{error exponent bounds}}
&
\makecell[c]{
\textcolor{orange}{$D^{\mathds{SEP},\infty}(\mc S\|\ms R)$}}
&
\makecell[c]{
\textcolor{violet}{Lami--Regula}\\
\textcolor{violet}{universal setting}}
&
--
\\
\cline{2-7}

&
$\mc F_{\mathrm{DPPT}}$
&
\makecell[c]{
\textcolor{orange}{$\mathds{PPT}$-measurement}\\
\textcolor{orange}{hypothesis testing bounds}}
&
\makecell[c]{
\textcolor{orange}{$\mathds{PPT}$-measurement}\\
\textcolor{orange}{error exponent bounds}}
&
\makecell[c]{
\textcolor{orange}{$D^{\mathds{PPT},\infty}(\mc{PPT}\|\ms R)$}}
&
\makecell[c]{
\textcolor{violet}{PPT analogue}\\
\textcolor{violet}{universal setting}}
&
--
\\

\hline

\multirow{2}{*}{\makecell[c]{Coherence}}
&
MIO
&
\makecell[c]{
\textcolor{violet}{Exact optimization}\\
\textcolor{violet}{characterization}}
&
\multirow{2}{*}{\makecell[c]{
\textcolor{orange}{$H_r^\infty(\ms C\Vert\mc I)$}\\
\textcolor{orange}{high-rate regime}\\
\textcolor{orange}{$H_r^\infty(\omega\Vert\mc I)$}\\
\textcolor{orange}{positive zero-error rate}}}
&
--
&
\makecell[c]{
\textcolor{orange}{dimension-independent}\\
\textcolor{orange}{lower bound}}
&
\makecell[c]{
\textcolor{orange}{$D_{\Petz,0}(\omega\|\mc I)$}}
\\
\cline{2-3}\cline{5-7}

&
DIO
&
\makecell[c]{
\textcolor{violet}{Exact optimization}\\
\textcolor{violet}{characterization}}
&
&
--
&
\makecell[c]{
\textcolor{orange}{dimension-independent}\\
\textcolor{orange}{lower bound}}
&
\makecell[c]{
\textcolor{orange}{$D_{\Petz,0}(\omega\|\mc I)$}}
\\

\hline
\end{tabular}
\end{table*}

\section{Notation and Preliminaries}
\subsection{Notation}
A quantum system $A$ is mathematically represented by a Hilbert space $\mc{H}_A$. The set of positive semi-definite operators on $\mc{H}_A$ is denoted as $\ms{P}(\mc{H}_A)$. 
The support of an operator $R$ is denoted by $\supp(R)$.
Quantum states of system $A$ are modeled as density operators, i.e., positive semi-definite operators on $\mc{H}_A$ with trace one. Let $\ms{D}(\mc{H}_A)$ be the set of density operators on $\mc{H}_A$. We omit the Hilbert space subscript whenever there is no ambiguity. For entanglement theory, we consider the bipartite quantum system $AB$ with Hilbert space $\mc{H}_{AB}:=\mc{H}_A\otimes\mc{H}_B$. A bipartite positive semi-definite operator $R\in\ms{P}(\mc{H}_{AB})$ is called \emph{separable} if it can be decomposed in the form $R=\sum_{i=1}^{N}P_i\ox Q_i$, where $P_i$ and $Q_i$ are positive semi-definite operators on $\mc{H}_A$ and $\mc{H}_B$ respectively, and $N\in\mathbb{N}$, and is called the \emph{positive-partial-transpose (PPT)} if it satisfies $R^\Gamma\geq 0$,  where $\Gamma$ denotes the partial transpose. Separable states and PPT states are unit‑trace operators that are separable and PPT, respectively. Let $\mc{S}_n$ denote the set of separable states and $\mc{PPT}_n$ the set of PPT states on $\mathcal{H}_{AB}^{\otimes n}$ with respect to the bipartition $A_1,\dots,A_n:B_1,\dots,B_n$. In this paper, they are instances of the set of free states $\ms{F}_n$. We denote by $\ms{F}=\{\ms{F}_n\}_{n\in\mathbb{N}}$ the sequence of free state sets\footnote{The notation $\ms{F}$ is abused to denote sets of states or sequences thereof, depending on context; likewise for $\mc{S},\ms{A}$, $\ms{R}$, etc. in what follows.}, as well as by $\mc{S}=\{\mc{S}_n\}_{n\in\mathbb{N}}$ and $\mc{PPT}=\{\mc{PPT}_n\}_{n\in\mathbb{N}}$ the sequences for separable and PPT states, respectively. A quantum state is called \emph{entangled} if it is not separable, e.g., the \emph{two-qubit maximally entangled state}  $\Phi_{2}:=|\Phi_2\rangle\langle\Phi_2|$, where $|\Phi_2\rangle:=\frac{1}{\sqrt{2}}|00\rangle+|11\rangle.$ 

Let $\mathcal{H}$ be a finite-dimensional Hilbert space with a fixed orthonormal basis $\{|i\rangle\}$. A state is said to be \emph{incoherent} if it is diagonal in this basis; the set of all such states is denoted by
$\mc{I}(\mc H) := \{ \sum_i p_i |i\rangle\langle i| : p_i \ge 0, \sum_i p_i = 1 \}$. Otherwise, it is said to be \emph{coherent}.
For the $n$-fold tensor product $\mathcal{H}^{\otimes n}$, we denote by $\mathcal{I}_n$ the set of incoherent states on $\mathcal{H}^{\otimes n}$, i.e., density operators that are diagonal in the $n$-fold computational basis. Similarly, by a slight abuse of notation, we collectively write the sequence of incoherent state sets as $\mathcal{I} = \{\mathcal{I}_n\}_{n\in\mathbb{N}}$.
For any integer $K \ge 2$, the \emph{maximally coherent state} is defined as $\Psi_K := |\Psi_K\rangle\langle\Psi_K|$ with
$|\Psi_K\rangle := \frac{1}{\sqrt{K}} \sum_{i=1}^K |i\rangle$.

A quantum channel $\Lambda$ is a completely positive and trace-preserving (CPTP) linear map acting on quantum states. A quantum measurement is mathematically represented by a set of positive semi-definite operators $\mc{M}=\{M_x\}_x$ satisfying the normalization condition $\sum_xM_x=\mathds{1}$. The set of all measurements is denoted by $\mathds{ALL}$. When the measurement operators are restricted to being separable and PPT, respectively, the corresponding measurements are referred to as separable measurements and PPT measurements, denoted by $\mathds{SEP}$ and $\mathds{PPT}$, respectively. These are instances of the allowed measurement set $\mathds{M}$, i.e., $\mathds{M}\in\{\mathds{ALL},\ \mathds{SEP},\ \mathds{PPT}\}$.

In many practical scenarios, the sequences of sets under consideration are not arbitrary but possess a structure that is compatible with tensor products.
This property, known as {stability} (or closedness) under tensor products, is formalized as follows:
A sequence of sets $\{\ms{A}_n\}_{n\in\mathbb{N}}$ with $\ms{A}_n\subseteq\ms{P}(\mc{H}^{\ox n})$ is called \emph{stable under tensor product} if $\ms{A}_n \otimes \ms{A}_m \subseteq \ms{A}_{n+m}~\text{for all }n,m\in\mathbb{N},$ where $\ms{A}_n\otimes\ms{A}_m:=\{P\otimes Q: P\in\ms{A}_n,\;Q\in\ms{A}_m\}$. It is evident that the sequences of separable states, PPT states and incoherent states satisfy stability. Regarding the sequence of sets of initial entangled states and initial coherent states in the black-box scenario, further details are given in Section~\ref{sec: Concrete Physical Illustrations of Noise}.

The floor function, denoted by $\lfloor\cdot\rfloor$, gives the greatest integer less than or equal to a real number, while the ceiling function, denoted by $\lceil\cdot\rceil$, gives the smallest integer greater than or equal to a real number. Given any set $\ms{A}$, the convex hull generated by $\ms{A}$ is denoted by $\mathrm{conv}(\ms{A})$, and the convex cone generated by $\ms{A}$ is denoted by $\mathrm{cone}(\ms{A})$. Throughout this paper, all logarithms are taken to the base $2$ unless otherwise specified.

\subsection{Quantum Divergences}
A functional $\DD: \density(\mc{H}) \times \PSD(\mc{H}) \to \RR$ is called a \emph{quantum divergence} if it satisfies the data-processing inequality: for any quantum channel $\Lambda$ and any $(\rho,\sigma) \in \density(\mc{H}) \times \PSD(\mc{H})$, it holds that $\DD(\Lambda(\rho)\|\Lambda(\sigma)) \leq \DD(\rho\|\sigma)$. In the following, we introduce several quantum divergences that will be used throughout this work. We also define quantum divergences between two sets of quantum states and their regularized versions.
\begin{definition}[Umegaki relative entropy and $\mathds{M}$-relative entropy]
    For any $(\rho,\sigma) \in \density(\mc{H}) \times \PSD(\mc{H})$,  the Umegaki relative entropy~\cite{umegaki} is defined by
    \begin{align}
        D(\rho\|\sigma):= \tr[\rho(\log \rho - \log \sigma)],\nonumber
    \end{align}
    if $\supp(\rho)\subseteq \supp(\sigma)$, and $+\infty$ otherwise.
    The relative entropy under the restricted measurement $\mathds{M}$---also called the $\mathds{M}$-relative entropy~\cite{piani2009relative}---is defined as 
    \begin{align}
        D^{\mathds{M}}(\rho\|\sigma):=\sup_{\{M_x\}_x\in\mathds{M}}\sum_x\tr(M_x\rho)\log\frac{\tr(M_x\rho)}{\tr(M_x\sigma)}.\nonumber
    \end{align}
\end{definition}
    
\begin{definition}[Petz \Renyi divergence and Sandwiched \Renyi divergence]
    For any $(\rho,\sigma) \in \density(\mc{H}) \times \PSD(\mc{H})$ and $\alpha \in (0,1) \cup (1,\infty)$, the Petz \Renyi divergence~\cite{petz1986quasi} is defined by
    \begin{align}
        D_{\Petz,\alpha}(\rho\|\sigma) := \frac{1}{\alpha-1}\log\tr\left[\rho^\alpha\sigma^{1-\alpha}\right],\nonumber
    \end{align}
     and the sandwiched \Renyi divergence~\cite{muller2013quantum,wilde2014strong} is defined by
    \begin{align}
        D_{\Sand,\alpha}(\rho\|\sigma) := \frac{1}{\alpha-1}\log\tr\left[\sigma^{\frac{1-\alpha}{2\alpha}}\rho\sigma^{\frac{1-\alpha}{2\alpha}}\right]^\alpha,\nonumber
    \end{align}
   if either $\alpha < 1$ and $\rho \not\perp \sigma$ or $\alpha > 1$ and $\supp(\rho) \subseteq \supp(\sigma)$. Otherwise, the divergences are set to $+\infty$.
\end{definition}

For later use, we also consider the order-zero Petz R\'enyi
divergence, defined as the limit of the Petz R\'enyi divergence as
$\alpha$ tends to zero. For any
$\rho\in\mathcal{D}(\mathcal{H})$ and
$\sigma\in\mathcal{P}(\mathcal{H})$,
\begin{equation}
D_{\Petz,0}(\rho\Vert\sigma)
:=
\lim_{\alpha\to0}
D_{\Petz,\alpha}(\rho\Vert\sigma)
=
-\log\tr(\Pi_\rho\sigma),\nonumber
\label{eq:order-zero-Petz}
\end{equation}
where $\Pi_\rho$ denotes the projection onto the support of $\rho$.

\begin{definition}[the fidelity]
    For state $\rho$ and $\sigma$ on Hilbert space $\mathcal{H}$, the fidelity between them is defined as $$F(\rho,\sigma):=\|\sqrt{\rho}\sqrt{\sigma}\|_1^2=\tr\left(\sqrt{\sqrt{\sigma}\rho\sqrt{\sigma}}\right)^2.$$
In particular, when one of the states is a pure state $\psi$, the fidelity $F(\rho,\psi)=\tr(\rho\psi).$
\end{definition}

\begin{definition}[Quantum divergence between two sets of states and its regularized version]\label{def: divergence between two sets}
    Let $\DD$ be a quantum divergence between two quantum states. Then for any sets $\sA,\sB\subseteq \density(\cH)$, the quantum divergence between these two sets is defined by
    \begin{align}
        \DD(\sA\|\sB):= \inf_{{\rho \in \sA, \sigma \in \sB}} \DD(\rho\|\sigma). \nonumber
    \end{align}
   Let $\sA = \{\sA_n\}_{n\in \NN}$ and $\sB = \{\sB_n\}_{n\in \NN}$ be two sequences of sets of quantum states, where each $\sA_n, \sB_n \subseteq \density(\cH^{\ox n})$. The regularized divergence between these sequences is defined by 
\begin{align}
    \DD^{\reg}(\sA \| \sB) := \lim_{n \to \infty} \frac{1}{n} \DD(\sA_{n} \| \sB_{n}),\nonumber
\end{align}
if the limit exists.
\end{definition}

\subsection{Composite Correlated Hypothesis Testing}
We set up a hypothesis testing problem with the composite correlated null hypothesis comprising all free states $\ms{F}_n$, against the composite alternative hypothesis which is the set of initial states $\ms{R}_n$. The decision between the hypotheses is made via the quantum measurement $\mathds{M}$. Two types of error can occur. The type-I error is the probability that we incorrectly accept the resource state $\ms{R}_n$ while it is actually the free state $\ms{F}_n$, and the type-II error corresponds to the converse situation. 
Under the constraint that the type‑I error is at most $\epsilon_n$, the optimal type‑II error is defined as
\begin{align}
    &\beta_n(\ms{F}_n,\ms{R}_n,\mathds{M},\epsilon_n)\nonumber\\
    :=&\min_{\substack{\{M_n,\mathds{1}-M_n\}\in\mathds{M}\\ \max_{\sigma_n\in\ms{F}_n}\tr((\mathds{1}-M_n)\sigma_n)\leq \epsilon_n}}\max_{\rho_n\in\ms{R}_n}\tr M_n\rho_n. \nonumber
\end{align}
For two sequences of state sets
$\ms{F}=\{\ms{F}_n\}_{n\in\NN}$ and
$\ms{R}=\{\ms{R}_n\}_{n\in\NN}$ and a sequence of allowed
measurements $\mathds{M}$, we define the lower and upper
hypothesis testing error exponents at rate $r>0$ as
\begin{align}
\underline{E}_{HT}(\ms{F},\ms{R},\mathds{M},r)\nonumber
:=&
\liminf_{n\to\infty}
-\frac{1}{n}
\log
\beta_n(
\ms{F}_n,\ms{R}_n,\mathds{M},2^{-nr}),\\
\overline{E}_{HT}(\ms{F},\ms{R},\mathds{M},r)\nonumber
:=&
\limsup_{n\to\infty}
-\frac{1}{n}
\log
\beta_n(
\ms{F}_n,\ms{R}_n,\mathds{M},2^{-nr}).
\end{align}
Whenever the two quantities coincide, we denote their common value by
\begin{equation}
E_{HT}(\ms{F},\ms{R},\mathds{M},r).\nonumber
\end{equation}

Since both
$r\mapsto\underline{E}_{HT}(\ms{F},\ms{R},\mathds{M},r)$
and
$r\mapsto\overline{E}_{HT}(\ms{F},\ms{R},\mathds{M},r)$
are non-increasing, their left limits exist as extended real numbers.
We denote
\begin{align}
\underline{E}_{HT}(\ms{F},\ms{R},\mathds{M},r^-)
&:=
\lim_{\delta\to0}
\underline{E}_{HT}(
\ms{F},\ms{R},\mathds{M},r-\delta),\nonumber
\\
\overline{E}_{HT}(\ms{F},\ms{R},\mathds{M},r^-)
&:=
\lim_{\delta\to0}
\overline{E}_{HT}(
\ms{F},\ms{R},\mathds{M},r-\delta).\nonumber
\end{align}

For the case where all measurements $\mathds{ALL}$ are permitted, Fang and Hayashi~\cite{fang2025errorexponentsquantumstate} studied the error exponent of composite correlated hypothesis testing. For ease of reading, we reformulate it in our notation as follows. 
\begin{theorem}[Fang and Hayashi~\cite{fang2025errorexponentsquantumstate}]
\label{them:Fang and Hayashi}
    Let $\mc{H}$ be a finite-dimensional Hilbert space. Let $\ms{A}=\{\ms{A}_n\}_{n\in\NN}$ and $\ms{B}=\{\ms{B}_n\}_{n\in\NN}$ be two stable sequences of convex compact sets of quantum states with each $\ms{A}_n,\ms{B}_n\subseteq\density(\mathcal{H}^{\ox n})$. Let $r$ be a real number with $0<r<D^{\infty}(\ms{A}\|\ms{B})$. 
    Assume that $D_{\mathrm{P},\alpha}(\ms{A}_1 \| \ms{B}_1) < \infty$ for $\alpha \in (0,1)$. Then 
    \begin{align}
        E_{HT}(\ms{B},\ms{A},\mathds{ALL},r)=H_r^{\infty}(\ms{A}\|\ms{B}),
    \end{align}
    where the right-hand side follows Definition~\ref{def: divergence between two sets} and is induced by the quantum Hoeffding divergence
    \begin{align}
        H_{n,r}(\rho_n\|\sigma_n):=\sup_{\alpha\in(0,1)}\frac{\alpha-1}{\alpha}(nr-D_{\Petz,\alpha}(\rho_n\|\sigma_n)).\nonumber
    \end{align}
\end{theorem}

In fact, Fang and Hayashi~\cite{fang2025errorexponentsquantumstate} explored two approaches in generalizing the quantum Hoeffding divergence. The first approach as shown above is 
\begin{align*}
    H_{n,r}(\ms{A}_n\|\ms{B}_n):&=\inf_{{\rho_n\in\ms{A}_n\sigma_n\in\ms{B}_n}}H_{n,r}(\rho_n\|\sigma_n),\\
    H_r^\infty(\ms{A} \|\ms{B}) :&= \lim_{n \to \infty} \frac{1}{n} H_{n,r}(\ms{A}_n \|\ms{B}_n).
\end{align*}
The second approach is
\begin{align*}
    \mathfrak{H}_{n,r}(\ms{A}_n \|\ms{B}_n) :&= \sup_{\alpha \in (0,1)} \frac{\alpha - 1}{\alpha} \left( nr - D_{\mathrm{P},\alpha}(\ms{A}_n \|\ms{B}_n) \right),\\
    \mathfrak{H}_{r}^{\infty}(\ms{A} \|\ms{B}) :&= \sup_{\alpha \in (0,1)} \frac{\alpha - 1}{\alpha} \left( r - D_{\mathrm{P},\alpha}^{\infty}(\ms{A} \|\ms{B}) \right).
\end{align*}
They have shown that $H_{n,r}(\ms{A}_n\|\ms{B}_n)=\mathfrak{H}_{n,r}(\ms{A}_n \|\ms{B}_n)$. However, in the regularized sense, they only obtain
\begin{align}
    H_r^\infty(\ms{A} \|\ms{B})&\ge \mathfrak{H}_{r}^{\infty}(\ms{A} \|\ms{B}).\label{eq:regulaized}
\end{align}
They asked whether equality holds in Equation~\eqref{eq:regulaized}. We give a negative answer to this question for sequences $\ms{A}$ and $\ms{B}$ satisfying the same conditions as in Theorem~\ref{them:Fang and Hayashi}.

\begin{theorem}
\label{them: regularized inequivalent}
    There exist two stable sequences of convex compact sets of quantum states 
\[
\ms{A} = \{\ms{A}_n\}_{n\in\NN},\quad \ms{B} = \{\ms{B}_n\}_{n\in\NN},
\] 
with each $\ms{A}_n,\ms{B}_n\subseteq\density(\mathcal{H}^{\otimes n})$, and with 
$D_{\mathrm{P},\alpha}(\ms{A}_1 \| \ms{B}_1) < \infty$ for all $\alpha\in(0,1)$, 
and a real number $r$ satisfying $0<r<D^{\infty}(\ms{A}\|\ms{B})$, such that  
\[
H_r^\infty(\ms{A}\|\ms{B}) \;>\; \mathfrak{H}_r^{\infty}(\ms{A}\|\ms{B}).\]
\end{theorem}

\begin{proof}
\emph{Construct the sequences $\ms{A}$ and $\ms{B}$.} Let
\[
X\simeq\mathbb C^3,\quad E\simeq\mathbb C^2,\quad \mc H:=X\otimes E.
\]
Define commuting states on $X$ by
\[
\rho=\operatorname{diag}\Bigl(\frac34,\frac14,0\Bigr),\qquad
\sigma=\operatorname{diag}\Bigl(\frac14,\frac14,\frac12\Bigr),
\]
and let $\tau_E:=\1_E/2$.

For each $n\ge1$, let $W_n:\mc H^{\otimes n}\to X^{\otimes n}\otimes E^{\otimes n}$ be the canonical shuffle unitary
\[
W_n\bigl(\!(x_1\otimes e_1)\otimes\cdots\otimes(x_n\otimes e_n)\!\bigr)
\!=\!(x_1\otimes\cdots\otimes x_n)\otimes(e_1\otimes\cdots\otimes e_n).
\]
Set $S_n:=|0\rangle_{E_1}\otimes E^{\otimes(n-1)}\subseteq E^{\otimes n}$ and
\[
\ms D(S_n):=\{\omega\in\ms D(E^{\otimes n}):\operatorname{supp}\omega\subseteq S_n\}.
\]
Now define
\begin{align*}
\ms A_n &:=\Bigl\{W_n^\dagger\bigl(\rho^{\otimes n}\otimes\omega\bigr)W_n:\omega\in\ms D(S_n)\Bigr\},\\
\ms B_n &:=\Bigl\{W_n^\dagger\bigl(\sigma^{\otimes n}\otimes\tau_E^{\otimes n}\bigr)W_n\Bigr\}
        =\bigl\{(\sigma\otimes\tau_E)^{\otimes n}\bigr\}.
\end{align*}

\emph{Verification of conditions.}
The sequence $\ms{B}$ is a singleton and thus trivially a stable sequence of convex compact sets.
The set $\ms A_n$ is convex because $\ms D(S_n)$ is convex and the map $\omega\mapsto W_n^\dagger(\rho^{\otimes n}\otimes\omega)W_n$ is affine.
It is compact as the continuous image of the compact set $\ms D(S_n)$.

For stability under tensor products, take arbitrary
\begin{align*}
    a_m&=W_m^\dagger(\rho^{\otimes m}\otimes\omega_m)W_m\in\ms A_m,\\
    a_n&=W_n^\dagger(\rho^{\otimes n}\otimes\omega_n)W_n\in\ms A_n,
\end{align*}
with $\operatorname{supp}\omega_m\subseteq S_m$ and $\operatorname{supp}\omega_n\subseteq S_n$.
Under the canonical shuffle on $m+n$ copies,
\[
W_{m+n}(a_m\otimes a_n)W_{m+n}^\dagger=\rho^{\otimes(m+n)}\otimes(\omega_m\otimes\omega_n).
\]
Clearly $\operatorname{supp}(\omega_m\otimes\omega_n)\subseteq S_{m+n}$, we have $a_m\otimes a_n\in\ms A_{m+n}$.
Hence $\ms{A}$ is stable under tensor products.

Because $\ms B_1$ consists solely of the full-rank state $\sigma\otimes \1_E/2$, we have
$D_{\mathrm{P},\alpha}(\ms A_1\|\ms B_1)<\infty$ for every $\alpha\in(0,1)$.

\emph{Preliminary calculation.}
Compute the Petz R\'enyi divergence
\[
d_\alpha:=D_{\mathrm P,\alpha}(\rho\|\sigma)=\frac{1}{\alpha-1}\log\!\Bigl(\frac{1+3^\alpha}{4}\Bigr).
\]
Endpoint limits:
\[
\lim_{\alpha\to0}d_\alpha=1,\qquad
\lim_{\alpha\to1}d_\alpha=\frac34\log 3.
\]
For any $\omega\in\ms D(S_n)$, additivity yields
\[
D_{\mathrm P,\alpha}\bigl(\rho^{\otimes n}\otimes\omega\,\big\|\,\sigma^{\otimes n}\otimes\tau_E^{\otimes n}\bigr)
= n d_\alpha + D_{\mathrm P,\alpha}(\omega\|\tau_E^{\otimes n}).
\]
When $\omega$ is the maximally mixed state on $S_n$, the second term attains its minimum $1$.
Thus
\[
D_{\mathrm P,\alpha}(\ms A_n\|\ms B_n)=n d_\alpha+1,~D_{\mathrm P,\alpha}^\infty(\ms A\|\ms B)=d_\alpha.
\]

Similarly, using the Umegaki relative entropy,
\[
D(\ms A_n\|\ms B_n)=n D(\rho\|\sigma)+1,~D^\infty(\ms A\|\ms B)=\frac34\log 3\approx1.19.
\]
Take $r=1$, which satisfies $0<1<D^\infty(\ms A\|\ms B)$.

\emph{First variant.}
\[
H_{n,1}(\ms A_n\|\ms B_n)\!\!=\!\mathfrak H_{n,1}(\ms A_n\|\ms B_n)\!
=\!\!\sup_{\alpha\in(0,1)}\!\!\frac{\alpha-1}{\alpha}\bigl[n-(n d_\alpha+1)\bigr].
\]
As $\alpha$ tends to $0$, $d_\alpha$ tends to $1$; consequently $H_{n,1}(\ms A_n\|\ms B_n)=+\infty$ for every $n$.
Therefore $H_1^\infty(\ms A\|\ms B)=+\infty$.

\emph{Second variant.}
\[
\mathfrak H_1^\infty(\ms A\|\ms B)=\sup_{\alpha\in(0,1)}\frac{\alpha-1}{\alpha}(1-d_\alpha)=1-\frac12\log 3.
\]
\end{proof}

The mechanism behind the inequality derived from this example is a nonuniform boundary effect near $\alpha=0$. For each finite block length, $D_{\mathrm{P},\alpha}(\ms A_n\|\ms B_n)=n d_\alpha+1$. This additive correction $+1$ vanishes after regularization, giving $D_{\mathrm{P},\alpha}^\infty=d_\alpha$. In the first variant, the supremum over $\alpha$ is taken before regularization, multiplying the constant by the divergent factor $(1-\alpha)/\alpha$, so $H_{n,1}=+\infty$ for every $n$ and hence $H_1^\infty=+\infty$. In the second variant, regularization removes the constant first, leaving the finite product $(\alpha-1)(1-d_\alpha)/\alpha$ and yielding $\mathfrak H_1^\infty=1-\frac12\log 3$.

In fact, from the perspective of hypothesis testing we can construct a measurement to gain clearer insight. Let $P_\rho
=
|0\rangle\langle0|
+
|1\rangle\langle1|$
be the support projector of \(\rho\), and let \(\Pi_{S_n}\) be the projector onto \(S_n\). Consider the test
\[
\mathds{1}-M_n
=
W_n^\dagger
\left(
P_\rho^{\otimes n}\otimes\Pi_{S_n}
\right)
W_n.
\]
Every state in \(\ms A_n\) is supported on \(\mathds{1}-M_n\), so the worst-case type-II error is exactly zero. For the unique state in \(\ms B_n\),
\[
\begin{aligned}
\operatorname{Tr}
\left[
(\sigma^{\otimes n}\otimes\tau_E^{\otimes n})
(P_\rho^{\otimes n}\otimes\Pi_{S_n})
\right]
&=
\left(
\operatorname{Tr}\sigma P_\rho
\right)^n
\operatorname{Tr}
\left(
\tau_E^{\otimes n}\Pi_{S_n}
\right)\\
&=
\left(\frac12\right)^n
\frac{2^{n-1}}{2^n}\\
&=
2^{-(n+1)}\\
&\leq
2^{-n}.
\end{aligned}
\]
Thus the type-I error constraint at rate \(r=1\) is satisfied while the type-II error is zero. This directly shows that the error exponent is infinite. Therefore, in this paper we still adopt the first generalization approach of the quantum Hoeffding divergence.

For the exponential analysis of universal entanglement distillation, including the reliability function and the strong converse exponent, we formally introduce two quantities. The two quantities initially emerged as the error exponent~\cite{PhysRevA.76.062301,nagaoka2006conversetheoremquantumhoeffding,audenaert2008asymptotic} and the strong converse exponent~\cite{mosonyi2015quantum} in independent and identically distributed hypothesis testing. They were later refined into the divergence and the anti-divergence in~\cite{fang2025errorexponentsquantumstate}. In this paper, we endow them with new operational interpretations in universal entanglement distillation. For convenience, we directly present their definitions in the $n$-copy form along with their fundamental properties.

\begin{definition}
    Let $r > 0$ be a real number, $n\in\mathbb{N}$, and $\rho_n,\sigma_n\in\ms{D}(\mc{H}^{\ox n})$. The quantum Hoeffding divergence is defined as
    \begin{align}
        H_{n,r}(\rho_n\|\sigma_n):=\sup_{\alpha\in(0,1)}\frac{\alpha-1}{\alpha}(nr-D_{\Petz,\alpha}(\rho_n\|\sigma_n)),
    \end{align}
    and the quantum Hoeffding anti-divergence is defined as
    \begin{align}
        \widetilde{H}_{n,r}(\rho_n\|\sigma_n):=\sup_{\alpha>1}\frac{\alpha-1}{\alpha}(nr-D_{\Sand,\alpha}(\rho_n\|\sigma_n)).
    \end{align}
\end{definition}

The properties presented below are inherited from those of the Petz and sandwiched Rényi relative entropies, and can be obtained with the aid of simple analytical techniques.
\begin{proposition}
\label{prop: Hoeffding}
    For the quantum Hoeffding divergence $H_{n,r}(\rho_n\|\sigma_n)$, the following properties hold:
\begin{enumerate}
    \item \emph{Data-Processing Inequality.} For any quantum channel~$\Lambda$, it holds that $$H_{n,r}(\Lambda(\rho_n)\|\Lambda(\sigma_n)) \leq H_{n,r}(\rho_n\|\sigma_n).$$
    \item \emph{Joint Convexity.} For any probability distribution $\{p_i\}_{i=1}^N$ and states $\{\rho_n^i,\sigma_n^i\}_{i=1}^N$, the following inequality holds:
    $$H_{n,r}\left(\sum_{i=1}^Np_i\rho_n^i\|\sum_{i=1}^Np_i\sigma_n^i\right)\leq\sum_{i=1}^Np_iH_{n,r}(\rho_n^i\|\sigma_n^i)$$
    \item \emph{$\text{Continuity,~Convexity,~and Non‑Increasing~in}~r$.} 
    \begin{itemize}
        \item When $r<\frac{1}{n}D_{\Petz,0}(\rho_n\|\sigma_n)$, $H_{n,r}(\rho_n\|\sigma_n)=+\infty$; When $r>\frac{1}{n}D_{\Petz,0}(\rho_n\|\sigma_n)$, it is finite and continuous with respect to $r$.
        \item For any $r_1,r_2>0$ and $\lambda\in[0,1]$, it holds that
    $$H_{n,\lambda r_1+(1-\lambda)r_2}\leq \lambda H_{n,r_1}+(1-\lambda)H_{n,r_2}.$$
        \item For any $0<r_1<r_2$, it holds that
        $$H_{n,r_1}(\rho_n\|\sigma_n) \ge H_{n,r_2}(\rho_n\|\sigma_n).$$
    \end{itemize}
    \item \emph{Subadditivity.} For any $\rho_m,\sigma_m\in\ms{D}(\mc{H}^{\ox m})$ and $\rho_n,\sigma_n\in\ms{D}(\mc{H}^{\ox n})$, it holds that
    $$H_{\!(\!m+n\!),r}(\!\rho_m \smallotimes \rho_n \| \sigma_m \smallotimes \sigma_n\!) \!\leq\! H_{\!m,r}\!(\!\rho_m \| \sigma_m\!) + H_{\!n,r\!}(\!\rho_n \| \sigma_n\!).$$
    Furthermore, for any two stable sequences of quantum state sets $\ms{A}$ and $\ms{B}$, it holds that
    $$H_{\!(\!m+n\!),r}(\!\ms{A}_{m+n}\! \|\! \ms{B}_{m+n}\!)\! \leq\! H_{m\!,r}(\!\ms{A}_{m} \!\| \!\ms{B}_{m}\!) + H_{n\!,r}(\!\ms{A}_{n} \!\| \!\ms{B}_{n}\!).$$
    \item \emph{Positivity and Faithfulness.} The quantum Hoeffding divergence satisfies $ H_{n,r}(\rho_n\|\sigma_n) \ge 0,$ and equality holds if and only if $nr \geq D(\rho_n\|\sigma_n).$
    \item \emph{Isometric Invariance.} For all isometries $V$,
    $$H_{n,r}(\rho_n\|\sigma_n)=H_{n,r}(V\rho_nV^\dagger\|V\sigma_nV^\dagger).$$    
\end{enumerate}
\end{proposition}

\begin{proposition}
\label{prop: Hoeffding anti}
    For the quantum Hoeffding anti-divergence $\widetilde{H}_{n,r}(\rho_n\|\sigma_n)$, the following properties hold:
\begin{enumerate}
    \item \emph{Anti-Data Processing Inequality.} For any quantum channel~$\Lambda$, it holds that $$\widetilde{H}_{n,r}(\Lambda(\rho_n)\|\Lambda(\sigma_n)) \geq \widetilde{H}_{n,r}(\rho_n\|\sigma_n).$$
    \item \emph{Mixture Lower Bound.} For any probability distribution $\{p_i\}_{i=1}^N$ and states $\{\rho_n^i,\sigma_n^i\}_{i=1}^N$, the following inequality holds:
    \begin{align*}
        \widetilde{H}_{n,r}&\left(\sum_{i=1}^Np_i\rho_n^i\|\sum_{i=1}^Np_i\sigma_n^i\right)\\
        \geq&\sup_{\alpha>1}\frac{\alpha-1}{\alpha}(nr-\max_iD_{\Sand,\alpha}(\rho_n^i\|\sigma_n^i)).
    \end{align*}
    \item \emph{$\text{Continuity,~Convexity,~and Non‑Decreasing~in}~r$.} 
    \begin{itemize}
        \item The quantum Hoeffding anti-divergence is continuous with respect to $r$.
        \item For any $r_1,r_2>0$ and $\lambda\in[0,1]$, it holds that
    $$\widetilde{H}_{n,\lambda r_1+(1-\lambda)r_2}\leq \lambda \widetilde{H}_{n,r_1}+(1-\lambda)\widetilde{H}_{n,r_2}.$$
        \item For any $0<r_1<r_2$, it holds that
        $$\widetilde{H}_{n,r_1}(\rho_n\|\sigma_n) \le \widetilde{H}_{n,r_2}(\rho_n\|\sigma_n).$$
    \end{itemize}
    \item \emph{Subadditivity.} For any $\rho_m,\sigma_m\in\ms{D}(\mc{H}^{\ox m})$ and $\rho_n,\sigma_n\in\ms{D}(\mc{H}^{\ox n})$, it holds that
    $$\widetilde{H}_{\!(\!m+n\!),r}(\!\rho_m \smallotimes \rho_n \| \sigma_m \smallotimes \sigma_n\!) \!\leq\! \widetilde{H}_{\!m,r}\!(\!\rho_m \| \sigma_m\!) + \widetilde{H}_{\!n,r\!}(\!\rho_n \| \sigma_n\!).$$
    \item \emph{Positivity and Faithfulness.} When $\operatorname{supp}(\rho_n)\subseteq\operatorname{supp}(\sigma_n)$, $\widetilde{H}_{n,r}(\rho_n\|\sigma_n)\ge 0$, with equality if and only if $nr\le D(\rho_n\|\sigma_n)$. When $\operatorname{supp}(\rho_n)\nsubseteq\operatorname{supp}(\sigma_n)$, $\widetilde{H}_{n,r}(\rho_n\|\sigma_n)=-\infty$.
    \item \emph{Isometric Invariance.} For all isometries $V$,
    $$\widetilde{H}_{n,r}(\rho_n\|\sigma_n)=\widetilde{H}_{n,r}(V\rho_nV^\dagger\|V\sigma_nV^\dagger).$$    
\end{enumerate}
\end{proposition}

For any two convex, compact, and stable sequences of states $\ms{A}$ and $\ms{B}$, the subadditivity of the quantum Hoeffding divergence guarantees the existence of the limits of its regularized versions $H_r^{\infty}(\ms{A}\|\ms{B})$ .
Since the quantum Hoeffding anti-divergence obeys the anti-data processing inequality, the supremum, rather than the infimum, is taken when extending it to sets, i.e., 
$$\widetilde{H}_{n,r}(\ms{A}_n\|\ms{B}_n):=\sup_{{\rho_n\in\ms{A}_n\sigma_n\in\ms{B}_n}}\widetilde H_{n,r}(\rho_n\|\sigma_n).$$
Fang and Hayashi~\cite{fang2025errorexponentsquantumstate} also show that the regularized quantum Hoeffding anti-divergence $\widetilde{H}_r^{\infty}(\ms{A}\|\ms{B})$ exists.

\section{Universal Entanglement Distillation}
Entanglement distillation involves two separated parties, conventionally named Alice and Bob. In previous studies~\cite{onewayLOCC,Reversible2008,Reversible2010,Dually,lami2024asymptotic}, two parties share multiple copies of a bipartite entangled state $\rho_{AB}$ and aim to convert them into a smaller number of maximally entangled states. However, the initial state is not fully accessible in practical implementations due to unidentified noise and adversarial influence. Therefore, considering the initial resource states as a black-box set $\ms{R}_n$, it becomes necessary to investigate universal distillation protocols in ``state-agnostic" scenarios. 
For this transformation, not all quantum channels are available. Those that are admissible are called \emph{free operations}, denoted by $\mathcal{O}$. 

In this paper, we investigate universal entanglement distillation under four types of free operations: non‑entangling, PPT‑preserving, dually non‑entangling, and dually PPT‑preserving operations. Specifically, non‑entangling operations and PPT‑preserving operations are those that preserve separability and PPT property of states, respectively. The class of non‑entangling operations is denoted by $\mc{F}_{\NE}:=\{\Lambda\in \mathrm{CPTP}:\Lambda(\sigma)\in \mc{S},~\forall \sigma\in\mc{S}\}$, and the class of  PPT‑preserving operations is denoted by $\mathcal{F}_{\mathrm{PPT}}:=\{\Lambda\in \mathrm{CPTP}:\Lambda(\sigma)\in \mc{PPT},~\forall \sigma\in\mc{PPT}\}$. Duality means that not only does the channel preserve the property of states in the Schrödinger picture, but its adjoint also preserves the property of measurement operators in the Heisenberg picture.
The collection of dually non-entangling operations is denoted as $\mathcal{F}_{\mathrm{DNE}}:=\{\Lambda\in\mathrm{CPTP}:\Lambda(\sigma)\in\mc{S},~\forall \sigma\in\mc{S},\text{and}~ \Lambda^\dagger(M)\in\mathrm{cone}(\mc{S}),~\forall~M\in\mathrm{cone}(\mc{S})\},$~and the collection of dually PPT-preserving operations is denoted as $\mathcal{F}_{\mathrm{DPPT}}:=\{\Lambda\in\mathrm{CPTP}:\Lambda(\sigma)\in\mc{PPT},~\forall \sigma\in\mc{PPT},~\text{and}~ \Lambda^\dagger(M)\in\mathrm{cone}(\mc{PPT}),~\forall~M\in\mathrm{cone}(\mc{PPT})\}.$

\subsection{Finite-Blocklength Analysis}
Let $\ms{R}_n$ be a convex compact set of quantum states on Hilbert space $\mathcal{H}_{AB}^{\ox n}$. When the distillation rate is $r$, which means that the yield is $\lfloor rn\rfloor$ copies of the two-qubit maximally entangled states $\Phi_2^{\otimes \lfloor rn \rfloor}$, and the free operations are restricted to $\mc{O}$, the error of universal entanglement distillation is defined as
\begin{equation}
    \err(\ms{R}_n\!\!\xrightarrow{\mathcal{O}}\! \Phi_2^{\otimes\lfloor \!rn\! \rfloor})\!\!:=\!\!\min_{\Lambda_n\in\mathcal{O}} \max_{\rho_n\in\ms{R}_n}\!\!(1\!-\!F(\Lambda_n(\rho_n),\!\Phi_2^{\otimes \lfloor \! rn \!\rfloor})\!).\nonumber
\end{equation}
We measure the distillation error by the infidelity $1-F$, take the worst case over all input states in black box $\ms{R}
_n$, and then minimize over all allowed operations to achieve the best possible performance.

In the finite-blocklength regime, by shifting our perspective
from the only integer-valued and thus discrete yield to the
continuously varying error, we establish an exact equivalence between the error of entanglement distillation in the black-box setting and the optimal type-II error in the composite correlated hypothesis testing.

\begin{theorem}[Finite-blocklength under $\mathcal{F}_{\NE}$ and $\mathcal{F}_{\mathrm{PPT}}$]
\label{thm: error in finite blocklenth}
Let $\ms{R}_n$ be a convex compact set of quantum states on $\mathcal{H}_{AB}^{\otimes n}$, and let $r>0$ be a real number. Then 
the errors of universal entanglement distillation under non‑entangling and PPT‑preserving operations equal the optimal type‑II errors in hypothesis testing with null hypothesis being the set of separable states $\mc{S}_n$ and the set of PPT states $\mc{PPT}_n$, respectively, and alternative hypothesis $\ms{R}_n$, using all measurements $\mathds{ALL}$, with the type‑I error bounded by $2^{-\lfloor rn \rfloor}$,
    \begin{align}
        \err(\ms{R}_n\xrightarrow{\mathcal{F}_{\NE}} \Phi_2^{\otimes \lfloor rn \rfloor})=&\beta_n(\mc{S}_n,\ms{R}_n,\mathds{ALL},2^{-\lfloor rn \rfloor}),\label{eq: F_NE}\\
        \err(\ms{R}_n\xrightarrow{\mathcal{F}_{\mathrm{PPT}}} \Phi_2^{\otimes \lfloor rn \rfloor})=&\beta_n(\mc{PPT}_n,\ms{R}_n,\mathds{ALL},2^{-\!\lfloor\! rn \!\rfloor}).\label{eq:F_PPT}
    \end{align}
\end{theorem}

\begin{proof}
To begin with, we prove the equality for the case of non‑entangling operations, i.e., Equation~\eqref{eq: F_NE}. The core idea of the proof is to take an optimal strategy for one task and use it to construct a feasible strategy for the other.\\
\textit{1). Proof of the upper bound:}
    Construct a distillation protocol,
    \begin{equation}
    \label{eq: construction of distillation protocol in black-box}
        \Lambda_n(X)\!=\!\tr(\!(\mathds{1}-M_n)X)\Phi_2^{\ox \lfloor\! rn\!\rfloor}\!+\!\tr(M_nX)\frac{\mathds{1}-\Phi_{2}^{\otimes \lfloor \!rn\!\rfloor}}{(2^{\lfloor\! rn\!\rfloor})^2\!-\!1},
    \end{equation}
    where $\{\mathds{1}-M_n,~M_n\}$ is the measurement that achieves the optimal type-II error $\beta_n(\mc{S}_n,\ms{R}_n,\mathds{ALL},2^{-\lfloor rn \rfloor})$.
    
    First, prove that $\Lambda_n$ constructed above belongs to non-entangling channels. For any $\sigma_n\in\mc{S}_n$, 
    \begin{align}
        \Lambda_n(\sigma_n)\!=\!\operatorname{Tr}(\!(\mathds{1}-M_n)\sigma_n)\Phi_{2}^{\otimes \lfloor \!rn\!\rfloor}\!\!+\!\!\operatorname{Tr}(M_n\sigma_n)\frac{\mathds{1}-\Phi_{2}^{\otimes \lfloor\! rn\!\rfloor}}{(2^{\lfloor\! rn\!\rfloor})^2-1}.\nonumber
    \end{align}
    It is known that the isotropic state
    \begin{align} 
    p\Phi_{2}^{\otimes m}+(1-p)\frac{\mathds{1}-\Phi_{2}^{\otimes m}}{(2^m)^2-1}\nonumber
    \end{align} 
    is separable for $p\in[0,2^{-m}]$~\cite[p.~319]{watrous2018theory}. Due to $\tr ( (\mathds{1}-M_n)\sigma_n)\leq2^{-\lfloor rn\rfloor}$, we obtain $\Lambda_n\in \mathcal{F}_{\NE}$.

    Then, we verify the figure of merit of distillation. For any $\rho_n\in\ms{R}_n,$
    \begin{align}
        1-F&(\Lambda_n(\rho_n),\Phi_{2}^{\otimes \lfloor rn\rfloor})\label{eq:distillation 1}\\
        =&1-\tr(\Lambda_n(\rho_n)\Phi_{2}^{\otimes \lfloor rn\rfloor})\label{eq:distillation 2}\\
        =&\tr(M_n\rho_n)\label{eq:distillation 3}\\
        \leq&\max_{\rho_n\in\ms{R}_n}\tr(M_n\rho_n),\label{eq:distillation 4}
    \end{align}
    where the first equality follows that $\Phi_{2}^{\otimes \lfloor rn\rfloor}$ is a pure state, and the second equality results from a direct computation upon substituting $\Lambda_n(\rho_n)$ according to Equation~\eqref{eq: construction of distillation protocol in black-box}. 
    Based on the above, due to definition $	\err(\ms{R}_n\!\!\xrightarrow{\mathcal{F}_{\NE}}\! \Phi_2^{\otimes\lfloor \!rn\! \rfloor})$, we get $\err(\ms{R}_n\!\!\xrightarrow{\mathcal{F}_{\NE}}\! \Phi_2^{\otimes\lfloor \!rn\! \rfloor}) \leq \beta_n(\mc{S}_n,\ms{R}_n,\mathds{ALL},2^{-\lfloor rn \rfloor})$.

\textit{2). Proof of the lower bound:} Let $\Lambda_n$ be the optimal distillation operation for $\err(\ms{R}_n\!\!\xrightarrow{\mathcal{F}_{\NE}}\! \Phi_2^{\otimes\lfloor \!rn\! \rfloor})$.
    Construct a measurement,
    \begin{align}
        \label{eq: construction of measurement in black-box}
        \left\{M_n,\ \mathds{1}-M_n\right\}:=\left\{\mathds{1}-\Lambda_n^\dagger\big(\Phi_2^{\otimes \lfloor rn\rfloor}\big),\ \Lambda_n^\dagger\big(\Phi_2^{\otimes \lfloor rn\rfloor}\big)\right\}.
    \end{align}
    On the one hand, for any $\sigma_n \in \mc{S}_n$, 
    \begin{align}
        \tr\left((\mathds{1}-M_n)\sigma_n\right)&=\tr\left(\Lambda_n^\dagger(\Phi_2^{\otimes \lfloor rn\rfloor})\sigma_n\right)\label{eq: measurement 1-1}\\
        & =\tr\left(\Phi_2^{\otimes \lfloor rn\rfloor}\Lambda_n(\sigma_n)\right)\label{eq: measurement 1-2}\\
        & =\tr\left((\Phi_2^{\otimes \lfloor rn\rfloor})^\Gamma\Lambda_n(\sigma_n)^\Gamma\right)\label{eq: measurement 1-3}\\
        & =\tr\left(\frac{1}{2^{\lfloor rn\rfloor}}W\Lambda_n(\sigma_n)^\Gamma\right)\label{eq: measurement 1-4}\\
        & \leq\tr\left(\frac{1}{2^{\lfloor rn\rfloor}}\Lambda_n(\sigma_n)^\Gamma\right)\label{eq: measurement 1-5}\\
        & =2^{-\lfloor rn\rfloor},\label{eq: measurement 1-6}
    \end{align}
    where $\Gamma$ is the partial transpose on Hilbert space $\mathcal{H}_B^{\ox n}$, and $W$ is the swap operator on the Hilbert space $\mathds{C}_2^{\ox \lfloor rn \rfloor}\ox \mathds{C}_2^{\ox \lfloor rn \rfloor}$. The inequality is because $W\leq\mathds{1}$ and $\Lambda_n(\sigma_n)^\Gamma$ is separable since $\Lambda_n$ is a non-entangling channel, and hence is positive semi-definite. The final equality also follows from the fact that $\Lambda_n(\sigma_n)^\Gamma$ is a separable state. 

    On the other hand, for any $\rho_n\in\ms{R}_n,$
    \begin{align}
        \tr(M_n\rho_n)&=\operatorname{Tr}((\mathds{1}-\Lambda_n^\dagger(\Phi_2^{\otimes \lfloor rn \rfloor}))\rho_n)\label{eq: measurement 2-1}\\
        &=\operatorname{Tr}((\mathds{1}-\Phi_2^{\otimes \lfloor rn \rfloor})\Lambda_n(\rho_n))\label{eq: measurement 2-2}\\
        &=1-\mathrm{F}(\Phi_2^{\otimes \lfloor rn \rfloor},\Lambda_n(\rho_n))\label{eq: measurement 2-3}\\
        & \leq\max_{\rho_n\in\mathscr{R}_n}\left(1-F(\Lambda_n(\rho_n),\Phi_2^{\otimes \lfloor rn \rfloor})\right).\label{eq: measurement 2-4}
    \end{align}
    Due to definition $\beta_n(\mc{S}_n,\ms{R}_n,\mathds{ALL},2^{-\lfloor rn \rfloor})$, we get $\err(\ms{R}_n\!\!\xrightarrow{\mathcal{F}_{\NE}}\! \Phi_2^{\otimes\lfloor \!rn\! \rfloor}) \geq \beta_n(\mc{S}_n,\ms{R}_n,\mathds{ALL},2^{-\lfloor rn \rfloor})$. For the proof of the upper and lower bounds in both directions, we provide a complete proof of the first equality.

    Second, to prove Equation~\eqref{eq:F_PPT} for PPT‑preserving free operations, it suffices to replace the separable states in the above proof with PPT states. Here we only note the clarifications needed to make the proof go through. The distillation protocol constructed here remains as described in Equation~\eqref{eq: construction of distillation protocol in black-box}, except that the measurement now achieves the optimal type‑II error $\beta_n(\mc{PPT}_n,\ms{R}_n,\mathds{ALL},2^{-\!\lfloor\! rn \!\rfloor})$. Because $\operatorname{Tr}((\mathds{1} - M_n)\sigma_n) \le 2^{-\lfloor\! rn\!\rfloor}$ holds for all $\sigma_n \in \mc{PPT}_n$, it follows that $\Lambda_n(\sigma_n) \in \mc{S}_n \subset \mc{PPT}_n$, which implies that $\Lambda_n$ is PPT-preserving. As for the distillation performance, Equations~\eqref{eq:distillation 1}--\eqref{eq:distillation 4} remain fully valid. Additionally, the construction of the quantum measurement for hypothesis testing is the same as in Equation~\eqref{eq: construction of measurement in black-box}. For the type‑I error, because $\Lambda_n(\sigma_n)$ is a PPT state, its partial transpose $\Lambda_n(\sigma_n)^\Gamma$ is positive semi-definite; consequently, Equations~\eqref{eq: measurement 1-1}--\eqref{eq: measurement 1-6} hold. For the type‑II error, Equations~\eqref{eq: measurement 2-1}--\eqref{eq: measurement 2-4} hold directly. With these constructions in both directions, we have proved Equation~\eqref{eq:F_PPT}.
\end{proof}

When the free operations are dually non‑entangling and dually PPT‑preserving, we bound the error of universal entanglement distillation from above and below by the type‑II error in hypothesis testing. These bounds are tight enough to establish the relation between the error exponents for universal entanglement distillation and hypothesis testing. The rigorous statement is given in Theorem~\ref{them: exponents} of Section~\ref{sec: The Reliability Function}.

\begin{theorem}[Finite-blocklength under $\mathcal{F}_{\mathrm{DNE}}$ and $\mathcal{F}_{\mathrm{DPPT}}$]
\label{thm: error in finite blocklenth under DNE and DPPT}
Let $\ms{R}_n$ be a convex compact set of quantum states on $\mathcal{H}_{AB}^{\otimes n}$, and let $r>0$ be a real number. Then the errors of universal entanglement distillation under dually non‑entangling and dually PPT‑preserving operations are bounded above and below by the optimal type‑II errors in hypothesis testing that distinguish the separable set $\mc{S}_n$ and the PPT set $\mc{PPT}_n$ from the initial set $\ms{R}_n$ using measurements $\mathds{SEP}$ and measurements $\mathds{PPT}$, respectively, with the type‑I error bounded by $2^{-\lfloor\! rn\!\rfloor}$ for the upper bound and by $2\cdot2^{-\lfloor\! rn\!\rfloor}$ for the lower bound,
    \begin{align}
        \err(\ms{R}_n\!\xrightarrow{\mathcal{F}_{\mathrm{DNE}}} \!\Phi_2^{\otimes \lfloor\! rn \!\rfloor})\!\leq&\beta_n\!(\mc{S}_n,\ms{R}_n,\mathds{SEP},2^{-\lfloor\! rn \!\rfloor}),\label{eq:F_DNE 1}\\
        \err(\ms{R}_n\!\xrightarrow{\mathcal{F}_{\mathrm{DNE}}} \!\Phi_2^{\otimes \lfloor\! rn \!\rfloor})\!\geq&\beta_n\!(\mc{S}_n,\ms{R}_n,\mathds{SEP},2\!\cdot\!2^{-\lfloor\! rn\! \rfloor}),\label{eq:F_DNE 2}\\
        \err(\ms{R}_n\!\xrightarrow{\mathcal{F}_{\mathrm{DPPT}}} \!\Phi_2^{\otimes \lfloor\! rn \!\rfloor})\!\leq&\beta_n\!(\mc{PPT}_n,\!\ms{R}_n,\mathds{PPT},2^{-\!\lfloor \!rn \!\rfloor}),\label{eq:F_DPPT 1}\\
        \err(\ms{R}_n\!\xrightarrow{\mathcal{F}_{\mathrm{DPPT}}} \!\Phi_2^{\otimes \lfloor \!rn\! \rfloor})\!\geq&\beta_n\!(\mc{PPT}_n,\!\ms{R}_n,\mathds{PPT},2\!\cdot\!2^{-\!\lfloor \!rn \!\rfloor}\!).\label{eq:F_DPPT 2}
    \end{align}
\end{theorem}
\begin{proof}
    The proofs of these inequalities share the same core idea as Theorem~\ref{thm: error in finite blocklenth}, namely, to construct a feasible protocol for one task from an optimal protocol for the other. There are minor modifications and clarifications in the technical details. For brevity, we highlight only the necessary adjustments rather than reproducing the complete proofs.

    First, for Equation~\eqref{eq:F_DNE 1}, we construct a distillation protocol still given by Equation~\eqref{eq: construction of distillation protocol in black-box}, but with the measurement chosen to attain the optimal type‑II error $\beta_n\!(\mc{S}_n,\ms{R}_n,\mathds{SEP},2^{-\lfloor\! rn \!\rfloor})$. It is already known to be non‑entangling in the Schrödinger picture. Its adjoint map is 
    \begin{align}
    \label{eq:adjoint}
        \Lambda_n^\dagger(Y)=\tr&(\Phi_2^{\ox \lfloor \!rn\!\rfloor}Y)(\mathds{1}-M_n)\nonumber\\
        &+\frac{1}{(2^{\lfloor\! rn\!\rfloor})^2-1}\tr((\mathds{1}-\Phi_{2}^{\otimes \lfloor rn\rfloor})Y)M_n.
    \end{align}
    It is also non‑entangling in the Heisenberg picture, since the measurement operators $\{\mathds{1}-M_n,M_n\}$ are separable. Hence, the constructed $\Lambda_n$ is dually non‑entangling. As for the figure of merit of distillation, Equations~\eqref{eq:distillation 1}--\eqref{eq:distillation 4} are applicable here as well. 
    
    Second, for Equation~\eqref{eq:F_DNE 2}, the construction of the measurement for composite correlated hypothesis testing is modified to
    \begin{align}
    M_n:=&\frac{2^{\lfloor rn \rfloor}}{2^{\lfloor rn \rfloor}+1}\Lambda_n^\dagger(\mathds{1}-\Phi_2^{\otimes \lfloor rn\rfloor}),\label{eq:M_n}\\
    \mathds{1}-M_n:=& \Lambda_n^\dagger(\Phi_2^{\otimes \lfloor rn\rfloor}+\frac{\mathds{1}-\Phi_2^{\otimes \lfloor rn\rfloor}}{2^{\lfloor rn \rfloor}+1}).\label{eq:1-M_n}
    \end{align}
    Because $\mathds{1}-\Phi_2^{\otimes \lfloor rn\rfloor}$ and $\Phi_2^{\otimes \lfloor rn\rfloor}+\frac{1}{2^{\lfloor rn \rfloor}+1}(\mathds{1}-\Phi_2^{\otimes \lfloor rn\rfloor})$ are separable and channel $\Lambda_n$ belongs to $\mathcal{F}_{\mathrm{DNE}}$, the above measurement is a separable measurement. To bound the type-I error, for any separable state $\sigma_n\in\mc{S}_n$, it holds that 
    \begin{align}
        \tr(\!(\mathds{1}-M_n\!)\sigma_n)&=\tr((\Phi_2^{\otimes \lfloor \!rn\!\rfloor}\!+\!\frac{\mathds{1}-\Phi_2^{\otimes \lfloor\! rn\!\rfloor}}{2^{\lfloor rn \rfloor}+1})\Lambda_n(\sigma_n))\label{eq:typee 1-1}\\
        & =\! \frac{2^{\!\lfloor \!rn\! \rfloor}}{2^{\!\lfloor\! rn \!\rfloor}\!+\!1}\!\!\tr(\Phi_2^{\otimes \lfloor\! rn\!\rfloor}\!\Lambda_n(\!\sigma_n\!)\!)\!+\!\frac{1}{2^{\!\lfloor\! rn\! \rfloor}\!+\!1}\\
        & \leq\frac{2}{2^{\!\lfloor\! rn\! \rfloor}\!+\!1}\\
        & \leq2\cdot2^{-\lfloor rn\rfloor},\label{eq:typee 1-4}
    \end{align}
    where the first inequality follows from Equations~\eqref{eq: measurement 1-2}--\eqref{eq: measurement 1-6}. As for the type‑II error, since $\frac{2^{\lfloor rn \rfloor}}{2^{\lfloor rn \rfloor}+1}<1$, Equations~\eqref{eq: measurement 2-1}--\eqref{eq: measurement 2-4} are also valid.

    Third, for Equation~\eqref{eq:F_DPPT 1}, we construct a distillation protocol of the form in Equation~\eqref{eq: construction of distillation protocol in black-box} using the measurement that achieves the optimal type‑II error $\beta_n\!(\mc{PPT}_n,\!\ms{R}_n,\mathds{PPT},2^{-\!\lfloor \!rn \!\rfloor})$. Since $\Lambda_n(\sigma_n) \in \mc{S}_n \subset \mc{PPT}_n$ for any $\sigma_n\in \mathcal{PPT}_n$, it is PPT‑preserving in the Schrödinger picture; its dual is given by Equation~\eqref{eq:adjoint}, because the measurement operator is PPT, its adjoint is also PPT‑preserving in the Heisenberg picture, and hence $\Lambda_n$ is dually PPT‑preserving. The distillation performance remains as described in Equations~\eqref{eq:distillation 1}-\eqref{eq:distillation 4}.

    Finally, for Equation~\eqref{eq:F_DPPT 2}, the measurement operators are constructed as in Equations~\eqref{eq:M_n}--\eqref{eq:1-M_n}, which is readily seen to be a PPT measurement. The type‑I error is controlled by Equations~\eqref{eq:typee 1-1}--\eqref{eq:typee 1-4} upon replacing separable states with PPT states, and the type‑II error is controlled by Equations~\eqref{eq: measurement 2-1}--\eqref{eq: measurement 2-4} due to $\frac{2^{\lfloor rn \rfloor}}{2^{\lfloor rn \rfloor}+1}<1$.
\end{proof}

\subsection{The Reliability Function}
\label{sec: The Reliability Function}
We now turn to the asymptotic behavior of entanglement distillation. Let $\ms{R}=\{\ms{R}_n\}_{n\in\NN}$ be a stable sequence of convex compact sets of quantum states, where $\ms{R}_n\subseteq\ms{D}(\mathcal{H}_{AB}^{\ox n})$. Let the free operations be $\mc{O}$.  
We define the lower
and upper reliability functions of universal entanglement
distillation at rate $r>0$ as
\begin{align}
\underline{E}^{\rf}(\ms{R},\mc{O},r)
:=&
\liminf_{n\to\infty}
-\frac{1}{n}
\log
\varepsilon
\left(
\ms{R}_n
\xrightarrow{\mc{O}}
\Phi_2^{\ox\lfloor rn\rfloor}
\right),\\
\overline{E}^{\rf}(\ms{R},\mc{O},r)
:=&
\limsup_{n\to\infty}
-\frac{1}{n}
\log
\varepsilon
\left(
\ms{R}_n
\xrightarrow{\mc{O}}
\Phi_2^{\ox\lfloor rn\rfloor}
\right).
\end{align}
When the limit exists, we define the reliability function of universal entanglement distillation at rate $r$ as
\begin{equation}
    E^{\rf}(\ms{R},\mc{O},r):=\lim_{n\to\infty}-\frac{1}{n}\log\err\left(\ms{R}_n\xrightarrow{\mathcal{O}}\Phi_2^{\otimes \lfloor rn \rfloor}\right).
\end{equation}

Based on the finite-blocklength analysis established in Theorem~\ref{thm: error in finite blocklenth} and Theorem~\ref{thm: error in finite blocklenth under DNE and DPPT}, we establish a correspondence between the reliability function of universal entanglement distillation and the error exponent of composite correlated hypothesis testing.

\begin{theorem}
\label{them: exponents}
Let $\ms{R}=\{\ms{R}_n\}_{n\in\NN}$ be a stable sequence of convex
compact sets of quantum states with each
$\ms{R}_n\subseteq\density(\mathcal{H}_{AB}^{\ox n})$, and let $r>0$.
Then the lower and upper reliability functions of universal
entanglement distillation satisfy
\begin{align}
\underline{E}_{HT}&(\ms{F},\ms{R},\mathds{M},r)
\le
\underline{E}^{\rf}(\ms{R},\mc{O},r)
\nonumber\\
&\le
\overline{E}^{\rf}(\ms{R},\mc{O},r)
\le
\overline{E}_{HT}(\ms{F},\ms{R},\mathds{M},r^-),
\label{eq:rf-HT-sandwich}
\end{align}
with the corresponding triples
\begin{align*}
(\mc{O},\ms{F},\mathds{M})
\in
\bigl\{
&(\mc{F}_{\NE},\mc{S},\mathds{ALL}),
(\mc{F}_{\mathrm{PPT}},\mc{PPT},\mathds{ALL}),
\nonumber\\
&(\mc{F}_{\mathrm{DNE}},\mc{S},\mathds{SEP}),
(\mc{F}_{\mathrm{DPPT}},\mc{PPT},\mathds{PPT})
\bigr\}.
\end{align*}

In particular, if
\begin{equation}
\underline{E}_{HT}(\ms{F},\ms{R},\mathds{M},r)
=
\overline{E}_{HT}(\ms{F},\ms{R},\mathds{M},r^-),
\label{eq:HT-regularity-condition}
\end{equation}
then the reliability-function limit exists and
\begin{equation}
E^{\rf}(\ms{R},\mc{O},r)
={E}_{HT}(\ms{F},\ms{R},\mathds{M},r)
\end{equation}
\end{theorem}

\begin{proof}
For each of the four classes of free operations, the finite-blocklength
results established above express the distillation error either exactly
as, or bound it between, hypothesis testing type-II errors whose
type-I constraints are of the form
\begin{equation}
c\,2^{-\lfloor rn\rfloor},
\qquad
c\in\{1,2\}.
\end{equation}
Fix $0<\delta<r$. Then
for all sufficiently large $n$,
\begin{equation}
2^{-nr}
\le
c\,2^{-\lfloor rn\rfloor}
\le
2^{-n(r-\delta)}.
\label{eq:threshold-sandwich}
\end{equation}
Since
$\beta_n(\ms{F}_n,\ms{R}_n,\mathds{M},\varepsilon)$
is non-increasing in the allowed type-I error $\varepsilon$,
Equation~\eqref{eq:threshold-sandwich}, together with the
finite-blocklength characterizations, implies
\begin{align}
-\frac1n&
\log
\beta_n(
\ms{F}_n,\ms{R}_n,\mathds{M},2^{-nr})\nonumber\\
\le&
-\frac1n
\log
\err\left(
\ms{R}_n\xrightarrow{\mc O}
\Phi_2^{\ox\lfloor rn\rfloor}
\right)
\nonumber\\
\le&
-\frac1n
\log
\beta_n(
\ms{F}_n,\ms{R}_n,\mathds{M},
2^{-n(r-\delta)})
\end{align}
for all sufficiently large $n$.

Taking $\liminf$ in the left inequality and $\limsup$ in the right
inequality gives
\begin{align}
\underline{E}_{HT}(\ms{F},\ms{R},\mathds{M},r)
&\le
\underline{E}^{\rf}(\ms{R},\mc{O},r)
\nonumber\\
&\le
\overline{E}^{\rf}(\ms{R},\mc{O},r)
\le
\overline{E}_{HT}(
\ms{F},\ms{R},\mathds{M},r-\delta).
\end{align}
Letting $\delta$ tend to $0$ proves
Equation~\eqref{eq:rf-HT-sandwich}.

Finally, if
Equation~\eqref{eq:HT-regularity-condition} holds, the two outer
bounds coincide. Hence
$\underline{E}^{\rf}=\overline{E}^{\rf}$, so the reliability-function
limit exists and equals their common value.
\end{proof}

Using the error exponent established in Theorem~\ref{them:Fang and Hayashi} of Fang and Hayashi for composite correlated hypothesis testing under all measurements, together with the relationship between the error exponents of universal entanglement distillation and hypothesis testing given in Theorem~\ref{them: exponents}, we obtain the reliability functions of universal entanglement distillation under non‑entangling and PPT‑preserving operations.

\begin{theorem}[Reliability function under
$\mc{F}_{\NE}$ and $\mc{F}_{\mathrm{PPT}}$]
\label{them:rf,NE,PPT}
Let $\ms{R}=\{\ms{R}_n\}_{n\in\NN}$ be a stable sequence of convex
compact sets of quantum states with
$\ms{R}_n\subseteq\density(\mathcal{H}_{AB}^{\ox n})$. Then the
reliability functions of universal entanglement distillation under
non-entangling and PPT-preserving operations are given, respectively,
by
\begin{align}
&E^{\rf}(\ms{R},\mc{F}_{\NE},r)
\nonumber\\
&\quad=
\begin{cases}
+\infty,
&
0<r\leq
D_{\Petz,0}^{\infty}(\ms{R}\Vert\mc S),
\\[1mm]
H_r^\infty(\ms{R}\Vert\mc S),
&
D_{\Petz,0}^{\infty}(\ms{R}\Vert\mc S)
<
r
<
D^\infty(\ms{R}\Vert\mc S),
\end{cases}
\label{eq:rf,NE}
\end{align}
and
\begin{align}
&E^{\rf}(\ms{R},\mc{F}_{\mathrm{PPT}},r)
\nonumber\\
&=
\begin{cases}
+\infty,&0<r\leq D_{\Petz,0}^{\infty}(\ms{R}\Vert\mc{PPT}),
\\
H_r^\infty(\ms{R}\Vert\mc{PPT}),&D_{\Petz,0}^{\infty}(\ms{R}\Vert\mc{PPT})
\!<
\!r\!
<
D^\infty(\ms{R}\Vert\mc{PPT}).
\end{cases}
\label{eq:rf,PPT}
\end{align}
In particular, the canonical target sequence at the critical
zero-error rate can be distilled with zero error, and hence its
reliability exponent is $+\infty$.
\end{theorem}

\begin{proof}
We prove both statements simultaneously. Let
\begin{equation}
(\mc O,\ms F)
\in
\left\{
(\mc F_{\NE},\mc S),
(\mc F_{\mathrm{PPT}},\mc{PPT})
\right\},
\end{equation}
and set
\begin{equation}
d_n
:=
D_{\Petz,0}(\ms R_n\Vert\ms F_n),
\qquad
r_0
:=
D_{\Petz,0}^{\infty}(\ms R\Vert\ms F).
\label{eq:rn-order-zero-threshold}
\end{equation}

By stability of $\ms R$ and $\ms F$ and additivity of the order-zero
Petz R\'enyi divergence under tensor products, the sequence
$\{d_n\}_{n\in\NN}$ is subadditive:
\begin{equation}
d_{m+n}\leq d_m+d_n.
\end{equation}
Fekete's lemma therefore gives
\begin{equation}
r_0
=
\lim_{n\to\infty}\frac{d_n}{n}
=
\inf_{n\geq1}\frac{d_n}{n}.
\label{eq:r0-infimum}
\end{equation}
In particular,
\begin{equation}
d_n\geq nr_0,
\qquad
\forall n\in\NN.
\label{eq:dn-lower-r0}
\end{equation}

We first consider the zero-error region
\begin{equation}
0<r\leq r_0.
\end{equation}
Let
\begin{equation}
\Pi_{\ms R_n}
:=
\bigvee_{\rho_n\in\ms R_n}\Pi_{\rho_n}
\end{equation}
be the projection onto the linear span of the supports of all states
in $\ms R_n$. Since the underlying Hilbert space is finite
dimensional, there exist finitely many states in $\ms R_n$ whose
supports span the range of $\Pi_{\ms R_n}$. By convexity of
$\ms R_n$, a strictly positive convex combination of these states
belongs to $\ms R_n$ and has support projection
$\Pi_{\ms R_n}$. Consequently,
\begin{equation}
d_n
=
-\log
\max_{\sigma_n\in\ms F_n}
\tr(\Pi_{\ms R_n}\sigma_n).
\label{eq:D0-joint-support}
\end{equation}

By Equations~\eqref{eq:dn-lower-r0} and $r\leq r_0$,
\begin{equation}
\lfloor nr\rfloor
\leq nr
\leq nr_0
\leq d_n.
\label{eq:mn-below-dn}
\end{equation}
Define the hypothesis testing operator
\begin{equation}
M_n:=\mathds{1}-\Pi_{\ms R_n}.
\end{equation}
For every $\rho_n\in\ms R_n$,
\begin{equation}
\tr(M_n\rho_n)=0,
\end{equation}
while Equation~\eqref{eq:D0-joint-support} gives
\begin{align}
\max_{\sigma_n\in\ms F_n}
\tr[(\mathds{1}-M_n)\sigma_n]
&=
\max_{\sigma_n\in\ms F_n}
\tr(\Pi_{\ms R_n}\sigma_n)
\nonumber\\
&=
2^{-d_n}
\leq
2^{-\lfloor nr\rfloor}.
\end{align}
Thus $M_n$ is feasible for
\begin{equation}
\beta_n(
\ms F_n,\ms R_n,\mathds{ALL},2^{-\lfloor nr\rfloor})
\end{equation}
and achieves zero type-II error. Hence
\begin{equation}
\beta_n(
\ms F_n,\ms R_n,\mathds{ALL},2^{-\lfloor nr\rfloor})
=
0.
\label{eq:beta-zero-below-r0}
\end{equation}
By the finite-blocklength equivalence in
Theorem~\ref{thm: error in finite blocklenth},
\begin{equation}
\err\left(
\ms R_n\xrightarrow{\mc O}
\Phi_2^{\ox\lfloor nr\rfloor}
\right)
=
0.
\end{equation}
Therefore,
\begin{equation}
E^{\rf}(\ms R,\mc O,r)=+\infty,
\qquad
0<r\leq r_0.
\label{eq:rf-infinite-up-to-r0}
\end{equation}

Notice that, when $0<r<r_0$, this operational conclusion agrees with
the Hoeffding divergence. Indeed, by
Equation~\eqref{eq:r0-infimum},
\begin{equation}
nr<d_n
=
D_{\Petz,0}(\ms R_n\Vert\ms F_n)
\end{equation}
for every $n$, and hence
\begin{equation}
H_{n,r}(\ms R_n\Vert\ms F_n)=+\infty.
\end{equation}
Thus
\begin{equation}
H_r^\infty(\ms R\Vert\ms F)=+\infty,
\qquad
0<r<r_0.
\end{equation}

We next consider the region
\begin{equation}
r_0<r<D^\infty(\ms R\Vert\ms F).
\label{eq:finite-Hoeffding-region}
\end{equation}
According to Theorem~\ref{them: exponents}, it remains only to verify the left continuity of the regularized quantum Hoeffding divergence with respect to the rate parameter.

We first show that
$H_r^\infty(\ms R\Vert\ms F)$ is finite in this region. By
Equation~\eqref{eq:r0-infimum}, there exists some $k\in\NN$ such that
\begin{equation}
\frac1k
D_{\Petz,0}(\ms R_k\Vert\ms F_k)<r.
\end{equation}
It follows that
\begin{equation}
H_{k,r}(\ms R_k\Vert\ms F_k)<\infty.
\end{equation}
Together with the subadditivity of the Hoeffding divergence, this
implies
\begin{equation}
H_r^\infty(\ms R\Vert\ms F)<\infty.
\end{equation}

Moreover, the map
\begin{equation}
r\longmapsto H_r^\infty(\ms R\Vert\ms F)
\end{equation}
is convex on $(r_0,\infty)$. To see this, fix
$r_1,r_2>r_0$ and first let $\lambda=p/q\in[0,1]$ be rational.
For every $n\in\NN$ and $\eta>0$, choose
$\rho_n^{(i)}\in\ms R_n$ and $\sigma_n^{(i)}\in\ms F_n$,
$i=1,2$, such that
\begin{equation}
H_{n,r_i}
(\rho_n^{(i)}\Vert\sigma_n^{(i)})
\leq
H_{n,r_i}(\ms R_n\Vert\ms F_n)+\eta n.
\end{equation}
By stability, the concatenated states
\begin{align}
\rho_{qn}
&:=
\bigl(\rho_n^{(1)}\bigr)^{\ox p}
\ox
\bigl(\rho_n^{(2)}\bigr)^{\ox(q-p)}
\in\ms R_{qn},
\\
\sigma_{qn}
&:=
\bigl(\sigma_n^{(1)}\bigr)^{\ox p}
\ox
\bigl(\sigma_n^{(2)}\bigr)^{\ox(q-p)}
\in\ms F_{qn}.
\end{align}
Using the additivity of the Petz R\'enyi divergence and
\begin{equation}
\sup_\alpha(a_\alpha+b_\alpha)
\leq
\sup_\alpha a_\alpha+\sup_\alpha b_\alpha,
\end{equation}
we obtain
\begin{align}
&H_{qn,\lambda r_1+(1-\lambda)r_2}
(\ms R_{qn}\Vert\ms F_{qn})
\nonumber\\
&\quad\leq
pH_{n,r_1}(\ms R_n\Vert\ms F_n)
+
(q-p)H_{n,r_2}(\ms R_n\Vert\ms F_n)
+
q\eta n.
\end{align}
Dividing by $qn$, letting $n\to\infty$, and then letting
$\eta$ tend to $0$ gives
\begin{align}
&H_{\lambda r_1+(1-\lambda)r_2}^{\infty}
(\ms R\Vert\ms F)
\nonumber\\
&\quad\leq
\lambda H_{r_1}^{\infty}(\ms R\Vert\ms F)
+
(1-\lambda)H_{r_2}^{\infty}(\ms R\Vert\ms F).
\end{align}
This first holds for rational $\lambda$. Since
$H_r^\infty(\ms R\Vert\ms F)$ is finite and non-increasing on
$(r_0,\infty)$, it is locally bounded there; hence rational convexity
extends to arbitrary $\lambda\in[0,1]$. Therefore,
$H_r^\infty(\ms R\Vert\ms F)$ is a finite convex function on
$(r_0,\infty)$ and is consequently continuous there. In particular,
\begin{equation}
\lim_{\delta\to 0}
H_{r-\delta}^{\infty}(\ms R\Vert\ms F)
=
H_r^{\infty}(\ms R\Vert\ms F),
\qquad
r>r_0.
\label{eq:H-left-continuity-above-r0}
\end{equation}

For the rates in Equation~\eqref{eq:finite-Hoeffding-region},
Theorem~\ref{them: exponents} and the composite Hoeffding theorem give
\begin{align}
H_r^\infty(\ms R\Vert\ms F)
&\leq
\underline E^{\rf}(\ms R,\mc O,r)
\nonumber\\
&\leq
\overline E^{\rf}(\ms R,\mc O,r)
\nonumber\\
&\leq
\lim_{\delta\to0}
H_{r-\delta}^\infty(\ms R\Vert\ms F).
\end{align}
Using Equation~\eqref{eq:H-left-continuity-above-r0}, the two outer
bounds coincide. Hence the reliability-function limit exists and
\begin{equation}
E^{\rf}(\ms R,\mc O,r)
=
H_r^\infty(\ms R\Vert\ms F),~
r_0<r<D^\infty(\ms R\Vert\ms F).
\end{equation}
Taking
$(\mc O,\ms F)
=
(\mc F_{\NE},\mc S)$
and
$(\mc O,\ms F)
=
(\mc F_{\mathrm{PPT}},\mc{PPT})$
proves Equations~\eqref{eq:rf,NE}
and~\eqref{eq:rf,PPT}, respectively.
\end{proof}

Note from the above proof that when
$D_{P,0}^{\infty}(\ms{R}\Vert\ms{F})>0$, if the distillation rate
$r$ lies in
$\bigl(0,D_{P,0}^{\infty}(\ms{R}\Vert\ms{F})\bigr],$
then the reliability function
$H_r^\infty(\ms{R}\Vert\ms{F})=+\infty.$
This is consistent with the fact, established later in
Section~\ref{sec:Zero-Error Distillation Rate under NE and PPT},
that the zero-error distillation rate is characterized by
$D_{P,0}^{\infty}(\ms{R}\Vert\ms{F})$.
At the critical rate
\begin{equation}
r=D_{\Petz,0}^{\infty}(\ms R\Vert\ms F),
\end{equation}
the operational reliability exponent is $+\infty$ because the
canonical target
$\Phi_2^{\ox\lfloor nr\rfloor}$ is exactly achievable at every
blocklength. This conclusion follows from the finite-blocklength
order-zero criterion and the identity
\begin{equation}
D_{\Petz,0}^{\infty}(\ms R\Vert\ms F)
=
\inf_{n\geq1}
\frac1n
D_{\Petz,0}(\ms R_n\Vert\ms F_n).
\end{equation}
The equality
\begin{equation}
E^{\rf}(\ms R,\mc O,r)
=
H_r^\infty(\ms R\Vert\ms F)
\end{equation}
is not asserted at this critical point unless one additionally knows
that
$H_r^\infty(\ms R\Vert\ms F)=+\infty$ there.

Given complete information about the initial state, where the distillation proceeds from a precisely known quantum state, i.e.,  $\ms{R}_n=\{\rho_{AB}^{\ox n}\}$. Then, under non‑entangling and PPT‑preserving free operations, the reliability functions of entanglement distillation with full knowledge of the initial state $\rho$ are $H_r^\infty(\rho\|\mc{S})$ and $H_r^\infty(\rho\|\mc{PPT})$, respectively. 

In particular, when the initial state $\rho$ is a maximally correlated state $\rho^{\mathrm{mc}}_{AB}$, that is, there exist bases $\{|a_i\rangle\}_i$ and $\{|b_i\rangle\}_i$ of $\mathcal{H}_A$ and $\mathcal{H}_B$ such that the support of $\rho$ is contained in the subspace spanned by $\{|a_i\rangle|b_i\rangle\}_i$, the Petz Rényi relative entropy of the maximally correlated state with respect to both the set of all separable states and the set of all PPT states is additive~\cite[Corollary 5]{Zhu_2017}, i.e., $$D_{\Petz,\alpha}\left((\rho^{\mathrm{mc}})^{\ox n}\|\ms{F}_n\right)=nD_{\Petz,\alpha}(\rho^{\mathrm{mc}}\|\ms{F}_1).$$
Moreover, it can be expressed in closed form~\cite[Corollary~1]{Zhu_2017} as
\begin{equation}
\label{eq: maximally Petz}
    D_{\Petz,\alpha}(\rho^{\mathrm{mc}}\|\ms{F}_1)=\frac{1}{\alpha - 1} \log \left\| \left((\rho^{\mathrm{mc}})^{\alpha}\right)^{\text{diag}} \right\|_{1/\alpha},
\end{equation}
where $\left((\rho^{\mathrm{mc}})^{\alpha}\right)^{\text{diag}}$ denotes the diagonal matrix with the same diagonal as $(\rho^{\mathrm{mc}})^{\alpha}$. In particular, as $\alpha$ approaches 1, 
$$\lim_{\alpha\to1}D_{\Petz,\alpha}(\rho^{\mathrm{mc}}\|\ms{F}_1)=D(\rho^{\mathrm{mc}}\|\ms{F}_1)=S(\rho^{\mathrm{mc}}_A)-S(\rho^{\mathrm{mc}}),$$
where $\rho^{\mathrm{mc}}_A$ is the reduced state of $\rho^{\mathrm{mc}}$ on system $A$, and $S(\rho) = -\tr(\rho\log\rho)$ is the von Neumann entropy.

Therefore, when the initial state is a maximally correlated state, we immediately obtain a single‑letter expression for the reliability function.
\begin{corollary}
\label{cor: maximally correlated state}
    Let $\rho^{\mathrm{mc}}$ be a maximally correlated state on Hilbert space  $\mc{H}_{AB}$. Let $r\in(0,S(\rho^{\mathrm{mc}}_A)-S(\rho^{\mathrm{mc}}))$ be a real number. When the distillation protocol is restricted to non‑entangling $\mc{F}_{\NE}$ or PPT‑preserving $\mc{F}_{\mathrm{PPT}}$ operations, the reliability function of entanglement distillation in both cases is
    \begin{align}
        &E^{\rf}(\rho^{\mathrm{mc}},\mc{F}_{\NE},r)=E^{\rf}(\rho^{\mathrm{mc}},\mc{F}_{\mathrm{PPT}},r)\nonumber\\
        =&\sup_{\alpha\in(0,1)}\frac{1}{\alpha}\left((\alpha-1)r-\log\left\| \left((\rho^{\mathrm{mc}})^{\alpha}\right)^{\text{diag}} \right\|_{1/\alpha}\right).
    \end{align}
\end{corollary}

Note that maximally correlated states include all pure states. Furthermore, the errors of entanglement distillation for LOCC, non-entangling and PPT-preserving operations coincide in the distillation of pure entangled states~\cite{regula2019one}. Therefore, we immediately reproduce the result of Hayashi et al. on the exponential analysis of entanglement concentration for pure entangled states under LOCC~\cite{hayashi2002error}.

More specifically,
    let the distillation begin with a pure state $\psi$, whose Schmidt decomposition is given by $|\psi\>=\sum_{i=1}^{d}\sqrt{\lambda_i}|i\>|i\>$. 
    Note that the setting of Hayashi et al.~\cite{hayashi2002error} differs from ours in that they define the error exponent as $r$ and the distillation rate $E(r)$ as a function of the error exponent $r$. Thus, our result yields 
    \begin{align}
        E(r)=&\sup_{\alpha\in(0,1)}\frac{\alpha-1}{\alpha}(r-D_{\Petz,\alpha}(\ms{F}_1\|\psi))\nonumber\\
        =&\sup_{\alpha\in(0,1)}(-\frac{\alpha}{1-\alpha}r+D_{\Petz,\alpha}(\psi\|\ms{F}_1))\nonumber\\
        =&\sup_{\alpha\in(0,1)}(-\frac{\alpha}{1-\alpha}r+\frac{\alpha}{\alpha-1}\log(\sum_{i}{\lambda_i}^{\frac{1}{\alpha}}))\nonumber\\
        =&\sup_{s\geq1}\frac{1}{1-s}(r+\log\sum_i\lambda_i^s)\label{eq: reproduce Hayashi}.
    \end{align}
    The first equality holds because of the symmetry between type-I and type-II errors in the Hoeffding setting; see Section~\ref{sec:distillable entanglement under NE and PPT} for further explanation. The second equality follows by substituting $1-\alpha$ for $\alpha$. The third equality follows from Equation~\eqref{eq: maximally Petz}. The last  equality is obtained via the substitution of $\alpha$ by $\frac{1}{s}$. Equation~\eqref{eq: reproduce Hayashi} is precisely the result presented by Hayashi et al.~\cite{hayashi2002error} in their equation (116).

\subsection{Zero-Rate Error Exponent and Distillable Entanglement under $\mc{F}_{\NE}$ and $\mc{F}_{\mathrm{PPT}}$}
In this section, we study two endpoint quantities associated with the
reliability function established in the preceding section: the zero-rate
error exponent and the distillable entanglement.
The zero-rate error exponent characterizes the largest exponential
decay rate of the distillation error that can be guaranteed for every
fixed target size. We define the zero-rate error exponent of universal
entanglement distillation as
\begin{align}
E^{\rf}(&\ms{R},\mc{O}):=
\sup\{
r\geq 0:\nonumber\\
&\forall m\in\NN,\;
\liminf_{n\to\infty}
-\frac{1}{n}
\log
\err\left(
\ms{R}_n
\xrightarrow{\mc{O}}
\Phi_2^{\ox m}
\right)
\geq r
\}.
\label{eq:zero-rate-rf}
\end{align}
Equivalently, since the distillation error is non-decreasing in the
target size $m$, the corresponding error exponent is non-increasing in
$m$, and hence
\begin{align}
E^{\rf}(\ms{R},\mc{O})
=
\lim_{m\to\infty}
\liminf_{n\to\infty}
-\frac{1}{n}
\log
\err\left(
\ms{R}_n
\xrightarrow{\mc{O}}
\Phi_2^{\ox m}
\right).
\label{eq:zero-rate-rf-equivalent}
\end{align}
The distillable entanglement refers to the maximum achievable distillation rate when the distillation error tends to zero,
\begin{align}
    E_d(\ms R,\mc O)
:=
\sup\left\{
r:
\lim_{n\to\infty}
\err\!\left(
\ms R_n\xrightarrow{\mc O}
\Phi_2^{\otimes\lfloor\! rn\!\rfloor}
\right)=0
\right\}.
\end{align}

We show below that both quantities can be obtained directly as endpoint
limits of the Hoeffding error-exponent regime, without invoking a separate
Stein or Sanov theorem.

\subsubsection{Zero-Rate Error Exponent under
$\mc{F}_{\NE}$ and $\mc{F}_{\mathrm{PPT}}$}
\label{sec:Zero-Rate Error Exponent under NE and PPT}

We first characterize the zero-rate endpoint of the regularized quantum Hoeffding divergence, which will serve as the key ingredient in relating the fixed-rate reliability function to the zero-rate error exponent.

\begin{lemma}[Uniform endpoint of the Hoeffding reliability function]
\label{lem:uniform-endpoint}
Let $\ms{A}=\{\ms{A}_n\}_{n\in\NN}$ and
$\ms{B}=\{\ms{B}_n\}_{n\in\NN}$ be stable sequences of convex compact
sets of quantum states. Suppose that
\begin{equation}
\sup_{\alpha\in(0,1)}
D_{\Petz,\alpha}^{\infty}(\ms{B}\Vert\ms{A})
=
D^{\infty}(\ms{B}\Vert\ms{A}).
\label{eq:reverse-petz-continuity}
\end{equation}
Then
\begin{equation}
\lim_{r\to0}
H_r^{\infty}(\ms{A}\Vert\ms{B})
=
D^{\infty}(\ms{B}\Vert\ms{A}).
\label{eq:uniform-H-endpoint}
\end{equation}
\end{lemma}

\begin{proof}
For every $r>0$, monotonicity of the Hoeffding divergence in $r$ gives
\begin{align}
H_r^{\infty}(\ms{A}\Vert\ms{B})
&\le
\lim_{n\to\infty}\frac{1}{n}
\inf_{\substack{\rho_n\in\ms{A}_n\\
                 \sigma_n\in\ms{B}_n}}
\sup_{\alpha\in(0,1)}
\frac{1-\alpha}{\alpha}
D_{\Petz,\alpha}(\rho_n\Vert\sigma_n)
\nonumber\\
&=
D^{\infty}(\ms{B}\Vert\ms{A}),
\label{eq:uniform-endpoint-upper}
\end{align}
where we used
\begin{equation}
\sup_{\alpha\in(0,1)}
\frac{1-\alpha}{\alpha}
D_{\Petz,\alpha}(\rho\Vert\sigma)
=
D(\sigma\Vert\rho).
\end{equation}

For the converse direction, recall the second regularized Hoeffding
quantity
\begin{equation}
\mathfrak{H}_r^{\infty}(\ms{A}\Vert\ms{B})
:=
\sup_{\alpha\in(0,1)}
\frac{\alpha-1}{\alpha}
\left(
r-D_{\Petz,\alpha}^{\infty}(\ms{A}\Vert\ms{B})
\right),
\end{equation}
which satisfies
\begin{equation}
\mathfrak{H}_r^{\infty}(\ms{A}\Vert\ms{B})
\le
H_r^{\infty}(\ms{A}\Vert\ms{B})
\end{equation}
by~\cite{fang2025errorexponentsquantumstate}. 
Using 
\begin{equation}
\frac{1-\alpha}{\alpha}
D_{\Petz,\alpha}^{\infty}(\ms{A}\Vert\ms{B})
=
D_{\Petz,1-\alpha}^{\infty}(\ms{B}\Vert\ms{A}),
\end{equation}
and setting $\gamma = 1-\alpha$, we obtain
\begin{align}
\mathfrak{H}_r^{\infty}(\ms{A}\Vert\ms{B})
=
\sup_{\gamma\in(0,1)}
\left\{
D_{\Petz,\gamma}^{\infty}(\ms{B}\Vert\ms{A})
-
\frac{\gamma}{1-\gamma}r
\right\}.
\label{eq:dual-H-expression}
\end{align}
Hence, for every fixed $\gamma\in(0,1)$,
\begin{equation}
\liminf_{r\to0}
H_r^{\infty}(\ms{A}\Vert\ms{B})
\ge
D_{\Petz,\gamma}^{\infty}(\ms{B}\Vert\ms{A}).
\end{equation}
Taking the supremum over $\gamma$ and using
Equation~\eqref{eq:reverse-petz-continuity} yields
\begin{equation}
\liminf_{r\to0}
H_r^{\infty}(\ms{A}\Vert\ms{B})
\ge
D^{\infty}(\ms{B}\Vert\ms{A}).
\end{equation}
Together with Equation~\eqref{eq:uniform-endpoint-upper}, this proves
Equation~\eqref{eq:uniform-H-endpoint}.
\end{proof}

With the zero-rate endpoint of the regularized quantum Hoeffding divergence characterized above, we now relate it to the fixed yield distillation problem and derive the corresponding zero-rate error exponent.

\begin{theorem}[Zero-rate error exponent]
\label{thm:zero-rate-rf}
Let $\ms{R}=\{\ms{R}_n\}_{n\in\NN}$ be a stable sequence of convex
compact sets of bipartite quantum states, and assume $D^{\infty}(\ms R\|\ms F)>0$. Let
$\mc{O}\in\{\mc{F}_{\NE},\mc{F}_{\mathrm{PPT}}\}$, and let
$\ms{F}=\mc{S}$ for $\mc{O}=\mc{F}_{\NE}$ and
$\ms{F}=\mc{PPT}$ for $\mc{O}=\mc{F}_{\mathrm{PPT}}$.
Suppose that the regularized Petz R\'enyi relative entropy satisfies
the continuity condition
\begin{equation}
\sup_{\alpha\in(0,1)}
D_{\Petz,\alpha}^{\infty}(\ms{F}\Vert\ms{R})
=
D^{\infty}(\ms{F}\Vert\ms{R}).
\label{eq:zero-rate-continuity-assumption}
\end{equation}
Then the zero-rate error exponent of universal entanglement
distillation is
\begin{equation}
E^{\rf}(\ms{R},\mc{O})
=
D^{\infty}(\ms{F}\Vert\ms{R}).
\label{eq:zero-rate-general}
\end{equation}
In particular,
\begin{align}
E^{\rf}(\ms{R},\mc{F}_{\NE})
&=
D^{\infty}(\mc{S}\Vert\ms{R}),
\\
E^{\rf}(\ms{R},\mc{F}_{\mathrm{PPT}})
&=
D^{\infty}(\mc{PPT}\Vert\ms{R}).
\end{align}
\end{theorem}

\begin{proof}
Fix an arbitrary $m\in\NN$. By the finite-blocklength equivalence
established in Theorem~\ref{thm: error in finite blocklenth}, the
distillation error for the fixed target $\Phi_2^{\ox m}$ is
\begin{equation}
\err\left(
\ms{R}_n
\xrightarrow{\mc{O}}
\Phi_2^{\ox m}
\right)
=
\beta_n(
\ms{F}_n,\ms{R}_n,\mathds{ALL},2^{-m}).
\label{eq:fixed-m-beta}
\end{equation}

We first prove the lower bound
\begin{equation}
E^{\rf}(\ms{R},\mc{O})
\geq
D^\infty(\ms{F}\Vert\ms{R}).
\end{equation}
Fix any $r>0$ in the direct regime. For all sufficiently large $n$,
\begin{equation}
2^{-nr}\leq2^{-m}.
\end{equation}
Since
$\beta_n(\ms{F}_n,\ms{R}_n,\mathds{ALL},\varepsilon)$ is
non-increasing in the allowed type-I error $\varepsilon$, it follows
that
\begin{equation}
\beta_n(
\ms{F}_n,\ms{R}_n,\mathds{ALL},2^{-m})
\leq
\beta_n(
\ms{F}_n,\ms{R}_n,\mathds{ALL},2^{-nr}).
\end{equation}
Therefore, by the fixed-rate reliability-function characterization,
\begin{align}
&\liminf_{n\to\infty}
-\frac{1}{n}
\log
\err\left(
\ms{R}_n
\xrightarrow{\mc{O}}
\Phi_2^{\ox m}
\right)
\geq
H_r^\infty(\ms{R}\Vert\ms{F}).
\end{align}
Since this holds for every sufficiently small fixed $r>0$, letting
$r$ tend to $0$ and applying
Lemma~\ref{lem:uniform-endpoint} gives
\begin{equation}
\liminf_{n\to\infty}
-\frac{1}{n}
\log
\err\left(
\ms{R}_n
\xrightarrow{\mc{O}}
\Phi_2^{\ox m}
\right)
\geq
D^\infty(\ms{F}\Vert\ms{R}).
\label{eq:zero-rate-lower}
\end{equation}
Since $m\in\NN$ was arbitrary, the definition
in Equation~\eqref{eq:zero-rate-rf} yields
\begin{equation}
E^{\rf}(\ms{R},\mc{O})
\geq
D^\infty(\ms{F}\Vert\ms{R}).
\label{eq:zero-rate-operational-lower}
\end{equation}

We next prove the converse bound. Set
\begin{equation}
\beta_{n,m}
:=
\beta_n(
\ms{F}_n,\ms{R}_n,\mathds{ALL},2^{-m}),
\end{equation}
and let $M_n$ be an optimal test. For arbitrary
$\sigma_n\in\ms{F}_n$ and $\rho_n\in\ms{R}_n$, define
\begin{equation}
a_n
:=
\tr[(\mathds{1}-M_n)\sigma_n],
\qquad
b_n
:=
\tr(M_n\rho_n).
\end{equation}
By feasibility and optimality,
\begin{equation}
a_n\leq2^{-m},
\qquad
b_n\leq\beta_{n,m}.
\end{equation}
Applying the data-processing inequality for the Umegaki relative
entropy to the binary measurement
$\{M_n,\mathds{1}-M_n\}$ gives
\begin{align}
D(\sigma_n\Vert\rho_n)
&\geq
D\bigl(
(1-a_n,a_n)
\Vert
(b_n,1-b_n)
\bigr)
\nonumber\\
&=
(1-a_n)\log\frac{1-a_n}{b_n}
+
a_n\log\frac{a_n}{1-b_n}
\nonumber\\
&=
(1-a_n)\log\frac{1}{b_n}
-h_2(a_n)
-a_n\log(1-b_n)
\nonumber\\
&\geq
(1-a_n)\log\frac{1}{b_n}
-h_2(a_n)
\nonumber\\
&\geq
(1-2^{-m})
\log\frac{1}{\beta_{n,m}}
-1,
\end{align}
where $h_2$ denotes the binary entropy and we have used
$h_2(a_n)\leq1$.

Since the preceding inequality holds for every
$\sigma_n\in\ms{F}_n$ and $\rho_n\in\ms{R}_n$, taking the infimum over
the two state sets yields
\begin{equation}
D(\ms{F}_n\Vert\ms{R}_n)
\geq
(1-2^{-m})
\log\frac{1}{\beta_{n,m}}
-1.
\end{equation}
Hence
\begin{equation}
-\frac{1}{n}\log\beta_{n,m}
\leq
\frac{
\frac{1}{n}D(\ms{F}_n\Vert\ms{R}_n)+\frac{1}{n}
}{
1-2^{-m}
}.
\end{equation}
Taking $\limsup_{n\to\infty}$ gives
\begin{equation}
\limsup_{n\to\infty}
-\frac{1}{n}\log\beta_{n,m}
\leq
\frac{
D^\infty(\ms{F}\Vert\ms{R})
}{
1-2^{-m}
}.
\label{eq:zero-rate-upper-fixed-m}
\end{equation}
In particular,
\begin{equation}
\liminf_{n\to\infty}
-\frac{1}{n}\log\beta_{n,m}
\leq
\frac{
D^\infty(\ms{F}\Vert\ms{R})
}{
1-2^{-m}
}.
\end{equation}
Using Equation~\eqref{eq:fixed-m-beta} and the equivalent form
\eqref{eq:zero-rate-rf-equivalent} of the zero-rate exponent, we obtain
\begin{align}
E^{\rf}(\ms{R},\mc{O})
&=
\lim_{m\to\infty}
\liminf_{n\to\infty}
-\frac{1}{n}
\log
\err\left(
\ms{R}_n
\xrightarrow{\mc{O}}
\Phi_2^{\ox m}
\right)
\nonumber\\
&\leq
\lim_{m\to\infty}
\frac{
D^\infty(\ms{F}\Vert\ms{R})
}{
1-2^{-m}
}
\nonumber\\
&=
D^\infty(\ms{F}\Vert\ms{R}).
\label{eq:zero-rate-operational-upper}
\end{align}
Combining
Equations~\eqref{eq:zero-rate-operational-lower}
and~\eqref{eq:zero-rate-operational-upper} proves
\begin{equation}
E^{\rf}(\ms{R},\mc{O})
=
D^\infty(\ms{F}\Vert\ms{R}),
\end{equation}
as claimed.
\end{proof}

Theorem~\ref{thm:zero-rate-rf} derives the zero-rate exponent entirely
from the Hoeffding reliability function and its endpoint behavior.
In the special case of complete knowledge of the initial state,
$\ms{R}_n=\{\rho^{\ox n}\}$, Lami et al.~\cite{lami2024asymptotic}
obtained the zero-rate error exponent $D(\mc{S}\Vert\rho)$ under non-entangling operations
 by means of a generalized quantum Sanov theorem.
Our argument provides a different, Hoeffding-theoretic route:
whenever the continuity condition
\begin{equation}
\sup_{\alpha\in(0,1)}
D_{\Petz,\alpha}^{\infty}(\mc{S}\Vert\rho)
=
D^{\infty}(\mc{S}\Vert\rho)
\end{equation}
holds, Theorem~\ref{thm:zero-rate-rf}, together with the additivity of
the reverse relative entropy in this setting, reproduces the same
characterization. 

\subsubsection{Distillable Entanglement under
$\mc{F}_{\NE}$ and $\mc{F}_{\mathrm{PPT}}$}
\label{sec:distillable entanglement under NE and PPT}

We next derive the opposite endpoint of the reliability theory. The key
observation is that the roles of the two errors in binary hypothesis
testing can be exchanged exactly at every blocklength.

\begin{lemma}[Dual endpoint lemma]
\label{lem:dual-endpoint}
For every $n$, every $a,b\in[0,1]$, and arbitrary state sets
$\ms{A}_n,\ms{B}_n$,
\begin{equation}
\beta_n(\ms{B}_n,\ms{A}_n,\mathds{ALL},a)\le b
~\Longleftrightarrow~
\beta_n(\ms{A}_n,\ms{B}_n,\mathds{ALL},b)\le a.
\label{eq:error-exchange}
\end{equation}
Consequently, if
\begin{equation}
C(\ms{A},\ms{B})
:=
\sup\left\{
r:
\lim_{n\to\infty}
\beta_n(
\ms{B}_n,\ms{A}_n,\mathds{ALL},
2^{-\lfloor rn\rfloor})
=0
\right\},
\end{equation}
then
\begin{equation}
C(\ms{A},\ms{B})
=
\lim_{\varepsilon\to0}
\liminf_{n\to\infty}
-\frac{1}{n}
\log
\beta_n(
\ms{A}_n,\ms{B}_n,\mathds{ALL},\varepsilon).
\label{eq:dual-capacity}
\end{equation}
\end{lemma}

\begin{proof}
Suppose that
$\beta_n(\ms{B}_n,\ms{A}_n,\mathds{ALL},a)\le b$.
Then there exists a test $0\le M_n\le\mathds{1}$ such that
\begin{equation}
\max_{\sigma_n\in\ms{B}_n}
\tr[(\mathds{1}-M_n)\sigma_n]\le a,
\qquad
\max_{\rho_n\in\ms{A}_n}
\tr(M_n\rho_n)\le b.
\end{equation}
Replacing $M_n$ by $\mathds{1}-M_n$ exchanges the two hypotheses and
the two errors, proving Equation~\eqref{eq:error-exchange}. The
converse follows identically.

To prove Equation~\eqref{eq:dual-capacity}, denote its right-hand side
by $L$. If $r<L$, then for every sufficiently small fixed
$\varepsilon>0$ and all sufficiently large $n$,
\begin{equation}
\beta_n(
\ms{A}_n,\ms{B}_n,\mathds{ALL},\varepsilon)
\le
2^{-\lfloor rn\rfloor}.
\end{equation}
Equation~\eqref{eq:error-exchange} therefore gives
\begin{equation}
\beta_n(
\ms{B}_n,\ms{A}_n,\mathds{ALL},
2^{-\lfloor rn\rfloor})
\le\varepsilon.
\end{equation}
Since $\varepsilon>0$ is arbitrary, $r$ is achievable, and hence
$C(\ms{A},\ms{B})\ge L$.

Conversely, if $r$ is achievable, then for every fixed
$\varepsilon>0$ and all sufficiently large $n$,
\begin{equation}
\beta_n(
\ms{B}_n,\ms{A}_n,\mathds{ALL},
2^{-\lfloor rn\rfloor})
\le\varepsilon.
\end{equation}
Using Equation~\eqref{eq:error-exchange} once more,
\begin{equation}
\beta_n(
\ms{A}_n,\ms{B}_n,\mathds{ALL},\varepsilon)
\le
2^{-\lfloor rn\rfloor},
\end{equation}
and hence
\begin{equation}
\liminf_{n\to\infty}
-\frac{1}{n}
\log
\beta_n(
\ms{A}_n,\ms{B}_n,\mathds{ALL},\varepsilon)
\ge r.
\end{equation}
Letting $\varepsilon$ tend to $0$ and taking the supremum over all
achievable $r$ proves $L\ge C(\ms{A},\ms{B})$.
\end{proof}

We can now recover the distillable entanglement directly from the
opposite endpoint of the Hoeffding reliability function.

\begin{theorem}[Distillable entanglement from the dual Hoeffding endpoint]
\label{thm:distillable-entanglement-endpoint}
Let $\ms{R}=\{\ms{R}_n\}_{n\in\NN}$ be a stable sequence of convex
compact sets of bipartite quantum states, and assume $D^{\infty}(\ms F\|\ms R)>0$. Let
$\mc{O}\in\{\mc{F}_{\NE},\mc{F}_{\mathrm{PPT}}\}$ with the
corresponding free-state sequence $\ms{F}\in\{\mc{S},\mc{PPT}\}$.
Assume that
\begin{equation}
\sup_{\alpha\in(0,1)}
D_{\Petz,\alpha}^{\infty}(\ms{R}\Vert\ms{F})
=
D^{\infty}(\ms{R}\Vert\ms{F}).
\label{eq:forward-petz-continuity}
\end{equation}
Then
\begin{equation}
E_d(\ms{R},\mc{O})
=
D^{\infty}(\ms{R}\Vert\ms{F}).
\label{eq:general-distillable-entanglement}
\end{equation}
In particular,
\begin{align}
E_d(\ms{R},\mc{F}_{\NE})
=&
D^{\infty}(\ms{R}\Vert\mc{S}),
\\
E_d(\ms{R},\mc{F}_{\mathrm{PPT}})
=&
D^{\infty}(\ms{R}\Vert\mc{PPT}).
\end{align}
\end{theorem}

\begin{proof}
By Theorem~\ref{thm: error in finite blocklenth}, for every $r>0$,
\begin{equation}
\err\left(
\ms{R}_n\xrightarrow{\mc{O}}
\Phi_2^{\ox\lfloor rn\rfloor}
\right)
=
\beta_n(
\ms{F}_n,\ms{R}_n,\mathds{ALL},
2^{-\lfloor rn\rfloor}).
\label{eq:distillation-beta-capacity}
\end{equation}
Hence, by the definition of distillable entanglement and
Lemma~\ref{lem:dual-endpoint},
\begin{align}
E_d(\ms{R},\mc{O})
&=
C(\ms{R},\ms{F})
\nonumber\\
&=
\lim_{\varepsilon\to0}
\liminf_{n\to\infty}
-\frac{1}{n}
\log
\beta_n(
\ms{R}_n,\ms{F}_n,\mathds{ALL},\varepsilon).
\label{eq:capacity-fixed-threshold}
\end{align}
It therefore remains to determine the fixed-threshold exponent on the
right-hand side.

We first prove the lower bound. Fix an arbitrary
$\varepsilon\in(0,1)$ and a rate
\begin{equation}
0<s<D^\infty(\ms{F}\Vert\ms{R}).
\end{equation}
For all sufficiently large $n$,
\begin{equation}
2^{-ns}\leq\varepsilon.
\end{equation}
Since
$\beta_n(\ms{R}_n,\ms{F}_n,\mathds{ALL},\varepsilon)$ is
non-increasing in the allowed type-I error $\varepsilon$, we have
\begin{equation}
\beta_n(
\ms{R}_n,\ms{F}_n,\mathds{ALL},\varepsilon)
\leq
\beta_n(
\ms{R}_n,\ms{F}_n,\mathds{ALL},2^{-ns}).
\end{equation}
Therefore,
\begin{align}
&\liminf_{n\to\infty}
-\frac{1}{n}
\log
\beta_n(
\ms{R}_n,\ms{F}_n,\mathds{ALL},\varepsilon)
\nonumber\\
&\qquad\geq
E_{HT}(
\ms{R},\ms{F},\mathds{ALL},s)
\nonumber\\
&\qquad=
H_s^\infty(\ms{F}\Vert\ms{R}),
\label{eq:dual-fixed-epsilon-lower}
\end{align}
where the last equality follows from the composite Hoeffding theorem
applied to the exchanged hypothesis testing problem with null
hypothesis $\ms{R}$ and alternative hypothesis $\ms{F}$.

Since
$D^\infty(\ms{F}\Vert\ms{R})>0$, the preceding relation holds for all
sufficiently small $s>0$. Applying
Lemma~\ref{lem:uniform-endpoint} to the pair
$(\ms{A},\ms{B})=(\ms{F},\ms{R})$ and using
Equation~\eqref{eq:forward-petz-continuity}, we obtain
\begin{equation}
\lim_{s\to0}
H_s^\infty(\ms{F}\Vert\ms{R})
=
D^\infty(\ms{R}\Vert\ms{F}).
\label{eq:dual-Hoeffding-endpoint}
\end{equation}
Consequently, Equation~\eqref{eq:dual-fixed-epsilon-lower} yields
\begin{equation}
\liminf_{n\to\infty}
-\frac{1}{n}
\log
\beta_n(
\ms{R}_n,\ms{F}_n,\mathds{ALL},\varepsilon)
\geq
D^\infty(\ms{R}\Vert\ms{F})
\label{eq:dual-fixed-epsilon-lower-final}
\end{equation}
for every fixed $\varepsilon\in(0,1)$.

We next establish the converse bound. Let $M_n$ be an optimal test for
$\beta_n(\ms{R}_n,\ms{F}_n,\mathds{ALL},\varepsilon)$ and write
\begin{equation}
\beta_{n,\varepsilon}
:=
\beta_n(
\ms{R}_n,\ms{F}_n,\mathds{ALL},\varepsilon).
\end{equation}
For arbitrary $\rho_n\in\ms{R}_n$ and
$\sigma_n\in\ms{F}_n$, define
\begin{equation}
a_n
:=
\tr[(\mathds{1}-M_n)\rho_n],
\qquad
b_n
:=
\tr(M_n\sigma_n).
\end{equation}
By feasibility and optimality,
\begin{equation}
a_n\leq\varepsilon,
\qquad
b_n\leq\beta_{n,\varepsilon}.
\end{equation}
Applying the data-processing inequality for the Umegaki relative
entropy to the binary measurement
$\{M_n,\mathds{1}-M_n\}$ gives
\begin{align}
D(\rho_n\Vert\sigma_n)
&\geq
D\bigl(
(1-a_n,a_n)
\Vert
(b_n,1-b_n)
\bigr)
\nonumber\\
&=
(1-a_n)\log\frac{1-a_n}{b_n}
+
a_n\log\frac{a_n}{1-b_n}
\nonumber\\
&\geq
(1-a_n)\log\frac{1}{b_n}
-h_2(a_n)
\nonumber\\
&\geq
(1-\varepsilon)
\log\frac{1}{\beta_{n,\varepsilon}}
-1,
\end{align}
where $h_2$ denotes the binary entropy and we have used
$h_2(a_n)\leq1$.

Taking the infimum over
$\rho_n\in\ms{R}_n$ and $\sigma_n\in\ms{F}_n$ yields
\begin{equation}
D(\ms{R}_n\Vert\ms{F}_n)
\geq
(1-\varepsilon)
\log\frac{1}{\beta_{n,\varepsilon}}
-1.
\end{equation}
Hence
\begin{equation}
-\frac{1}{n}
\log\beta_{n,\varepsilon}
\leq
\frac{
\frac{1}{n}D(\ms{R}_n\Vert\ms{F}_n)+\frac{1}{n}
}{
1-\varepsilon
}.
\end{equation}
Taking $\limsup_{n\to\infty}$ gives
\begin{equation}
\limsup_{n\to\infty}
-\frac{1}{n}
\log
\beta_n(
\ms{R}_n,\ms{F}_n,\mathds{ALL},\varepsilon)
\leq
\frac{
D^\infty(\ms{R}\Vert\ms{F})
}{
1-\varepsilon
}.
\label{eq:dual-fixed-epsilon-upper}
\end{equation}
In particular,
\begin{equation}
\liminf_{n\to\infty}
-\frac{1}{n}
\log
\beta_n(
\ms{R}_n,\ms{F}_n,\mathds{ALL},\varepsilon)
\leq
\frac{
D^\infty(\ms{R}\Vert\ms{F})
}{
1-\varepsilon
}.
\end{equation}

Combining the lower and upper bounds, for every fixed
$\varepsilon\in(0,1)$ we have
\begin{align}
D^\infty(\ms{R}\Vert\ms{F})
&\leq
\liminf_{n\to\infty}
-\frac{1}{n}
\log
\beta_n(
\ms{R}_n,\ms{F}_n,\mathds{ALL},\varepsilon)
\nonumber\\
&\leq
\frac{
D^\infty(\ms{R}\Vert\ms{F})
}{
1-\varepsilon
}.
\label{eq:dual-fixed-threshold-sandwich}
\end{align}
Letting $\varepsilon$ tend to $0$ therefore yields
\begin{equation}
\lim_{\varepsilon\to0}
\liminf_{n\to\infty}
-\frac{1}{n}
\log
\beta_n(
\ms{R}_n,\ms{F}_n,\mathds{ALL},\varepsilon)
=
D^\infty(\ms{R}\Vert\ms{F}).
\end{equation}
Finally, substituting this identity into
Equation~\eqref{eq:capacity-fixed-threshold} gives
\begin{equation}
E_d(\ms{R},\mc{O})
=
D^\infty(\ms{R}\Vert\ms{F}),
\end{equation}
which proves Equation~\eqref{eq:general-distillable-entanglement}.
\end{proof}

The two endpoint results therefore exhibit a simple dual structure:
\begin{align}
E^{\rf}(\ms{R},\mc{O})
=
D^{\infty}(\ms{F}\Vert\ms{R}),
\end{align}
under the reverse Petz Rényi endpoint condition, i.e., $D_{\Petz,\alpha}^{\infty}(\ms{F}\Vert\ms R)
=
D^{\infty}(\ms{F}\Vert\ms R)$. 
\begin{align}
E_d(\ms{R},\mc{O})
=
D^{\infty}(\ms{R}\Vert\ms{F}),
\end{align}
under the forward Petz Rényi endpoint condition, i.e., $D_{\Petz,\alpha}^{\infty}(\ms{R}\Vert\ms{F})
=
D^{\infty}(\ms{R}\Vert\ms{F})$.
Thus, the two extremal regimes of the reliability theory naturally
produce the regularized relative entropy with the order of the two
arguments reversed.

In particular, when the initial state is completely known,
$\ms{R}_n=\{\rho^{\ox n}\}$, the continuity result of Grinko and
Lami~\cite{grinko2026samplecomplexityquantumresource} gives
\begin{equation}
\sup_{\alpha\in(0,1)}
D_{\Petz,\alpha}^{\infty}(\rho\Vert\ms{F})
=
D^{\infty}(\rho\Vert\ms{F}).
\end{equation}
Theorem~\ref{thm:distillable-entanglement-endpoint} therefore yields
\begin{align}
E_d(\rho,\mc{F}_{\NE})
&=
D^{\infty}(\rho\Vert\mc{S}),
\\
E_d(\rho,\mc{F}_{\mathrm{PPT}})
&=
D^{\infty}(\rho\Vert\mc{PPT}).
\end{align}
The first equality recovers the result of Brand{\~a}o and
Plenio~\cite{Reversible2008}, while the second gives the corresponding
characterization under PPT-preserving operations. Importantly, within
the present framework these conclusions are obtained from the
Hoeffding error-exponent analysis and its dual endpoint, rather than by
invoking a separate Stein theorem.

\subsection{Zero-Error Distillation Rate under $\mc{F}_{\NE}$ and $\mc{F}_{\mathrm{PPT}}$}
\label{sec:Zero-Error Distillation Rate under NE and PPT}
In this section, we investigate the zero-error rate of universal entanglement distillation under non-entangling and PPT-preserving operations. Specifically, we ask for the largest asymptotic rate at which maximally entangled states can be distilled perfectly and uniformly from every state in the uncertainty set. 
For a stable sequence of convex compact sets of bipartite states $\ms{R}=\{\ms{R}_n\}_{n\in\mathbb N}$ and free operations $\mathcal O\in\{\mc F_\mathrm{NE},\mc F_\mathrm{PPT}\}$, the universal zero-error distillation rate $E_{d,0}(\ms{R},\mc{O})$ is defined as the supremum of $r\geq 0$ such that for sufficiently large $n$, there exists $\Lambda_n\in\mathcal O$ such that 
$$\Lambda_n(\rho_n)
=
\Phi_2^{\otimes\lfloor rn\rfloor}, ~\forall ~\rho_n\in\ms{R}_n.$$
Thus, zero-error distillation requires exact distillation rather than merely vanishing error in the asymptotic limit for all input states in $\ms R_n$.

\begin{theorem}
\label{thm:zero-error-entanglement}
Let $\ms{R}=\{\ms{R}_n\}_{n\in\mathbb N}$ be a stable sequence of
convex compact sets of bipartite quantum states. Then the
universal zero-error entanglement distillation rates under
non-entangling and PPT-preserving operations are, respectively,
\begin{align}
    E_{d,0}(\ms{R},\mc{F}_{\NE})
    &=
    D_{\Petz,0}^{\infty}(\ms{R}\Vert\mc{S}),                                     
    \label{eq:zero-error-NE}\\
    E_{d,0}(\ms{R},\mc{F}_{\mathrm{PPT}})
    &=
    D_{\Petz,0}^{\infty}(\ms{R}\Vert\mc{PPT}).
    \label{eq:zero-error-PPT}
\end{align}

\end{theorem}

\begin{proof}
We prove the two statements simultaneously, taking
$\ms{F}_n=\mc{S}_n$ for non-entangling operations and
$\ms{F}_n=\mc{PPT}_n$ for PPT-preserving operations.

Recall that
\begin{align}
D_{\Petz,0}(\ms{R}_n\Vert\ms{F}_n)
:=&
\inf_{\substack{\rho_n\in\ms{R}_n\\
                 \sigma_n\in\ms{F}_n}}
D_{\Petz,0}(\rho_n\Vert\sigma_n)\\
=&
\inf_{\substack{\rho_n\in\ms{R}_n\\
                 \sigma_n\in\ms{F}_n}}
\left[-\log\tr(\Pi_{\rho_n}\sigma_n)\right].
\end{align}
Let
\begin{align}
    \Pi_{\ms{R}_n}
    :=
    \bigvee_{\rho_n\in\ms{R}_n}\Pi_{\rho_n}
\end{align}
be the projector onto the linear span of the supports of all states
in $\ms{R}_n$. 
We first show that
\begin{align}
D_{\Petz,0}(\ms{R}_n\Vert\ms{F}_n)
=
-\log
\max_{\sigma_n\in\ms{F}_n}
\tr(\Pi_{\ms{R}_n}\sigma_n).
\label{eq:D0-support-form}
\end{align}
Since
$\Pi_{\rho_n}\leq\Pi_{\ms{R}_n}$ for every
$\rho_n\in\ms{R}_n$,
\begin{align}
\tr(\Pi_{\rho_n}\sigma_n)
\leq
\tr(\Pi_{\ms{R}_n}\sigma_n).
\end{align}
Conversely, since the underlying Hilbert space is finite dimensional,
there exist finitely many states
$\rho_n^{(1)},\ldots,\rho_n^{(k)}\in\ms{R}_n$ whose supports span
the support of $\Pi_{\ms{R}_n}$. By convexity of $\ms{R}_n$, for any
strictly positive probabilities $p_1,\ldots,p_k$,
\begin{align}
    \rho_n:=\sum_{i=1}^k p_i\rho_n^{(i)}\in\ms{R}_n,
\end{align}
and
\begin{align}
    \Pi_{\rho_n}=\Pi_{\ms{R}_n}.
\end{align}
Therefore,
\begin{align}
\max_{\substack{\rho_n\in\ms{R}_n\\
                 \sigma_n\in\ms{F}_n}}
\tr(\Pi_{\rho_n}\sigma_n)
=
\max_{\sigma_n\in\ms{F}_n}
\tr(\Pi_{\ms{R}_n}\sigma_n),
\end{align}
which proves \eqref{eq:D0-support-form}.

We next determine exactly when $m$ two-qubit maximally entangled states can be distilled with
zero error. We claim that
$\ms{R}_n
\longrightarrow
\Phi_2^{\otimes m}
~\text{with zero error}
$
under the corresponding free operations if and only if
\begin{align}
\max_{\sigma_n\in\ms{F}_n}
\tr(\Pi_{\ms{R}_n}\sigma_n)
\leq 2^{-m}.
\label{eq:zero-error-support-condition}
\end{align}

We first prove necessity. Suppose that there exists a
non-entangling, respectively PPT-preserving, channel $\Lambda_n$
such that
\begin{align}
    \Lambda_n(\rho_n)=\Phi_2^{\otimes m},
    \qquad
    \forall\,\rho_n\in\ms{R}_n.
\label{eq:exact-distillation}
\end{align}
Since $\Phi_2^{\otimes m}$ is pure, \eqref{eq:exact-distillation}
implies
\begin{align}
\tr\!\left[
\Lambda_n^\dagger(\Phi_2^{\otimes m})\rho_n
\right]
=1,
\qquad
\forall\,\rho_n\in\ms{R}_n.
\end{align}
Moreover,
\begin{align}
0\leq
\Lambda_n^\dagger(\Phi_2^{\otimes m})
\leq \mathds{1}.
\end{align}
Hence
$\Lambda_n^\dagger(\Phi_2^{\otimes m})$ acts as the identity on the
support of every $\rho_n\in\ms{R}_n$, and therefore
\begin{align}
\Lambda_n^\dagger(\Phi_2^{\otimes m})
\geq
\Pi_{\ms{R}_n}.
\label{eq:adjoint-support}
\end{align}

For every $\sigma_n\in\ms{F}_n$, the state
$\Lambda_n(\sigma_n)$ is separable in the non-entangling case and
PPT in the PPT-preserving case. Every separable or PPT state $\tau$
on a $2^m\times2^m$ system satisfies
\begin{align}
    \tr(\Phi_2^{\otimes m}\tau)\leq 2^{-m}.
\label{eq:max-overlap-free}
\end{align}

Combining this with \eqref{eq:adjoint-support}, we obtain
\begin{align}
\tr(\Pi_{\ms{R}_n}\sigma_n)
&\leq
\tr\!\left[
\Lambda_n^\dagger(\Phi_2^{\otimes m})\sigma_n
\right]\nonumber\\
&=
\tr\!\left[
\Phi_2^{\otimes m}\Lambda_n(\sigma_n)
\right]
\leq
2^{-m},
\end{align}
for every $\sigma_n\in\ms{F}_n$. Thus
\eqref{eq:zero-error-support-condition} is necessary.

Conversely, suppose that
\eqref{eq:zero-error-support-condition} holds. Define
\begin{align}
\Lambda_n(X)
=
\tr(\Pi_{\ms{R}_n}X)\Phi_2^{\otimes m}
+
\tr\!\left[
(\mathds{1}-\Pi_{\ms{R}_n})X
\right]
\frac{\mathds{1}-\Phi_2^{\otimes m}}
     {2^{2m}-1}.
\label{eq:zero-error-channel}
\end{align}
This is a quantum channel. Since
$\Pi_{\rho_n}\leq\Pi_{\ms{R}_n}$ for every
$\rho_n\in\ms{R}_n$,
\begin{align}
\tr(\Pi_{\ms{R}_n}\rho_n)=1,
\end{align}
and hence
\begin{align}
    \Lambda_n(\rho_n)=\Phi_2^{\otimes m},
    \qquad
    \forall\,\rho_n\in\ms{R}_n.
\end{align}

It remains only to verify that \eqref{eq:zero-error-channel} is a
free operation. For every $\sigma_n\in\ms{F}_n$, write
\begin{align}
p:=\tr(\Pi_{\ms{R}_n}\sigma_n).
\end{align}
By assumption, $p\leq2^{-m}$, and
\begin{align}
\Lambda_n(\sigma_n)
=
p\Phi_2^{\otimes m}
+
(1-p)
\frac{\mathds{1}-\Phi_2^{\otimes m}}{2^{2m}-1}.
\label{eq:isotropic-output}
\end{align}
The isotropic state in \eqref{eq:isotropic-output} is separable
whenever $0\leq p\leq2^{-m}$.  Hence
\eqref{eq:zero-error-channel} is non-entangling in the first case and PPT-preserving
in the second case. This proves the sufficiency of
\eqref{eq:zero-error-support-condition}.

Combining \eqref{eq:D0-support-form} and
\eqref{eq:zero-error-support-condition}, we conclude that exact
distillation of $m$ two-qubit maximally entangled states is possible if and only if
\begin{align}
    m
    \leq
    D_{\Petz,0}(\ms{R}_n\Vert\ms{F}_n).
\label{eq:exact-m-D0}
\end{align}
Consequently, the maximal zero-error yield at blocklength $n$ is
\begin{align}
\label{eq:one-shot-zero-error-yield}
\left\lfloor
D_{\Petz,0}(\ms{R}_n\Vert\ms{F}_n)
\right\rfloor .
\end{align}

It remains to pass to the asymptotic rate. By stability,
\begin{align}
\ms{R}_n\otimes\ms{R}_k
&\subseteq\ms{R}_{n+k},\\
\ms{F}_n\otimes\ms{F}_k
&\subseteq\ms{F}_{n+k}.
\end{align}
Since the order-zero Petz R\'enyi relative entropy is additive under
tensor products,
\begin{align}
D_{\Petz,0}
(\rho_n\otimes\rho_k\Vert\sigma_n\otimes\sigma_k)
=
D_{\Petz,0}(\rho_n\Vert\sigma_n)
+
D_{\Petz,0}(\rho_k\Vert\sigma_k),
\end{align}
we have
\begin{align}
&D_{\Petz,0}
(\ms{R}_{n+k}\Vert\ms{F}_{n+k})
\nonumber\\
&\qquad\leq
D_{\Petz,0}(\ms{R}_n\Vert\ms{F}_n)
+
D_{\Petz,0}(\ms{R}_k\Vert\ms{F}_k).
\label{eq:D0-subadditive}
\end{align}
Thus Fekete's lemma gives
\begin{align}
D_{\Petz,0}^{\infty}(\ms{R}\Vert\ms{F})
&=
\lim_{n\to\infty}
\frac{1}{n}
D_{\Petz,0}(\ms{R}_n\Vert\ms{F}_n)
\nonumber\\
&=
\inf_{n\geq1}
\frac{1}{n}
D_{\Petz,0}(\ms{R}_n\Vert\ms{F}_n).
\end{align}

Finally, dividing \eqref{eq:one-shot-zero-error-yield} by $n$ and
letting $n\to\infty$, the integer rounding contributes at most
$1/n$, and hence
\begin{align}
E_{d,0}(\ms{R},\mc{F}_{\NE})
&=
D_{\Petz,0}^{\infty}(\ms{R}\Vert\mc{S}),\\
E_{d,0}(\ms{R},\mc{F}_{\mathrm{PPT}})
&=
D_{\Petz,0}^{\infty}(\ms{R}\Vert\mc{PPT}).
\end{align}
This completes the proof.
\end{proof}

\subsection{Zero-Rate Error Exponent and Distillable Entanglement
under Dually Free Operations}
\label{sec:Exact Zero‑Rate Error Exponent and Distillable Entanglement under Dually Free Operations}

In this section, we investigate the zero-rate error exponent and the distillable
entanglement of universal entanglement distillation under dually
non-entangling and dually PPT-preserving operations. In contrast to the
unrestricted-measurement setting considered above, an explicit
Hoeffding-type characterization of the reliability function under
restricted measurements is not presently available. Nevertheless, the
two endpoint quantities can be determined exactly by combining our
finite-blocklength bounds with the quantum Stein's lemma for restricted
measurements established by Brand{\~a}o et
al.~\cite[Theorem~16]{brandao2020adversarial}.

We first recall the restricted-measurement Stein lemma in the notation
used in this paper. Let
$\mathds{M}=\{\mathds{M}_n\}_{n\in\NN}$ be a sequence of measurement
classes and let
$\ms{A}=\{\ms{A}_n\}_{n\in\NN}$ and
$\ms{B}=\{\ms{B}_n\}_{n\in\NN}$ be sequences of state sets. 
Following~\cite[Definition~4]{brandao2020adversarial}, a pair
$(\mathds{M},\ms{A})$ is said to be \emph{compatible} if
$\ms{A}$ is closed under permutations and convex combinations, and if
conditioning a state in $\ms{A}_{n+k}$ on any outcome of a measurement
from $\mathds{M}_k$ leaves, whenever the outcome has nonzero
probability, a state in $\ms{A}_n$. More explicitly, for every
$\rho_{n+k}\in\ms{A}_{n+k}$,
$\mc{M}_k=\{M_j^k\}_j\in\mathds{M}_k$, and outcome $j$, the
unnormalized residual operator
\begin{align}
\widetilde{\rho}_n
:=
\tr_{n+1,\ldots,n+k}
\left[
\rho_{n+k}
\left(
\mathds{1}_n\ox M_j^k
\right)
\right]
\end{align}
satisfies
\begin{equation}
\frac{\widetilde{\rho}_n}{\tr\widetilde{\rho}_n}
\in\ms{A}_n
\end{equation}
whenever $\tr\widetilde{\rho}_n>0$.

A measurement family
$\mathds{M}=\{\mathds{M}_n\}_{n\in\NN}$ is called
\emph{self-consistent} if it is closed under tensor products, tensoring
with identity measurements, permutations of the systems, and finite
labelled mixtures~\cite[Definition~5]{brandao2020adversarial}.
The quantum Stein's lemma for restricted measurements can then be
stated as follows.

\begin{theorem}[Quantum Stein's lemma for restricted measurements
{\cite[Theorem~16]{brandao2020adversarial}}]
\label{thm:restricted-Stein}
Let $\ms{A}=\{\ms{A}_n\}_{n\in\NN}$ and
$\ms{B}=\{\ms{B}_n\}_{n\in\NN}$ be closed families of quantum states,
and let $\mathds{M}=\{\mathds{M}_n\}_{n\in\NN}$ be a self-consistent
family of measurements. If both $(\mathds{M},\ms{A})$ and
$(\mathds{M},\ms{B})$ are compatible pairs, then
\begin{align}
\lim_{\varepsilon\to0}
\lim_{n\to\infty}
-\frac{1}{n}
\log
\beta_n(
\ms{A}_n,\ms{B}_n,\mathds{M},\varepsilon)
=
D^{\mathds{M},\infty}(\ms{A}\Vert\ms{B}).
\label{eq:restricted-Stein}
\end{align}
\end{theorem}

The measurement families $\mathds{SEP}$ and $\mathds{PPT}$ are
self-consistent. Moreover, $(\mathds{SEP},\mc{S})$ and
$(\mathds{PPT},\mc{PPT})$ are compatible pairs. Hence, provided that
the initial-state sequence $\ms{R}$ is compatible with the
corresponding measurement family, all the assumptions of
Theorem~\ref{thm:restricted-Stein} are satisfied: compactness of
$\ms{R}_n$ implies closedness, while the separable and PPT state sets
are themselves closed.

We can now characterize both endpoint quantities under dually free
operations.

\begin{theorem}
\label{them: zero-rate and capacity}
Let $\ms{R}=\{\ms{R}_n\}_{n\in\NN}$ be a stable sequence of convex
compact sets of quantum states, with
$\ms{R}_n\subseteq\density(\mathcal{H}_{AB}^{\ox n})$, and suppose
that each $\ms{R}_n$ is closed under permutations of the systems.

If $(\mathds{SEP},\ms{R})$ is a compatible pair, then under dually
non-entangling operations $\mc{F}_{\mathrm{DNE}}$, the zero-rate error
exponent and the distillable entanglement are given by
\begin{align}
E^{\rf}(\ms{R},\mc{F}_{\mathrm{DNE}})
&=
D^{\mathds{SEP},\infty}(\mc{S}\Vert\ms{R}),
\label{eq:DNE-zero-rate}\\
E_d(\ms{R},\mc{F}_{\mathrm{DNE}})
&=
D^{\mathds{SEP},\infty}(\ms{R}\Vert\mc{S}).
\label{eq:DNE-capacity}
\end{align}

If $(\mathds{PPT},\ms{R})$ is a compatible pair, then under dually
PPT-preserving operations $\mc{F}_{\mathrm{DPPT}}$, the corresponding
quantities are
\begin{align}
E^{\rf}(\ms{R},\mc{F}_{\mathrm{DPPT}})
&=
D^{\mathds{PPT},\infty}(\mc{PPT}\Vert\ms{R}),
\label{eq:DPPT-zero-rate}\\
E_d(\ms{R},\mc{F}_{\mathrm{DPPT}})
&=
D^{\mathds{PPT},\infty}(\ms{R}\Vert\mc{PPT}).
\label{eq:DPPT-capacity}
\end{align}
\end{theorem}

\begin{proof}
We prove all four statements simultaneously. Let
$(\mc{O},\ms{F},\mathds{M})$ denote either
\begin{equation}
(\mc{F}_{\mathrm{DNE}},\mc{S},\mathds{SEP})
\end{equation}
or
\begin{equation}
(\mc{F}_{\mathrm{DPPT}},\mc{PPT},\mathds{PPT}).
\end{equation}

We first consider the zero-rate error exponent. For a fixed target
$\Phi_2^{\ox m}$, Theorem~
\ref{thm: error in finite blocklenth under DNE and DPPT} gives
\begin{align}
\beta_n(
\ms{F}_n,\ms{R}_n,\mathds{M},
2\cdot2^{-m})
\le&
\err\left(
\ms{R}_n\xrightarrow{\mc{O}}\Phi_2^{\ox m}
\right)
\nonumber\\
\le&
\beta_n(
\ms{F}_n,\ms{R}_n,\mathds{M},
2^{-m}).
\label{eq:dual-fixed-m-sandwich}
\end{align}
Taking $-\frac{1}{n}\log$ reverses the inequalities. Therefore,
\begin{align}
&\liminf_{n\to\infty}
-\frac{1}{n}
\log
\beta_n(
\ms{F}_n,\ms{R}_n,\mathds{M},
2^{-m})
\nonumber\\
&\qquad\le
\liminf_{n\to\infty}
-\frac{1}{n}
\log
\err\left(
\ms{R}_n\xrightarrow{\mc{O}}\Phi_2^{\ox m}
\right)
\nonumber\\
&\qquad\le
\limsup_{n\to\infty}
-\frac{1}{n}
\log
\err\left(
\ms{R}_n\xrightarrow{\mc{O}}\Phi_2^{\ox m}
\right)
\nonumber\\
&\qquad\le
\limsup_{n\to\infty}
-\frac{1}{n}
\log
\beta_n(
\ms{F}_n,\ms{R}_n,\mathds{M},
2\cdot2^{-m}).
\label{eq:dual-zero-rate-sandwich}
\end{align}
Both $2^{-m}$ and $2\cdot2^{-m}$ tend to zero as $m\to\infty$.
Applying the restricted-measurement Stein lemma,
Theorem~\ref{thm:restricted-Stein}, to the pair
$(\ms{F},\ms{R})$ therefore gives
the zero-rate exponent
\begin{equation}
E^{\rf}(\ms{R},\mc{O})
=
D^{\mathds{M},\infty}(\ms{F}\Vert\ms{R}).
\label{eq:dual-zero-rate-general}
\end{equation}

We next consider the distillable entanglement. The two errors in the
restricted hypothesis testing problem can be exchanged exactly. Indeed,
since both $\mathds{SEP}$ and $\mathds{PPT}$ are closed under
interchanging the two outcomes of a binary measurement, for every
$a,b\in[0,1]$,
\begin{equation}
\beta_n(
\ms{F}_n,\ms{R}_n,\mathds{M},a)
\le b
\quad\Longleftrightarrow\quad
\beta_n(
\ms{R}_n,\ms{F}_n,\mathds{M},b)
\le a.
\label{eq:restricted-error-exchange}
\end{equation}
Consequently, the largest rate for which
\begin{equation}
\beta_n(
\ms{F}_n,\ms{R}_n,\mathds{M},
2^{-\lfloor rn\rfloor})
\longrightarrow0
\end{equation}
is
\begin{align}
\lim_{\varepsilon\to0}
\lim_{n\to\infty}
-\frac{1}{n}
\log
\beta_n(
\ms{R}_n,\ms{F}_n,\mathds{M},\varepsilon).
\label{eq:restricted-dual-Stein-endpoint}
\end{align}
The constant prefactor appearing in the lower finite-blocklength bound
does not affect this threshold rate: if
\begin{equation}
\beta_n(
\ms{F}_n,\ms{R}_n,\mathds{M},
2\cdot2^{-\lfloor rn\rfloor})
\longrightarrow0,
\end{equation}
then, for every $\delta>0$ and all sufficiently large $n$,
\begin{equation}
2\cdot2^{-\lfloor rn\rfloor}
\le
2^{-\lfloor (r-\delta)n\rfloor},
\end{equation}
and hence the rate $r-\delta$ is achievable in the corresponding
hypothesis testing problem. Letting $\delta$ tend to $0$ shows that the
finite-blocklength sandwich has the same threshold rate as the
hypothesis testing problem.

Applying Theorem~\ref{thm:restricted-Stein} to the exchanged pair
$(\ms{R},\ms{F})$ now yields
\begin{equation}
E_d(\ms{R},\mc{O})
=
D^{\mathds{M},\infty}(\ms{R}\Vert\ms{F}).
\label{eq:dual-capacity-general}
\end{equation}
Substituting
$(\mc{O},\ms{F},\mathds{M})
=
(\mc{F}_{\mathrm{DNE}},\mc{S},\mathds{SEP})$
and
$(\mc{F}_{\mathrm{DPPT}},\mc{PPT},\mathds{PPT})$
into Equations~\eqref{eq:dual-zero-rate-general} and
\eqref{eq:dual-capacity-general} proves the theorem.
\end{proof}

Theorem~\ref{them: zero-rate and capacity} reveals the same dual
structure as in the unrestricted-measurement setting: the zero-rate
error exponent is governed by the regularized restricted relative
entropy with the free-state sequence in the first argument, whereas the
distillable entanglement is governed by the same quantity with the order
of the two arguments reversed. Unlike the preceding unrestricted
analysis, however, no Rényi endpoint continuity assumption is required
here, since the restricted-measurement Stein lemma directly
characterizes the fixed-threshold endpoint.

In particular, when the initial state is completely known,
$\ms{R}=\{\rho^{\ox n}\}_{n\in\NN}$, we recover the result of Lami and
Regula~\cite{Dually} that
\begin{equation}
E_d(\rho,\mc{F}_{\mathrm{DNE}})
=
D^{\mathds{SEP},\infty}(\rho\Vert\mc{S}).
\end{equation}
The same argument additionally gives
\begin{equation}
E_d(\rho,\mc{F}_{\mathrm{DPPT}})
=
D^{\mathds{PPT},\infty}(\rho\Vert\mc{PPT})
\end{equation}
for dually PPT-preserving operations.

\subsection{The Strong Converse Exponent under
$\mc{F}_{\NE}$ and $\mc{F}_{\mathrm{PPT}}$}
\label{sec:strong converse exponent}

In this section, we continue to examine the asymptotic behavior of
universal entanglement distillation, now for rates above the regularized
relative entropy threshold. For convenience of presentation, in this
section we adopt the convention that the free operations $\mc{O}$
denote the non-entangling operations $\mc{F}_{\NE}$ and the
PPT-preserving operations $\mc{F}_{\mathrm{PPT}}$, and the free states
$\ms{F}$ refer to the separable states $\mc{S}$ and the PPT states
$\mc{PPT}$. Moreover, in the corresponding statements,
$\mc{F}_{\NE}$ is paired with $\mc{S}$, and
$\mc{F}_{\mathrm{PPT}}$ is paired with $\mc{PPT}$.

Let $\ms{R}=\{\ms{R}_n\}_{n\in\NN}$ be a stable sequence of convex
compact sets of quantum states with each
$\ms{R}_n\subseteq\density(\mathcal{H}_{AB}^{\ox n})$. We consider
rates
\begin{equation}
r>D^{\infty}(\ms{R}\Vert\ms{F}),
\end{equation}
for which the distillation fidelity converges exponentially to zero,
or equivalently, the distillation error converges to one, as shown in
Proposition~\ref{pro: strong converse} below. To precisely describe
the asymptotic behavior in this regime, we define the fidelity of
universal entanglement distillation as
\begin{align}
\fid\left(
\ms{R}_n\xrightarrow{\mc{O}}
\Phi_2^{\otimes\lfloor rn\rfloor}
\right)
:=
\max_{\Lambda_n\in\mc{O}}
\min_{\rho_n\in\ms{R}_n}
F\left(
\Lambda_n(\rho_n),
\Phi_2^{\otimes\lfloor rn\rfloor}
\right).
\label{eq:universal-entanglement-fidelity}
\end{align}
The strong converse exponent for universal entanglement distillation
is defined as the convergence exponent of the fidelity,
\begin{align}
E^{\sce}(\ms{R},\mc{O},r)
:=
\liminf_{n\to\infty}
-\frac{1}{n}
\log
\fid\left(
\ms{R}_n\xrightarrow{\mc{O}}
\Phi_2^{\otimes\lfloor rn\rfloor}
\right).
\label{eq:strong-converse-exponent-definition}
\end{align}

A direct observation from the definitions is that the fidelity and the
error of universal entanglement distillation in the black-box setting
satisfy
\begin{align}
\fid\left(
\ms{R}_n\xrightarrow{\mc{O}}
\Phi_2^{\otimes\lfloor rn\rfloor}
\right)
=
1-
\err\left(
\ms{R}_n\xrightarrow{\mc{O}}
\Phi_2^{\otimes\lfloor rn\rfloor}
\right).
\label{eq:strong-converse-fidelity-error}
\end{align}

Similar to the approach used in proving the reliability function, it
is useful to establish a connection with hypothesis testing. We
therefore consider \emph{detection probability}~\cite{hayashi2025general} of $\ms{R}_n$ under the
constraint imposed by the free-state set $\ms{F}_n$ at level
$2^{-\lfloor rn\rfloor}$:
\begin{align}
&1-
\beta_n\left(
\ms{F}_n,\ms{R}_n,\mathds{ALL},
2^{-\lfloor rn\rfloor}
\right)
\nonumber\\
&=
\max_{\substack{
0\leq M_n\leq\mathds{1}\\
\max_{\sigma_n\in\ms{F}_n}
\tr[(\mathds{1}-M_n)\sigma_n]
\leq2^{-\lfloor rn\rfloor}
}}
\min_{\rho_n\in\ms{R}_n}
\left(
1-\tr(M_n\rho_n)
\right).
\label{eq:strong-converse-detection-probability}
\end{align}

 From Theorem~\ref{thm: error in finite blocklenth}, we immediately derive the following corollary.
\begin{corollary}
\label{cor:strong-converse-fidelity-beta}
Let $\ms{R}_n$ be a convex compact set of quantum states on Hilbert space $\mathcal{H}_{AB}^{\ox n}$. Let $r>0$ be a real number, and the distillation protocol is restricted to free operations $\mathcal{O}$. Then the fidelity of universal entanglement distillation is related to the detection probability of hypothesis testing
    \begin{align}
         \fid\left(\ms{R}_n\xrightarrow{\mathcal{O}} \Phi_2^{\otimes \lfloor rn \rfloor}\right)=1-\beta_n\left(\ms{F}_n,\ms{R}_n,\mathds{ALL},
2^{-\lfloor rn\rfloor}\right).
    \end{align}
\end{corollary}

Applying the strong converse exponent bound for composite correlated
hypotheses established in
\cite[Theorem~2]{fang2025errorexponentsquantumstate}, together with a
careful treatment of the floor function, yields the following result
for universal entanglement distillation.

\begin{proposition}
    \label{pro: strong converse}
    Let $\ms{R}=\{\ms{R}_n\}_{n\in\NN}$ be a stable sequence of convex compact sets of quantum states with each $\ms{R}_n\subseteq\density(\mathcal{H}_{AB}^{\ox n})$. Let $r>D^{\infty}(\ms{R}\|\ms{F})$ be a real number, and the distillation protocol is restricted to free operations $\mathcal{O}$. Then the strong converse exponent for universal entanglement distillation is lower bounded by the regularized quantum Hoeffding anti-divergences
    \begin{align}
    \label{eq:strong-converse-anti-Hoeffding-bound}
        E^{\sce}(\ms{R},\mathcal{O},r)\geq \widetilde{H}_{r}^{\infty}(\ms{R}\|\ms{F}).
    \end{align}
    Moreover,
\begin{equation}
\widetilde{H}_{r}^{\infty}
(\ms{R}\Vert\ms{F})>0.
\label{eq:strict-positive-anti-Hoeffding}
\end{equation}
Consequently, the distillation fidelity converges exponentially to
zero, and the distillation error converges to one.
\end{proposition}

\begin{proof}
By Corollary~\ref{cor:strong-converse-fidelity-beta},
\begin{align}
E^{\sce}&(\ms{R},\mc{O},r)\nonumber\\
&=
\liminf_{n\to\infty}
-\frac{1}{n}
\log
\left[
1-
\beta_n\left(
\ms{F}_n,\ms{R}_n,\mathds{ALL},
2^{-\lfloor rn\rfloor}
\right)
\right].
\label{eq:strong-converse-beta-form}
\end{align}

We first prove
Equation~\eqref{eq:strong-converse-anti-Hoeffding-bound}. Fix
\begin{equation}
0<\delta<
r-D^{\infty}(\ms{R}\Vert\ms{F}).
\label{eq:strong-converse-delta}
\end{equation}
For all sufficiently large $n$, we have
\begin{equation}
\lfloor rn\rfloor
\geq
nr-1
\geq
n(r-\delta),
\end{equation}
and consequently
\begin{equation}
2^{-\lfloor rn\rfloor}
\leq
2^{-n(r-\delta)}.
\label{eq:strong-converse-floor-comparison}
\end{equation}
Since
\begin{equation}
\varepsilon
\longmapsto
\beta_n(
\ms{F}_n,\ms{R}_n,\mathds{ALL},\varepsilon)
\end{equation}
is non-increasing, its complement is non-decreasing. It follows from
Equation~\eqref{eq:strong-converse-floor-comparison} that
\begin{align}
&1-
\beta_n\left(
\ms{F}_n,\ms{R}_n,\mathds{ALL},
2^{-\lfloor rn\rfloor}
\right)
\nonumber\\
&\qquad\leq
1-
\beta_n\left(
\ms{F}_n,\ms{R}_n,\mathds{ALL},
2^{-n(r-\delta)}
\right).
\label{eq:detection-probability-comparison}
\end{align}
Taking $-\frac1n\log$ and then the limit inferior gives
\begin{align}
E^{\sce}&(\ms{R},\mc{O},r)\\
&\geq
\liminf_{n\to\infty}
-\frac1n
\log
\left[
1-
\beta_n\left(
\ms{F}_n,\ms{R}_n,\mathds{ALL},
2^{-n(r-\delta)}
\right)
\right].
\label{eq:strong-converse-rate-shift}
\end{align}
Since
\begin{equation}
r-\delta>D^{\infty}(\ms{R}\Vert\ms{F}),
\end{equation}
the strong converse exponent bound for composite correlated
hypotheses~\cite[Theorem~2]{fang2025errorexponentsquantumstate}
applies at rate $r-\delta$, yielding
\begin{equation}
E^{\sce}(\ms{R},\mc{O},r)
\geq
\widetilde{H}_{r-\delta}^{\infty}
(\ms{R}\Vert\ms{F}).
\label{eq:strong-converse-r-minus-delta}
\end{equation}

We next verify that the regularized quantum Hoeffding
anti-divergence is continuous with respect to the rate parameter. For
every $\alpha>1$ and every pair of states
$\rho_n\in\ms{R}_n$ and $\sigma_n\in\ms{F}_n$, we have
\begin{align}
&\frac{\alpha-1}{\alpha}
\left(
n(r-\delta)
-
D_{\Sand,\alpha}(\rho_n\Vert\sigma_n)
\right)
\nonumber\\
&=
\frac{\alpha-1}{\alpha}
\left(
nr-
D_{\Sand,\alpha}(\rho_n\Vert\sigma_n)
\right)
-
n\delta\frac{\alpha-1}{\alpha}.
\end{align}
Since
\begin{equation}
0<
\frac{\alpha-1}{\alpha}
<1,
\end{equation}
taking the supremum over $\alpha>1$ gives
\begin{align}
\widetilde{H}_{n,r}(\rho_n\Vert\sigma_n)-n\delta
&\leq
\widetilde{H}_{n,r-\delta}(\rho_n\Vert\sigma_n)
\nonumber\\
&\leq
\widetilde{H}_{n,r}(\rho_n\Vert\sigma_n).
\end{align}
Taking the supremum over
$\rho_n\in\ms{R}_n$ and $\sigma_n\in\ms{F}_n$, we obtain
\begin{align}
\widetilde{H}_{n,r}(\ms{R}_n\Vert\ms{F}_n)-n\delta
&\leq
\widetilde{H}_{n,r-\delta}
(\ms{R}_n\Vert\ms{F}_n)
\nonumber\\
&\leq
\widetilde{H}_{n,r}(\ms{R}_n\Vert\ms{F}_n).
\end{align}
Dividing by $n$ and taking the regularized limit yields
\begin{align}
\widetilde{H}_{r}^{\infty}(\ms{R}\Vert\ms{F})-\delta
&\leq
\widetilde{H}_{r-\delta}^{\infty}
(\ms{R}\Vert\ms{F})
\nonumber\\
&\leq
\widetilde{H}_{r}^{\infty}(\ms{R}\Vert\ms{F}).
\label{eq:anti-Hoeffding-continuity-bound}
\end{align}
Thus
\begin{equation}
\lim_{\delta\to0}
\widetilde{H}_{r-\delta}^{\infty}
(\ms{R}\Vert\ms{F})
=
\widetilde{H}_{r}^{\infty}
(\ms{R}\Vert\ms{F}).
\label{eq:anti-Hoeffding-left-continuity}
\end{equation}
Letting $\delta$ tend to $0$ in
Equation~\eqref{eq:strong-converse-r-minus-delta} proves
Equation~\eqref{eq:strong-converse-anti-Hoeffding-bound}.

It remains to prove the strict positivity in
Equation~\eqref{eq:strict-positive-anti-Hoeffding}. The assumptions
of the continuity result
\cite[Lemma~27]{fang2024generalizedquantumaep} are satisfied by the
sequences $\ms{R}$ and $\ms{F}$. Indeed, both sequences are convex,
compact, and stable under tensor products. It remains only to verify
the required linear upper bound on the max-relative entropy.

Let
\begin{equation}
d_{AB}:=\dim\mathcal{H}_{AB},
\qquad
\tau_{AB}:=\frac{\mathds{1}_{AB}}{d_{AB}}.
\end{equation}
The maximally mixed state is separable and hence PPT, so
\begin{equation}
\tau_{AB}^{\ox n}
\in
\mc{S}_n
\subseteq
\mc{PPT}_n.
\end{equation}
Moreover, for every $\rho_n\in\ms{R}_n$,
\begin{equation}
\rho_n
\leq
\mathds{1}_{AB}^{\ox n}
=
d_{AB}^{\,n}\tau_{AB}^{\ox n}.
\end{equation}
It follows that
\begin{equation}
D_{\max}
\left(
\rho_n\Vert\tau_{AB}^{\ox n}
\right)
\leq
n\log d_{AB},
\end{equation}
and consequently
\begin{equation}
D_{\max}(\ms{R}_n\Vert\ms{F}_n)
\leq
n\log d_{AB}.
\label{eq:linear-Dmax-bound}
\end{equation}
Therefore,
\cite[Lemma~27]{fang2024generalizedquantumaep} gives
\begin{equation}
\inf_{\alpha>1}
D_{\Sand,\alpha}^{\infty}
(\ms{R}\Vert\ms{F})
=
D^{\infty}(\ms{R}\Vert\ms{F}).
\label{eq:sandwiched-renyi-continuity-at-one}
\end{equation}

On the other hand, the equivalence between the two regularized
quantum Hoeffding anti-divergences established in
\cite{fang2025errorexponentsquantumstate} gives
\begin{equation}
\widetilde{H}_{r}^{\infty}
(\ms{R}\Vert\ms{F})
=
\sup_{\alpha>1}
\frac{\alpha-1}{\alpha}
\left(
r-
D_{\Sand,\alpha}^{\infty}
(\ms{R}\Vert\ms{F})
\right).
\label{eq:regularized-anti-Hoeffding-expression}
\end{equation}
Since
\begin{equation}
r>D^{\infty}(\ms{R}\Vert\ms{F}),
\end{equation}
Equation~\eqref{eq:sandwiched-renyi-continuity-at-one} ensures that
there exists an $\alpha_r>1$ such that
\begin{equation}
D_{\Sand,\alpha_r}^{\infty}
(\ms{R}\Vert\ms{F})
<r.
\end{equation}
Substituting this value into
Equation~\eqref{eq:regularized-anti-Hoeffding-expression}, we obtain
\begin{align}
\widetilde{H}_{r}^{\infty}
(\ms{R}\Vert\ms{F})
&\geq
\frac{\alpha_r-1}{\alpha_r}
\left(
r-
D_{\Sand,\alpha_r}^{\infty}
(\ms{R}\Vert\ms{F})
\right)>0,
\end{align}
which proves
Equation~\eqref{eq:strict-positive-anti-Hoeffding}.

Combining the two parts of the proof, we conclude that
\begin{equation}
E^{\sce}(\ms{R},\mc{O},r)
\geq
\widetilde{H}_{r}^{\infty}
(\ms{R}\Vert\ms{F})
>0.
\end{equation}
Consequently, there exists a constant $c>0$ such that, for all
sufficiently large $n$,
\begin{equation}
\fid\left(
\ms{R}_n\xrightarrow{\mc{O}}
\Phi_2^{\otimes\lfloor rn\rfloor}
\right)
\leq
2^{-cn}.
\end{equation}
Therefore, the distillation fidelity converges exponentially to zero,
and by Equation~\eqref{eq:strong-converse-fidelity-error}, the
distillation error converges to one.
\end{proof}

If, in addition, the endpoint condition in
Theorem~\ref{thm:distillable-entanglement-endpoint} holds, then
\begin{equation}
E_d(\ms{R},\mc{O})
=
D^{\infty}(\ms{R}\Vert\ms{F}).
\end{equation}
In this case, the rate condition in
Proposition~\ref{pro: strong converse} is equivalently
$r>E_d(\ms{R},\mc{O})$, and hence the proposition describes the
strong-converse regime above the distillable entanglement.

\subsection{Concrete Physical Illustrations of Noise}
\label{sec: Concrete Physical Illustrations of Noise}
This paper investigates universal entanglement distillation, that is, the scenario where the initial state to be extracted is not from a precisely known sequence $\{\rho^{\ox n}\}_{n\in\mathbb{N}}$, but rather from a black-box sequence $\ms{R}=\{\ms{R}_n\}_{n\in\mathbb{N}}$. The distillation process must be effective for any unknown sequence of states arbitrarily drawn from $\ms{R}$. This section provides a concrete physical interpretation of this black-box setting.

In both the finite-blocklength analysis (Theorems~\ref{thm: error in finite blocklenth}~and~\ref{thm: error in finite blocklenth under DNE and DPPT}) and the asymptotic analysis (Theorems~\ref{them: exponents}~and~\ref{them:rf,NE,PPT}) of universal entanglement distillation, the black-box sequence $\ms{R}$ is assumed to be very general: it need only satisfy \emph{convexity}, \emph{compactness}, and \emph{stability under tensor products}. 
Additionally, to present the zero-rate error exponent and distillable entanglement for universal entanglement distillation under dually free operations (Theorem~\ref{them: zero-rate and capacity}), it is required that $\ms{R}$ is \emph{closed under permutations of the systems} and is \emph{compatible with separable measurements and PPT measurements}.
Working within the framework of noise modeling, we illustrate that the derived $\ms{R}$ naturally satisfies all the required conditions mentioned above. 

We begin by introducing the Heisenberg–Weyl operators. 
Let $\zeta = e^{2\pi i/d}$. On the computational basis $\{|j\rangle\}_{j=0}^{d-1}$, define the qudit generalizations of the Pauli-$X$ and Pauli-$Z$ operators,
\[
X|j\rangle = |j +1(\mathrm{mod}~d)\rangle,\quad Z|j\rangle = \zeta^{j}|j\rangle.
\]
The Heisenberg–Weyl operators are then defined as
\[
W_{x,z} = X^{x}Z^{z},\quad x,z\in\{0,1,\dots,d-1\}.
\]

We present five noise models, each formulated as a quantum channel acting on a $d$-dimensional Hilbert space $\mc{H}$. For any input state $\omega\in \ms{D}(\mc{H})$:
\begin{enumerate}
    \item \emph{Shift Dit-Flip Channel}:
    $$\mathcal{B}_p(\omega) = (1-p)\,\omega + p\,X\omega X^{\dagger}.$$
    \item \emph{Generalized Phase-Flip Channel}:
    $$\mathcal{P}_p(\omega) = (1-p)\,\omega + p\,Z\omega Z^{\dagger}.$$
    \item \emph{Depolarizing Channel}:
    $$\mathcal{D}_p(\omega) = (1-p)\,\omega + p\,\frac{\mathds{1}_d}{d}.$$
    \item \emph{Erasure Channel}:
    $$\mathcal{E}_p(\omega) = (1-p)\,\omega + p\,|e\rangle\langle e|,$$
    where $|e\rangle$ denotes an erasure state orthogonal to the system's Hilbert space. At this point, the Hilbert space can be regarded as $\mathcal{H} \oplus \mathbb{C}|e\rangle$.
    \item \emph{Mixed Unitary Channel}:
    $$\mathcal{R}(\omega)=\sum_{k=1}^N p_k\,U_k\omega U_k^{\dagger},$$
    where $p_k\ge0,\;\sum_k p_k=1$, and $\{U_k\}_{k=1}^{N}$ are $N$ fixed unitary operators.
\end{enumerate}

Among the first four noise channels, each has a parameter $p$ that characterizes the level of noise contaminating the prepared state. Let the noise threshold, i.e., the maximum tolerable noise level, be denoted by $p_{\mathrm{th}}\in(0,1)$. We collectively denote the first four channels by $\mathcal{N}_p$ with $p\in[0,p_{\mathrm{th}}]$.

At the beginning of entanglement distillation, suppose the goal is to prepare an initial state $\rho_{AB}$. However, due to either adversarial effects on the preparation device or unknown noise contamination after perfect preparation, the state $\rho$ evolves into $\mathcal{N}_p(\rho)$. 
It is easy to see that the set 
\[
\left\{\mathcal{N}_p(\rho) : p\in[0,p_{\mathrm{th}}]\right\}
\]
is convex and compact; similarly, the set of quantum states $\left\{\mathcal{R}(\rho) : p_k\ge0,\;\sum_k p_k=1,\; \{U_k\}_{k=1}^{N}\text{ fixed}\right\}$ obtained by evolving $\rho$ through mixed unitary channels is also convex and compact. This is precisely $\ms{R}_1$ in the black-box setting of universal entanglement distillation.

We now proceed to the $n$-fold preparation process. Due to device instability or unpredictable noise, each round of preparation yields a different state. Thus, the resulting state is $\rho_1 \otimes \rho_2 \otimes \cdots \otimes \rho_n$ with $\rho_i \in \ms{R}_1$ for all $i$. Given two $n$-copy states $\rho_1 \otimes \rho_2 \otimes \cdots \otimes \rho_n$ and $\rho'_1 \otimes \rho'_2 \otimes \cdots \otimes \rho'_n$, suppose we decide which state to use for distillation by flipping an unfair coin that lands heads with probability $p$ and tails with probability $1-p$. Then the overall effect is equivalent to performing distillation on the mixed state
\[
p\,\rho_1 \otimes \rho_2 \otimes \cdots \otimes \rho_n \;+\; (1-p)\,\rho'_1 \otimes \rho'_2 \otimes \cdots \otimes \rho'_n.
\]
Consequently, following the above procedure, the black box in finite-blocklength universal entanglement distillation is
$$\ms{R}_n=\mathrm{conv}\left(\{\rho_1 \otimes \rho_2 \otimes \cdots \otimes \rho_n:\rho_i \in \ms{R}_1, \forall i\}\right).$$
It is clear that $\ms{R}_n$ induced by the aforementioned noise models is convex, compact, and closed under permutations of the systems, and that the sequence $\ms{R}=\{\ms{R}_n\}_{n\in\mathbb{N}}$  is stable under tensor products and compatible with both separable measurements and PPT measurements.

\section{Universal Coherence Distillation}
In this section, we analyze the one-shot distillation error and the reliability function of universal coherence distillation. As a complement, we also analyze the zero-error distillation rate and the one-shot distillation rate. The distillation is performed starting not only from a mixed coherent state $\omega$ but also from a black-box set $\ms{C}$, aiming to produce a maximally coherent state of the largest possible dimension.
We study two classes of free operations: maximally incoherent operations (MIO) and dephasing-covariant incoherent operations (DIO).  
A quantum channel $\Lambda$ belongs to MIO if it maps every incoherent state to an incoherent state; i.e., $\Lambda(\mathcal{I}) \subseteq \mathcal{I}$.  
A quantum channel $\Lambda$ belongs to DIO if it commutes with the completely dephasing channel $\Delta$; that is, $\Delta \circ \Lambda = \Lambda \circ \Delta$. Here $\Delta$ is the map that removes all off-diagonal elements in the fixed reference basis, i.e., $\Delta(X):=\sum_{i}|i\>\langle i|X |i\>\langle i|$.
Throughout this work, we denote the set of MIO and DIO by the same symbols $\mathrm{MIO}$ and $\mathrm{DIO}$, respectively, and we will analyze the performance of coherence distillation under these two families of free operations.

Let $\ms{C}$ be a convex compact set of quantum states on the Hilbert space $\mathcal{H}$. Under the free operations $\mathcal{O}$, the aim is to extract a $K$-dimensional maximally coherent state $\Psi_K$. Then the one-shot distillation fidelity of universal coherence distillation is defined as 
\begin{align}
    \fid(\ms{C}\!\!\xrightarrow{\mathcal{O}}\! \Psi_K):=\max_{\Lambda\in\mathcal{O}} \min_{\omega\in\ms{C}}F(\Lambda(\omega),\Psi_K),\nonumber
\end{align}
where the inner minimum ensures the worst-case performance over all possible input states in the convex compact set $\ms{C}$, and the outer maximum is taken over all allowed free operations $\Lambda\in\mc{O}$. 
The one-shot distillation error of universal coherence distillation is defined as 
\begin{align}
    \err(\ms{C}\!\!\xrightarrow{\mathcal{O}}\! \Psi_K)\!\!:=&1-\fid(\ms{C}\!\!\xrightarrow{\mathcal{O}}\! \Psi_K).\nonumber
\end{align}

\subsection{One-Shot Characterization and High-Rate Reliability}
\label{sec:One-Shot Characterization and High-Rate Reliability}
We expect to translate the one-shot distillation error to the type-II error of composite hypothesis testing that distinguishes the initial-state set $\ms{C}$ from the incoherent-state set $\mathcal{I}$, but the connection is less direct than in the entanglement case. To bridge this gap, we introduce two optimization problems as intermediaries: the exact optimization problem corresponds to the one-shot universal distillation error, while the relaxed optimization problem corresponds to composite hypothesis testing. By establishing crucial mutual bounds between the two optimization problems, we obtain the reliability function of universal coherence distillation for high distillation rates.

The exact optimization problem is formulated as 
\begin{align}
    \F_=(\ms{C}, K) := \max_{\substack{0 \le G \le I \\ \Delta(G) = \mathds{1}/K}} \min_{\omega\in\ms{C}}\operatorname{Tr}(\omega G),
\end{align}
and the relaxed optimization problem is described as 
\begin{align}
    \F_{\leq}(\ms{C}, K) := \max_{\substack{0 \le G \le I \\ \Delta(G) \leq \mathds{1}/K}} \min_{\omega\in\ms{C}}\operatorname{Tr}(\omega G).
\end{align}
We also denote the errors induced by the two optimization problems as 
\begin{align}
  \varepsilon_{=}(\ms{C},1/K)=1-\F_{=}(\ms{C}, K),\\
  \varepsilon_{\le}(\ms{C},1/K)=1-\F_{\leq}(\ms{C}, K).
\end{align}

First, we establish a precise equivalence between the one-shot distillation fidelity and the exact optimization problem, generalizing the result of~\cite{regula2018one} to the black-box setting.
\begin{theorem}
\label{them:op1}
Let $\ms{C}$ be a convex compact set of quantum states on $\mathcal{H}$, with the goal of extracting a $K$-dimensional maximally coherent state $\Psi_K$. Then, for both MIO and DIO, the fidelity of universal coherence distillation equals the optimum of the exact optimization problem,
    \begin{align}
        \fid(\ms{C}\!\!\xrightarrow{\mathrm{MIO}}\! \Psi_K)=\fid(\ms{C}\!\!\xrightarrow{\mathrm{DIO}}\! \Psi_K)=\F_=(\ms{C}, K).
    \end{align}
\end{theorem}
\begin{proof}
    Let $\Lambda$ be the maximally incoherent operation that achieves the optimal distillation in $\fid(\ms{C}\!\!\xrightarrow{\mathrm{MIO}}\! \Psi_K)$. Set $G := \Lambda^\dagger(\Psi_K)$. Since $\Lambda^\dagger$ is completely positive and unital, it follows that $0 \leq G \le \mathds{1}$. Furthermore, for each incoherent basis vector $|i\>$, it holds that
    $$\langle i | G | i \rangle = \text{Tr}[\Psi_K \Lambda(|i\rangle\langle i|)]=\frac{1}{K},$$
    where the last equality holds because $\Lambda(|i\rangle\langle i|)$ is incoherent, and $\operatorname{Tr}(\Psi_K \delta) = \frac{1}{K}$ for every incoherent state $\delta$. Hence, $\Delta(G) = \mathds{1}/K$. Therefore, the constructed $G$ is a feasible solution to the exact optimization problem $\F_=(\ms{C}, K)$. Finally, 
    $$\min_{\omega\in\ms{C}}\operatorname{Tr}(\omega G)=\min_{\omega\in\ms{C}}\tr(\Psi_K\Lambda(\omega)),$$
    where the right-hand side is precisely $\fid(\ms{C}\!\!\xrightarrow{\mathrm{MIO}}\! \Psi_K)$. This proves $\fid(\ms{C}\!\!\xrightarrow{\mathrm{MIO}}\! \Psi_K)\le \F_=(\ms{C}, K)$.

    Conversely, Let $G$ be the optimal solution to the exact optimization problem $\F_=(\ms{C}, K)$. Construct a distillation protocol,
    $$\Lambda(X) := \text{Tr}(GX)\Psi_K + \text{Tr}[(I - G)X]\frac{\1 - \Psi_K}{K - 1}.$$
    This is a CPTP map, and $\Delta(\Lambda(X)) = \frac{\text{Tr} X}{K} \1$ since $\Delta(\Psi_K) = \1/K.$ On the other hand, it holds that $\Lambda(\Delta(X))=\frac{\text{Tr} X}{K} \1$ since $\operatorname{Tr}[G\Delta(X)] = \operatorname{Tr}[\Delta(G)X] = \operatorname{Tr}[\frac{\1}{K}X]=\frac{\operatorname{Tr} X}{K}$. Therefore, the constructed $\Lambda$ is a dephasing-covariant incoherent operation. Furthermore, the figure of merit of distillation is
    $$\min_{\omega\in\ms{C}}\tr(\Psi_K\Lambda(\omega))=\min_{\omega\in\ms{C}}\tr(G\omega).$$
    This shows $\fid(\ms{C}\!\!\xrightarrow{\mathrm{DIO}}\! \Psi_K)\ge \F_=(\ms{C}, K)$.

    Combining the inequality $\fid(\ms{C}\!\!\xrightarrow{\mathrm{MIO}}\! \Psi_K)\ge\fid(\ms{C}\!\!\xrightarrow{\mathrm{DIO}}\! \Psi_K)$, which follows from the fact that $\mathrm{MIO}\supseteq\mathrm{DIO}$, we have proved this theorem.
\end{proof}

Second, we connect the relaxed optimization problem with composite hypothesis testing. Following the previous notation, the type‑II error of composite hypothesis testing is denoted by
\begin{align}
    \beta(\mc{I},\ms{C},\mathds{ALL},{1}/{K})=\min_{\substack{\{M,\mathds{1}-M\}\in\mathds{ALL}\\\max_{\delta\in\mc{I}}\tr((\1-M)\delta)\le 1/K}}\max_{\omega\in\ms{C}}\tr M\omega.\nonumber
\end{align}
\begin{lemma}
\label{lem:op2}
In composite hypothesis testing where the null hypothesis is the set of incoherent states $\mc{I}$ and the alternative hypothesis is the set of initial states $\ms{C}$, under all measurements and with the type-I error constrained to $1/K$, the optimal type-II error equals the error induced by the relaxed optimization problem,
    \begin{align}
       \beta(\mc{I},\ms{C},\mathds{ALL},{1}/{K}) =\varepsilon_{\le}(\ms{C},1/K).
    \end{align}
\end{lemma}
\begin{proof}
    Recalling the definitions of both, the left-hand side is actually 
    $$\min_{\substack{{0\le G\le \1}\\\max_{\delta\in\mc{I}}\tr(G\delta)\le 1/K}}\max_{\omega\in\ms{C}}\tr ((\1-G)\omega),$$
    while the right-hand side is actually 
    $$\min_{\substack{{0\le G\le \1}\\\Delta(G)\le\1/K}}\max_{\omega\in\ms{C}}\tr ((\1-G)\omega).$$
    Therefore, it suffices to show that the two constraints $\max_{\delta\in\mc{I}}\tr(G\delta)\le 1/K$ and $\Delta(G)\le\1/K$ are equivalent in this case.
    This is guaranteed by
    $$\max_{\delta \in \mathcal{I}} \operatorname{Tr}(G\delta) = \max_{x} \langle x | G | x \rangle = \|\Delta(G)\|_\infty.$$
\end{proof}

Next, we establish a key bounding relation: the error induced by the exact optimization problem is bounded from above and below by the error induced by the relaxed optimization problem. By virtue of the two equivalence relations established in Theorem~\ref{them:op1} and Lemma~\ref{lem:op2}, this is equivalent to bounding the one-shot distillation error of universal coherence distillation from above and below by the type-II error of composite hypothesis testing.
\begin{lemma}
\label{lem: lower and up bound coherence}
Let $\ms{C}$ be a convex compact set of quantum states on $\mathcal{H}$ and let $c$ be a real number with $0<c<1$. We have
    \begin{align}
        \varepsilon_{\le}(\ms{C},c)\le\varepsilon_{=}(\ms{C},c)\le \left( \sqrt{\varepsilon_{\le}(\ms{C}, c)} + 1 - \sqrt{1 - c} \right)^2.
    \end{align}
\end{lemma}
\begin{proof}
    Recalling the definitions of the two expressions above, they can be simplified as:
    \begin{align}
        \varepsilon_{=}(\ms{C}, c) &= \min_{\substack{0 \le H \le \1 \\ \Delta(H) = (1-c)\1}} \max_{\omega\in\ms{C}}\operatorname{Tr}(\omega H),\\
        \varepsilon_{\le}(\ms{C}, c) &= \min_{\substack{0 \le H \le \1 \\ \Delta(H) \ge (1-c)\1}}\max_{\omega\in\ms{C}} \operatorname{Tr}(\omega H).\label{eq:to find H}
    \end{align}
    Thus, the first inequality follows from the inclusion relation of the feasible sets.

    Let $H$ be optimal in Equation~\eqref{eq:to find H}, and $h_x:=\langle x|H|x\>$. Since $1-c\le h_x\le 1$, the diagonal operator $$D:=\sum_x\sqrt{\frac{1-c}{h_x}}|x\>\langle x|$$
    is a well-defined contraction. Set $H^\prime:=DHD$. We have 
    $$0\le H^\prime\le D^2\le \1,~\langle x|H'|x \rangle = 1 - c,$$
    so $H^\prime$ is feasible for the exact problem $\varepsilon_{=}(\ms{C}, c)$.
    
    Let $\omega^{*}$ be the state achieving the maximum $\max_{\omega \in \ms{C}} \operatorname{Tr}(\omega H^\prime)$. Using the Hilbert–Schmidt norm together with the triangle inequality and H\"older inequality, 
    \begin{align}
        \sqrt{\operatorname{Tr}(\omega^{*} H')} &= \|\sqrt{H} D \sqrt{\omega^{*}}\|_{2} \nonumber\\
        &\le \|\sqrt{H} \sqrt{\omega^{*}}\|_{2} + \|\sqrt{H} (I - D) \sqrt{\omega^{*}}\|_{2} \nonumber\\
        &\le \sqrt{\operatorname{Tr}(\omega^{*} H)} + \|I - D\|_{\infty}\nonumber\\
        &\le \sqrt{\max_{\omega\in\ms{C}}\operatorname{Tr}(\omega H)} + 1-\sqrt{1-c}.\label{eq:to squaring}
    \end{align}
    From the definition of $\varepsilon_{=}(\ms{C}, c)$, squaring Equation~\eqref{eq:to squaring} proves the second inequality.
\end{proof}

We now turn to the asymptotic behavior of universal coherence
distillation. Let
$\ms{C}=\{\ms{C}_n\}_{n\in\NN}$ be a stable sequence of convex
compact sets of quantum states with
$\ms{C}_n\subseteq\density(\mathcal{H}^{\ox n})$, and let
$\mc{O}\in\{\mathrm{MIO},\mathrm{DIO}\}$.
For a sequence of target dimensions
$\{K_n\}_{n\in\NN}$ satisfying
\begin{equation}
\lim_{n\to\infty}\frac{1}{n}\log K_n=r,
\label{eq:coherence-target-rate}
\end{equation}
or equivalently $K_n=2^{nr+o(n)}$, we define the lower and upper
reliability functions of universal coherence distillation at rate
$r>0$ as
\begin{align}
\underline{E}^{\rf}(\ms{C},\mc{O},r)
:=
\liminf_{n\to\infty}
-\frac{1}{n}
\log
\err\left(
\ms{C}_n\xrightarrow{\mc{O}}\Psi_{K_n}
\right),
\label{eq:coherence-lower-rf}
\end{align}
and
\begin{align}
\overline{E}^{\rf}(\ms{C},\mc{O},r)
:=
\limsup_{n\to\infty}
-\frac{1}{n}
\log
\err\left(
\ms{C}_n\xrightarrow{\mc{O}}\Psi_{K_n}
\right).
\label{eq:coherence-upper-rf}
\end{align}
Whenever the two quantities coincide and their common value depends
only on the asymptotic rate $r$, rather than on the particular choice
of the sequence $\{K_n\}$ satisfying
Equation~\eqref{eq:coherence-target-rate}, we denote this common value
by
\begin{equation}
E^{\rf}(\ms{C},\mc{O},r)
:=
\underline{E}^{\rf}(\ms{C},\mc{O},r)
=
\overline{E}^{\rf}(\ms{C},\mc{O},r),
\label{eq:coherence-rf}
\end{equation}
and call it the reliability function of universal coherence
distillation at rate $r$.

Leveraging the preceding one-shot comparison between the exact and
relaxed optimization problems, we now derive the reliability function
of universal coherence distillation in the high-rate regime, which is of greater physical significance.

\begin{theorem}
\label{them:rf for high rate}
Let $\ms{C}=\{\ms{C}_n\}_{n\in\NN}$ be a stable sequence of convex
compact sets of quantum states with
$\ms{C}_n\subseteq\density(\mathcal{H}^{\ox n})$, and let the free
operation $\mc O$ be MIO or DIO. Then, for any distillation rate $r$
satisfying
\begin{equation}
D_{\Petz,0}^{\infty}(\ms{C}\Vert\mc I)
<
r
<
D^{\infty}(\ms{C}\Vert\mc I)
\end{equation}
and
\begin{equation}
\frac{1}{2}H_r^{\infty}(\ms{C}\Vert\mc I)<r,
\end{equation}
the reliability function of universal coherence distillation exists
and is given by
\begin{equation}
\label{eq:C rf high rate}
E^{\rf}(\ms{C},\mc O,r)
=
H_r^{\infty}(\ms{C}\Vert\mc I).
\end{equation}

In particular, when the initial state is fully known, i.e.,
$\ms C=\{\omega^{\otimes n}\}_{n\in\NN}$, define
\begin{equation}
r_{\mathrm{crit}}
:=
\sup_{\alpha\in(0,1)}
\frac{1-\alpha}{1+\alpha}
D_{\Petz,\alpha}(\omega\Vert\mc I).
\end{equation}
Then, for every
\begin{equation}
r\in
\bigl(
r_{\mathrm{crit}},
D(\omega\Vert\mc I)
\bigr),
\end{equation}
the reliability function is
\begin{equation}
\label{eq:single state rf}
E^{\rf}(\omega,\mc O,r)
=
\sup_{\alpha\in(0,1)}
\frac{\alpha-1}{\alpha}
\left(
r-D_{\Petz,\alpha}(\omega\Vert\mc I)
\right).
\end{equation}
\end{theorem}

\begin{proof}
Let $\{K_n\}_{n\in\NN}$ be any sequence of target dimensions satisfying
\begin{equation}
\frac{1}{n}\log K_n\longrightarrow r,
\end{equation}
or equivalently,
\begin{equation}
K_n=2^{nr+o(n)}.
\end{equation}
By Lemma~\ref{lem:op2}, the error of the relaxed optimization problem
is exactly the type-II error of the composite hypothesis testing
problem,
\begin{equation}
\varepsilon_{\leq}
\left(
\ms C_n,\frac{1}{K_n}
\right)
=
\beta_n
\left(
\mc I_n,\ms C_n,\mathds{ALL},\frac{1}{K_n}
\right).
\label{eq:relaxed-error-beta}
\end{equation}

We first determine the exponential rate of the relaxed error. Set
\begin{equation}
r_0
:=
D_{\Petz,0}^{\infty}(\ms C\Vert\mc I).
\end{equation}
By assumption, $r>r_0$. As established above (Equations~\eqref{eq:finite-Hoeffding-region}-\eqref{eq:H-left-continuity-above-r0}), the regularized
Hoeffding divergence
\begin{equation}
s\longmapsto H_s^\infty(\ms C\Vert\mc I)
\end{equation}
is finite and convex on $(r_0,\infty)$ and is therefore continuous
throughout this interval.

Fix $\delta>0$ sufficiently small such that
\begin{equation}
r_0<r-\delta<r+\delta
<
D^\infty(\ms C\Vert\mc I).
\end{equation}
Since $K_n^{-1}=2^{-nr+o(n)}$, for all sufficiently large $n$,
\begin{equation}
2^{-n(r+\delta)}
\leq
\frac{1}{K_n}
\leq
2^{-n(r-\delta)}.
\label{eq:Kn-threshold-sandwich}
\end{equation}
The function
$\beta_n(\mc I_n,\ms C_n,\mathds{ALL},\varepsilon)$ is non-increasing
in the allowed type-I error $\varepsilon$. Hence
Equation~\eqref{eq:Kn-threshold-sandwich} implies
\begin{align}
&\beta_n(
\mc I_n,\ms C_n,\mathds{ALL},2^{-n(r+\delta)})
\nonumber\\
&\qquad\geq
\beta_n
\left(
\mc I_n,\ms C_n,\mathds{ALL},\frac{1}{K_n}
\right)
\nonumber\\
&\qquad\geq
\beta_n(
\mc I_n,\ms C_n,\mathds{ALL},2^{-n(r-\delta)}).
\end{align}
Taking $-\frac1n\log$, followed by $\liminf$ and $\limsup$, and using
Theorem~\ref{them:Fang and Hayashi}, we obtain
\begin{align}
H_{r+\delta}^{\infty}(\ms C\Vert\mc I)
&\leq
\liminf_{n\to\infty}
-\frac1n
\log
\varepsilon_{\leq}
\left(
\ms C_n,\frac1{K_n}
\right)
\nonumber\\
&\leq
\limsup_{n\to\infty}
-\frac1n
\log
\varepsilon_{\leq}
\left(
\ms C_n,\frac1{K_n}
\right)
\nonumber\\
&\leq
H_{r-\delta}^{\infty}(\ms C\Vert\mc I).
\label{eq:relaxed-error-H-sandwich}
\end{align}
Letting $\delta$ tend to $0$ and using the continuity of
$H_s^\infty(\ms C\Vert\mc I)$ at $s=r$ yields
\begin{equation}
\lim_{n\to\infty}
-\frac1n
\log
\varepsilon_{\leq}
\left(
\ms C_n,\frac1{K_n}
\right)
=
H_r^\infty(\ms C\Vert\mc I).
\label{eq:Coherence rf and inequality}
\end{equation}
Equivalently,
\begin{equation}
\varepsilon_{\leq}
\left(
\ms C_n,\frac1{K_n}
\right)
=
2^{-nH_r^\infty(\ms C\Vert\mc I)+o(n)}.
\label{eq:relaxed-error-asymptotic}
\end{equation}

We next transfer this exponent from the relaxed optimization problem
to the exact one. Set
\begin{equation}
c_n:=K_n^{-1}.
\end{equation}
Since $c_n=2^{-nr+o(n)}$, we have
\begin{align}
1-\sqrt{1-c_n}
&=
\frac{c_n}{1+\sqrt{1-c_n}}
\nonumber\\
&=
2^{-nr+o(n)}.
\label{eq:diagonal-correction-exponent}
\end{align}
By Lemma~\ref{lem: lower and up bound coherence},
\begin{align}
\varepsilon_{=}
\left(
\ms C_n,\frac1{K_n}
\right)
&\leq
\left(
\sqrt{
\varepsilon_{\leq}
\left(
\ms C_n,\frac1{K_n}
\right)}
+
1-\sqrt{1-c_n}
\right)^2
\nonumber\\
&=
\left(
2^{-nH_r^\infty(\ms C\Vert\mc I)/2+o(n)}
+
2^{-nr+o(n)}
\right)^2
\nonumber\\
&\leq
4\cdot
2^{-n\min\{
H_r^\infty(\ms C\Vert\mc I),\,2r
\}+o(n)}.
\end{align}
Therefore,
\begin{equation}
\liminf_{n\to\infty}
-\frac1n
\log
\varepsilon_{=}
\left(
\ms C_n,\frac1{K_n}
\right)
\geq
\min\left\{
H_r^\infty(\ms C\Vert\mc I),\,2r
\right\}.
\label{eq:exact-error-liminf}
\end{equation}

On the other hand, the first inequality in
Lemma~\ref{lem: lower and up bound coherence} gives
\begin{equation}
\varepsilon_{\leq}
\left(
\ms C_n,\frac1{K_n}
\right)
\leq
\varepsilon_{=}
\left(
\ms C_n,\frac1{K_n}
\right).
\end{equation}
Together with Equation~\eqref{eq:Coherence rf and inequality}, this
implies
\begin{equation}
\limsup_{n\to\infty}
-\frac1n
\log
\varepsilon_{=}
\left(
\ms C_n,\frac1{K_n}
\right)
\leq
H_r^\infty(\ms C\Vert\mc I).
\label{eq:exact-error-limsup}
\end{equation}

Since
\begin{equation}
\frac12 H_r^\infty(\ms C\Vert\mc I)<r,
\end{equation}
we have
\begin{equation}
\min\left\{
H_r^\infty(\ms C\Vert\mc I),\,2r
\right\}
=
H_r^\infty(\ms C\Vert\mc I).
\end{equation}
Equations~\eqref{eq:exact-error-liminf} and
\eqref{eq:exact-error-limsup} therefore show that the limit exists and
\begin{equation}
\lim_{n\to\infty}
-\frac1n
\log
\varepsilon_{=}
\left(
\ms C_n,\frac1{K_n}
\right)
=
H_r^\infty(\ms C\Vert\mc I).
\end{equation}
By Theorem~\ref{them:op1}, this is precisely the reliability function
of universal coherence distillation under either MIO or DIO, proving
Equation~\eqref{eq:C rf high rate}.

We finally specialize to the case of a fully known initial state,
$\ms C_n=\{\omega^{\otimes n}\}$. The Petz R\'enyi relative entropy
of coherence is additive, and hence
\begin{equation}
H_r^\infty(\omega\Vert\mc I)
=
\sup_{\alpha\in(0,1)}
\frac{\alpha-1}{\alpha}
\left(
r-D_{\Petz,\alpha}(\omega\Vert\mc I)
\right).
\label{eq:single-state-Hoeffding}
\end{equation}
Moreover, the sufficient high-rate condition
\begin{equation}
\frac12H_r^\infty(\omega\Vert\mc I)<r
\end{equation}
is ensured by
\begin{equation}
r>
\sup_{\alpha\in(0,1)}
\frac{1-\alpha}{1+\alpha}
D_{\Petz,\alpha}(\omega\Vert\mc I)
=
r_{\mathrm{crit}}.
\end{equation}
In addition,
\begin{equation}
r_{\mathrm{crit}}
\geq
\lim_{\alpha\to0}
\frac{1-\alpha}{1+\alpha}
D_{\Petz,\alpha}(\omega\Vert\mc I)
=
D_{\Petz,0}(\omega\Vert\mc I).
\end{equation}
Thus every $r>r_{\mathrm{crit}}$ automatically lies strictly above the
order-zero critical rate, so the continuity requirement used above is
satisfied. Substituting Equation~\eqref{eq:single-state-Hoeffding}
into the general result proves Equation~\eqref{eq:single state rf}.
\end{proof}

\begin{remark}
The high-rate interval
\begin{equation}
\bigl(
r_{\mathrm{crit}},
D(\omega\Vert\mc I)
\bigr)
\end{equation}
in Theorem~\ref{them:rf for high rate} is always nonempty whenever
\begin{equation}
D_{\Petz,0}(\omega\Vert\mc I)
<
D(\omega\Vert\mc I).
\end{equation}
Indeed, define
\begin{equation}
g_\omega(\alpha)
:=
\frac{1-\alpha}{1+\alpha}
D_{\Petz,\alpha}(\omega\Vert\mc I),
\qquad
\alpha\in(0,1).
\end{equation}
This function extends continuously to $[0,1]$, with
\begin{equation}
g_\omega(0)
=
D_{\Petz,0}(\omega\Vert\mc I),
\qquad
g_\omega(1)=0.
\end{equation}
In particular,
\begin{equation}
r_{\mathrm{crit}}
=
\max_{\alpha\in[0,1]}g_\omega(\alpha)
\geq
D_{\Petz,0}(\omega\Vert\mc I),
\end{equation}
so every rate $r>r_{\mathrm{crit}}$ also satisfies the condition
$r>D_{\Petz,0}(\omega\Vert\mc I)$ required in the general part of
Theorem~\ref{them:rf for high rate}.

For every $\alpha\in(0,1)$, monotonicity of the Petz R\'enyi
divergence in its order gives
\begin{equation}
g_\omega(\alpha)
<
D_{\Petz,\alpha}(\omega\Vert\mc I)
\leq
D(\omega\Vert\mc I).
\end{equation}
Together with
$g_\omega(0)<D(\omega\Vert\mc I)$ and
$g_\omega(1)=0<D(\omega\Vert\mc I)$, continuity on the compact
interval $[0,1]$ implies
\begin{equation}
r_{\mathrm{crit}}
<
D(\omega\Vert\mc I).
\end{equation}
Hence the interval
$\bigl(r_{\mathrm{crit}},D(\omega\Vert\mc I)\bigr)$ is nonempty.
\end{remark}

\subsection{Seed-Assisted Approach and Full-Range Reliability}
In this section, we adopt a different approach to study the reliability function of coherence distillation. First, we introduce an auxiliary seed system $B$ and derive the reliability function of coherence distillation on the principal system $A$ assisted by this seed. Then, we consider distilling the seed from the initial state itself, thereby eliminating the need for an external seed, and obtain the unassisted reliability function for the full range of distillation rates without any critical point. The final form of the reliability function coincides with that given in Section~\ref{sec:One-Shot Characterization and High-Rate Reliability}. For convenience, we assume full knowledge of the initial state throughout this section.

First, we consider the general case where the state on the seed system is not a perfect maximally coherent state but a generic state, and investigate the relation between the performance of coherence distillation on the joint system and the corresponding optimization problems on the individual systems.
\begin{theorem}
\label{them:relax seed}
Let $\omega_A$ be a state on the principal Hilbert space $\mathcal{H}_A$ and $\eta_B$ a state on the auxiliary seed Hilbert space $\mathcal{H}_B$, and let $K, L \ge 2$ be integers. Then
\begin{equation}
  \F_{=}(\omega_A\otimes\eta_B,KL)
  \geq
  \F_{\leq}(\omega_A,K)\F_{=}(\eta_B,L),
  \label{eq:seeded-completion-fidelity}
\end{equation}
and consequently,
\begin{equation}
  \varepsilon_{=}(\omega_A\otimes\eta_B,1/(KL))
  \leq
  \varepsilon_{\leq}(\omega_A,1/K)
  +\varepsilon_{=}(\eta_B,1/L).
  \label{eq:seeded-completion-particular}
\end{equation}
\end{theorem}

\begin{proof}
Let $Q_A$ be feasible for the relaxed problem,
\begin{equation}
  0\leq Q_A\leq \1_A,
  \qquad
  \Delta_A(Q_A)\leq \frac{1}{K}\1_A,
\end{equation}
and define its diagonal deficit by
\begin{equation}
  D_Q:=\frac{1}{K}\1_A-\Delta_A(Q_A).
\end{equation}
Then
\[
  0\leq D_Q\leq \frac{1}{K}\1_A\leq \1_A.
\]
Let $S_B$ be feasible for the exact seed problem,
\begin{equation}
  0\leq S_B\leq \1_B,
  \qquad
  \Delta_B(S_B)=\frac{1}{L}\1_B.
\end{equation}

Consider
\begin{equation}
  T_{AB}
  :=
  Q_A\otimes S_B
  +
  D_Q\otimes\frac{\1_B-S_B}{L-1}.
\end{equation}
Both summands are positive, so $T_{AB}\geq 0$. To check the upper
bound, diagonalize
\[
  S_B=\sum_j s_j\lvert u_j\rangle\langle u_j\rvert,
  \qquad 0\leq s_j\leq 1.
\]
In this decomposition,
\begin{equation}
  T_{AB}
  =
  \bigoplus_j
  \left(
    s_jQ_A+\frac{1-s_j}{L-1}D_Q
  \right).
\end{equation}
Using $Q_A\leq \1_A$, $D_Q\leq \1_A$, and $L\geq 2$, each block is
bounded by
\begin{equation}
  \left(
    s_j+\frac{1-s_j}{L-1}
  \right)\1_A
  \leq \1_A.
\end{equation}
Hence $T_{AB}\leq \1_{AB}$. Furthermore,
\begin{equation}
  \Delta_B\left(
    \frac{\1_B-S_B}{L-1}
  \right)
  =
  \frac{1}{L}\1_B.
\end{equation}
It follows that
\begin{align}
  \Delta_{AB}(T_{AB})
  &=
  \Delta_A(Q_A)\otimes\frac{1}{L}\1_B
  +
  D_Q\otimes\frac{1}{L}\1_B
  \notag\\
  &=
  \frac{1}{KL}\1_{AB}.
\end{align}
Thus $T_{AB}$ is exact feasible. Its objective value is
\begin{align}
  \operatorname{Tr}\!\left[
    (\omega_A\otimes\eta_B)T_{AB}
  \right]
  &=
  \operatorname{Tr}(\omega_AQ_A)
  \operatorname{Tr}(\eta_BS_B)
  \notag\\
  &\quad+
  \operatorname{Tr}(\omega_AD_Q)
  \frac{
    1-\operatorname{Tr}(\eta_BS_B)
  }{L-1}
  \notag\\
  &\geq
  \operatorname{Tr}(\omega_AQ_A)
  \operatorname{Tr}(\eta_BS_B).
\end{align}
Optimizing over $Q_A$ and $S_B$ proves
Eq.~\eqref{eq:seeded-completion-fidelity}. 

Taking complements gives
\begin{align}
  \varepsilon_{=}(\omega_A\otimes\eta_B,1/(KL))
  \leq{}&
  \varepsilon_{\leq}(\omega_A,1/K)
  +\varepsilon_{=}(\eta_B,1/L)
  \notag\\
  &-
  \varepsilon_{\leq}(\omega_A,1/K)
  \varepsilon_{=}(\eta_B,1/L),
  \label{eq:seeded-completion-error}
\end{align} 
and dropping its nonnegative
product term gives Eq.~\eqref{eq:seeded-completion-particular}.
\end{proof}

Furthermore, when the state on the seed system is a maximally coherent state, regardless of its dimension, the inequality relation established in Theorem~\ref{them:relax seed} above tightens to an equality, i.e., the fidelity of coherence distillation on the joint system transforms precisely into the relaxed optimization problem on the principal system.

\begin{theorem}
\label{them:exact seed maximally coherent state}
For every state $\omega_A$ on the Hilbert space $\mathcal{H}_A$ and integers $K,L\geq 2$,
\begin{equation}
    \F_{=}(\omega_A\otimes\Psi_L,KL)
    =
    \F_{\leq}(\omega_A,K),
\end{equation}
and equivalently,
\begin{equation}
  \varepsilon_{=}(\omega_A\otimes\Psi_L,1/(KL))
  =
  \varepsilon_{\leq}(\omega_A,1/K).
\end{equation}
\end{theorem}

\begin{proof}
Since $\Delta(\Psi_L)=\frac{\1}{L}$ and $\operatorname{Tr}(\Psi_L^2)=1,$
Theorem~\ref{them:relax seed} with $S_B=\Psi_L$ gives
\begin{equation}
\label{eq:easy direction}
  \F_{=}(\omega_A\otimes\Psi_L,KL)
  \geq
  \F_{\leq}(\omega_A,K).
\end{equation}

For the reverse inequality, we first prove
\begin{equation}
\label{eq:inter little bridge}
  \F_{\leq}(\omega_A\otimes\Psi_L,KL)
  =
  \F_{\leq}(\omega_A,K).
\end{equation}
If $Q_A$ is feasible for $\F_{\leq}(\omega_A,K)$, then
$Q_A\otimes\Psi_L$ is feasible for the joint relaxed problem and has
the same objective value. This proves the ``$\geq$'' direction.

Conversely, let $T_{AB}$ be feasible for the joint relaxed problem
and define 
\begin{equation}
  \widetilde{Q}_A
  :=
  \bigl(\1_A\otimes\langle\Psi_L|\bigr)
  T_{AB}
  \bigl(\1_A\otimes|\Psi_L\rangle\bigr).
\end{equation}
Clearly,
\[
  0\leq\widetilde{Q}_A\leq\1_A.
\]
For each incoherent basis vector $|i\rangle$ of $A$, set
\[
  T_i:=\langle i|T_{AB}|i\rangle\geq 0.
\]
The joint diagonal constraint gives
\begin{equation}
  \langle j|T_i|j\rangle
  \leq
  \frac{1}{KL}
\end{equation}
for every $j$. Therefore,
\begin{align}
  \langle i|\widetilde{Q}_A|i\rangle
  &=
  \langle\Psi_L|T_i|\Psi_L\rangle
  \notag\\
  &\leq
  \operatorname{Tr}T_i
  =
  \sum_{j=1}^{L}\langle j|T_i|j\rangle
  \leq
  \frac{1}{K}.
\end{align}
Thus $\widetilde{Q}_A$ is relaxed feasible, and
\begin{equation}
  \operatorname{Tr}(\omega_A\widetilde{Q}_A)
  =
  \operatorname{Tr}
  \bigl[
    (\omega_A\otimes\Psi_L)T_{AB}
  \bigr].
\end{equation}
Optimizing over $T_{AB}$ proves the opposite direction in
Equation~\eqref{eq:inter little bridge}. Finally,
\begin{equation}
  \F_{=}(\omega_A\otimes\Psi_L,KL)
  \leq
  \F_{\leq}(\omega_A\otimes\Psi_L,KL)
  =
  \F_{\leq}(\omega_A,K).
\end{equation}
Together with Equation~\eqref{eq:easy direction}, this proves
the claim.
\end{proof}

Consequently, we obtain the reliability function of coherence distillation when assisted by a seed system in a maximally coherent state of arbitrary dimension.
\begin{corollary}[Seed-assisted reliability function]
\label{cor:seed-assisted-rf}
Let $\omega_A$ be a state on the Hilbert space $\mathcal{H}_A$, and
let
\begin{equation}
r_0:=D_{\Petz,0}(\omega\Vert\mc I).
\end{equation}
Consider MIO or DIO as the free operations, and let
$\{K_n\}_{n\in\NN}$ be any sequence of target dimensions satisfying
\begin{equation}
\lim_{n\to\infty}\frac{1}{n}\log K_n=r,
\end{equation}
or equivalently, $K_n=2^{nr+o(n)}$. Assisted by a maximally coherent
seed $\Psi_L$ of any fixed dimension $L\geq2$, for every
\begin{equation}
r\in
\bigl(0,D(\omega\Vert\mc I)\bigr)
\setminus
\bigl\{D_{\Petz,0}(\omega\Vert\mc I)\bigr\},
\end{equation}
the reliability-function limit exists and is independent of the
particular choice of $\{K_n\}$ satisfying the above asymptotic rate.
More precisely,
\begin{align}
&\lim_{n\to\infty}
-\frac{1}{n}
\log
\varepsilon_{=}
\left(
\omega^{\ox n}\ox\Psi_L,
\frac{1}{K_nL}
\right)
\nonumber\\
&\qquad=
\sup_{\alpha\in(0,1)}
\frac{\alpha-1}{\alpha}
\left(
r-D_{\Petz,\alpha}(\omega\Vert\mc I)
\right).
\label{eq:seed-assisted-rf}
\end{align}
In particular, the right-hand side is $+\infty$ for
$0<r<D_{\Petz,0}(\omega\Vert\mc I)$, while it is finite for
\begin{equation}
D_{\Petz,0}(\omega\Vert\mc I)
<
r
<
D(\omega\Vert\mc I).
\end{equation}
\end{corollary}

\begin{proof}
By Theorem~\ref{them:exact seed maximally coherent state}, for every $n$ and every $K_n$,
the maximally coherent seed converts the relaxed optimization problem
on the principal system exactly into the exact optimization problem on
the joint system:
\begin{equation}
\varepsilon_{=}
\left(
\omega^{\ox n}\ox\Psi_L,
\frac{1}{K_nL}
\right)
=
\varepsilon_{\leq}
\left(
\omega^{\ox n},
\frac{1}{K_n}
\right).
\label{eq:seed-exact-relaxed}
\end{equation}
It therefore suffices to determine the exponent of the relaxed error.

We consider separately the two regions on either side of the critical
rate
\begin{equation}
r_0=D_{\Petz,0}(\omega\Vert\mc I).
\end{equation}

First, suppose that $0<r<r_0$. As will be established in Lemma~\ref{lem:upper bound for zero-error rate} below, with
\begin{equation}
p_*:=\max_i\langle i|P_\omega|i\rangle,
\qquad
r_0=-\log p_*,
\end{equation}
one has
\begin{equation}
\F_{\leq}
\left(
\omega^{\ox n},K_n
\right)
=1
\quad\Longleftrightarrow\quad
K_n\leq p_*^{-n}=2^{nr_0}.
\label{eq:relaxed-zero-error-condition}
\end{equation}
Since
\begin{equation}
\frac1n\log K_n\longrightarrow r<r_0,
\end{equation}
there exists $\delta>0$ such that $r+\delta<r_0$, and for all
sufficiently large $n$,
\begin{equation}
K_n
\leq
2^{n(r+\delta)}
<
2^{nr_0}
=
p_*^{-n}.
\end{equation}
Hence, by Equation~\eqref{eq:relaxed-zero-error-condition},
\begin{equation}
\varepsilon_{\leq}
\left(
\omega^{\ox n},
\frac1{K_n}
\right)
=0
\end{equation}
for all sufficiently large $n$. Equation~\eqref{eq:seed-exact-relaxed}
then implies
\begin{equation}
\varepsilon_{=}
\left(
\omega^{\ox n}\ox\Psi_L,
\frac1{K_nL}
\right)
=0
\end{equation}
for all sufficiently large $n$. Consequently,
\begin{equation}
\lim_{n\to\infty}
-\frac1n
\log
\varepsilon_{=}
\left(
\omega^{\ox n}\ox\Psi_L,
\frac1{K_nL}
\right)
=
+\infty.
\label{eq:seed-below-r0}
\end{equation}
On the other hand, the quantum Hoeffding divergence satisfies
\begin{equation}
H_r^\infty(\omega\Vert\mc I)=+\infty,
\qquad
0<r<D_{\Petz,0}(\omega\Vert\mc I),
\end{equation}
so Equation~\eqref{eq:seed-assisted-rf} follows in this region.

We next consider
\begin{equation}
r_0<r<D(\omega\Vert\mc I).
\end{equation}
Fix $\delta>0$ sufficiently small such that
\begin{equation}
r_0<r-\delta<r+\delta<D(\omega\Vert\mc I).
\end{equation}
Since $K_n=2^{nr+o(n)}$, for all sufficiently large $n$,
\begin{equation}
2^{-n(r+\delta)}
\leq
\frac1{K_n}
\leq
2^{-n(r-\delta)}.
\label{eq:seed-threshold-sandwich}
\end{equation}
By Lemma~\ref{lem:op2},
\begin{equation}
\varepsilon_{\leq}
\left(
\omega^{\ox n},
\frac1{K_n}
\right)
=
\beta_n
\left(
\mc I_n,\omega^{\ox n},
\mathds{ALL},\frac1{K_n}
\right).
\end{equation}
Since the type-II error $\beta_n$ is non-increasing in the allowed
type-I error, Equation~\eqref{eq:seed-threshold-sandwich} gives
\begin{align}
&\beta_n(
\mc I_n,\omega^{\ox n},
\mathds{ALL},2^{-n(r+\delta)})
\nonumber\\
&\qquad\geq
\varepsilon_{\leq}
\left(
\omega^{\ox n},
\frac1{K_n}
\right)
\nonumber\\
&\qquad\geq
\beta_n(
\mc I_n,\omega^{\ox n},
\mathds{ALL},2^{-n(r-\delta)}).
\end{align}
Taking $-\frac1n\log$, followed by $\liminf$ and $\limsup$, and
applying the composite Hoeffding theorem yields
\begin{align}
H_{r+\delta}^{\infty}(\omega\Vert\mc I)
&\leq
\liminf_{n\to\infty}
-\frac1n
\log
\varepsilon_{\leq}
\left(
\omega^{\ox n},
\frac1{K_n}
\right)
\nonumber\\
&\leq
\limsup_{n\to\infty}
-\frac1n
\log
\varepsilon_{\leq}
\left(
\omega^{\ox n},
\frac1{K_n}
\right)
\nonumber\\
&\leq
H_{r-\delta}^{\infty}(\omega\Vert\mc I).
\label{eq:seed-H-sandwich}
\end{align}

Since $r>r_0$, the regularized Hoeffding divergence is finite and
continuous at $r$. Hence, letting $\delta$ tend to $0$ in
Equation~\eqref{eq:seed-H-sandwich} gives
\begin{equation}
\lim_{n\to\infty}
-\frac1n
\log
\varepsilon_{\leq}
\left(
\omega^{\ox n},
\frac1{K_n}
\right)
=
H_r^\infty(\omega\Vert\mc I).
\end{equation}
Combining this with Equation~\eqref{eq:seed-exact-relaxed}, we obtain
\begin{equation}
\lim_{n\to\infty}
-\frac1n
\log
\varepsilon_{=}
\left(
\omega^{\ox n}\ox\Psi_L,
\frac1{K_nL}
\right)
=
H_r^\infty(\omega\Vert\mc I).
\end{equation}

Finally, for a single known state the Petz R\'enyi relative entropy
of coherence is additive, and therefore
\begin{equation}
H_r^\infty(\omega\Vert\mc I)
=
\sup_{\alpha\in(0,1)}
\frac{\alpha-1}{\alpha}
\left(
r-D_{\Petz,\alpha}(\omega\Vert\mc I)
\right).
\end{equation}
This proves Equation~\eqref{eq:seed-assisted-rf}.
\end{proof}
\begin{remark}
At the critical rate
\begin{equation}
r=D_{\Petz,0}(\omega\Vert\mc I),
\end{equation}
the first-order condition $K_n=2^{nr+o(n)}$ alone is insufficient for
the preceding argument. Indeed, Lemma~\ref{lem:upper bound for zero-error rate}
shows that the exact zero-error behavior of the relaxed problem depends
on whether $K_n$ lies below or above the threshold
$p_*^{-n}=2^{nr}$. Thus subexponential corrections to $K_n$ can affect
the behavior precisely at this critical rate. We therefore make no
sequence-independent reliability claim at
$r=D_{\Petz,0}(\omega\Vert\mc I)$ without further specification of the
target sequence.
\end{remark}

Finally, suppose that the initial state $\omega$ has a positive
zero-error distillation rate,\footnote{The zero-error distillation rate
will be precisely defined and characterized in
Section~\ref{sec:Exact Zero-Error Distillation} below.}
$C_{d,0}(\omega)>0$. We distill an exact coherent seed from an
arbitrarily small fraction of the input and use it to assist the
distillation of the remaining copies. This bootstrap argument extends
the reliability characterization beyond the high-rate regime.

When the target dimensions are specified only through their first-order
asymptotic rate, $J_n=2^{nr+o(n)}$, the order-zero rate
$D_{\Petz,0}(\omega\Vert\mc I)$ requires separate treatment. Below this
rate, exact zero-error distillation is achievable, whereas above it the
relevant Hoeffding divergence is finite and continuous. At the critical
rate itself, however, subexponential perturbations of $J_n$ may cross
the zero-error threshold and can therefore affect the asymptotic error
behavior. Accordingly, the following sequence-independent
characterization applies to the entire direct regime except for this
single critical rate.

\begin{theorem}
\label{them: coherence rf full regime}
Let $\omega$ be a quantum state on the Hilbert space $\mathcal{H}$ with
zero-error distillation rate $C_{d,0}(\omega)>0$, and let
$\mc O$ be MIO or DIO. Define
\begin{equation}
d_0:=D_{\Petz,0}(\omega\Vert\mc I).
\end{equation}
Let $\{J_n\}_{n\in\NN}$ be any sequence of target dimensions satisfying
\begin{equation}
\lim_{n\to\infty}\frac1n\log J_n=r.
\end{equation}
 Then, for every
\begin{equation}
r\in
\bigl(0,D(\omega\Vert\mc I)\bigr)
\setminus
\{d_0\},
\end{equation}
the reliability-function limit exists, is independent of the particular
choice of $\{J_n\}$, and satisfies
\begin{align}
E^{\rf}(\omega,\mc O,r)
=
\sup_{\alpha\in(0,1)}
\frac{\alpha-1}{\alpha}
\left(
r-D_{\Petz,\alpha}(\omega\Vert\mc I)
\right).
\label{eq:coherence-rf-full-regime}
\end{align}
More specifically, $E^{\rf}(\omega,\mc O,r)=+\infty,$ for $
0<r<d_0,$
whereas the reliability function is finite and given by the
Hoeffding expression in
Equation~\eqref{eq:coherence-rf-full-regime} for $d_0<r<D(\omega\Vert\mc I).$
\end{theorem}

\begin{proof}
Let $P_\omega$ be the projection onto the support of $\omega$, and set
\begin{equation}
p_*:=\max_i\langle i|P_\omega|i\rangle.
\end{equation}
Then
\begin{equation}
d_0
=
D_{\Petz,0}(\omega\Vert\mc I)
=
-\log p_*.
\label{eq:d0-pstar-full-regime}
\end{equation}

We first note that
\begin{equation}
C_{d,0}(\omega)\leq d_0.
\label{eq:zero-error-rate-d0-upper}
\end{equation}
Indeed, suppose that an exact test $G_n$ achieves a target of dimension
$L_n$ with unit fidelity. Then
\begin{equation}
0\leq G_n\leq\mathds{1},
\qquad
\Delta(G_n)=\frac1{L_n}\mathds{1},
\qquad
\tr(\omega^{\ox n}G_n)=1.
\end{equation}
Since $\omega^{\ox n}$ is strictly positive on the range of
$P_\omega^{\ox n}$, the last equality implies
\begin{equation}
G_n\geq P_\omega^{\ox n}.
\end{equation}
Consequently,
\begin{equation}
\frac1{L_n}
=
\max_{i^n}
\langle i^n|G_n|i^n\rangle
\geq
\max_{i^n}
\langle i^n|P_\omega^{\ox n}|i^n\rangle
=
p_*^n,
\end{equation}
and hence
\begin{equation}
\frac1n\log L_n\leq-\log p_*=d_0.
\end{equation}
Taking the asymptotic supremum proves
Equation~\eqref{eq:zero-error-rate-d0-upper}. In particular,
$C_{d,0}(\omega)>0$ implies $d_0>0$.

We now distinguish the two regions separated by $d_0$.

\medskip
\noindent
\textit{1) The region $0<r<d_0$.}

Fix
\begin{equation}
0<r<d_0.
\end{equation}
Choose
\begin{equation}
0<r_0<\min\{r,C_{d,0}(\omega)\}
\end{equation}
and then choose
\begin{equation}
r<r_2<d_0.
\end{equation}
Set
\begin{equation}
\theta
:=
\frac{r_2-r}{r_2-r_0}
\in(0,1),
\qquad
r=\theta r_0+(1-\theta)r_2.
\label{eq:low-rate-convex-combination}
\end{equation}

Split the $n$ input copies into two blocks of lengths
\begin{equation}
l_n:=\lfloor\theta n\rfloor,
\qquad
k_n:=n-l_n.
\end{equation}
Since $r_0<C_{d,0}(\omega)$, the first block can distill, with zero
error, a maximally coherent state of dimension $L_{l_n}$ satisfying
\begin{equation}
\log L_{l_n}
=
l_n r_0+o(n).
\label{eq:exact-seed-rate-low}
\end{equation}
Here, if necessary, one may first choose an achievable zero-error rate
strictly larger than $r_0$ and then reduce the target dimension,
using the monotonicity of exact coherence distillation with respect to
the target size.

For the prescribed target sequence $J_n=2^{nr+o(n)}$, define
\begin{equation}
K_{k_n}
:=
\left\lceil
\frac{J_n}{L_{l_n}}
\right\rceil.
\label{eq:Kn-low-rate}
\end{equation}
Using Equations~\eqref{eq:low-rate-convex-combination}
and~\eqref{eq:exact-seed-rate-low}, we obtain
\begin{equation}
\lim_{n\to\infty}
\frac1{k_n}\log K_{k_n}
=
r_2.
\label{eq:effective-rate-low}
\end{equation}

Since $r_2<d_0=-\log p_*$, choose $\delta>0$ such that
\begin{equation}
r_2+\delta<d_0.
\end{equation}
For all sufficiently large $n$,
\begin{equation}
K_{k_n}
\leq
2^{k_n(r_2+\delta)}
<
2^{k_nd_0}
=
p_*^{-k_n}.
\label{eq:Kn-below-pstar}
\end{equation}
Consider
\begin{equation}
Q_{k_n}:=P_\omega^{\ox k_n}.
\end{equation}
Then
\begin{equation}
0\leq Q_{k_n}\leq\mathds{1},
\qquad
\tr(\omega^{\ox k_n}Q_{k_n})=1,
\end{equation}
and
\begin{align}
\Delta(Q_{k_n})
&\leq
p_*^{k_n}\mathds{1}
\leq
\frac1{K_{k_n}}\mathds{1},
\end{align}
where the last inequality follows from
Equation~\eqref{eq:Kn-below-pstar}. Hence $Q_{k_n}$ is a relaxed
zero-error test, so
\begin{equation}
\varepsilon_{\leq}
\left(
\omega^{\ox k_n},\frac1{K_{k_n}}
\right)
=0.
\end{equation}

Combining this relaxed zero-error test with the exact seed distilled
from the first block via Theorem~\ref{them:relax seed} gives an exact
zero-error protocol with target dimension
$L_{l_n}K_{k_n}\geq J_n$. By monotonicity with respect to the target
dimension,
\begin{equation}
\varepsilon_{=}
\left(
\omega^{\ox n},\frac1{J_n}
\right)
=0
\end{equation}
for all sufficiently large $n$. Therefore,
\begin{equation}
\lim_{n\to\infty}
-\frac1n
\log
\varepsilon_{=}
\left(
\omega^{\ox n},\frac1{J_n}
\right)
=
+\infty.
\label{eq:full-regime-below-d0}
\end{equation}

On the other hand, for every $r<d_0$,
\begin{equation}
H_r^\infty(\omega\Vert\mc I)=+\infty.
\end{equation}
Since the Petz R\'enyi relative entropy of coherence is additive,
\begin{equation}
H_r^\infty(\omega\Vert\mc I)
=
\sup_{\alpha\in(0,1)}
\frac{\alpha-1}{\alpha}
\left(
r-D_{\Petz,\alpha}(\omega\Vert\mc I)
\right).
\end{equation}
Thus Equation~\eqref{eq:coherence-rf-full-regime} holds throughout
$0<r<d_0$.

\medskip
\noindent
\textit{2) The region $d_0<r<D(\omega\Vert\mc I)$.}

Fix
\begin{equation}
d_0<r<D(\omega\Vert\mc I).
\end{equation}
Choose any
\begin{equation}
0<r_0<C_{d,0}(\omega).
\end{equation}
By Equation~\eqref{eq:zero-error-rate-d0-upper}, $r_0<d_0<r$.

For a sufficiently small constant $\theta>0$, split the $n$ copies as
\begin{equation}
l_n:=\lfloor\theta n\rfloor,
\qquad
k_n:=n-l_n.
\end{equation}
As above, the first block produces with zero error a maximally coherent
seed of dimension
\begin{equation}
L_{l_n}
=
2^{l_nr_0+o(n)}.
\label{eq:exact-seed-rate-high}
\end{equation}
For the prescribed target $J_n=2^{nr+o(n)}$, set
\begin{equation}
K_{k_n}
:=
\left\lceil
\frac{J_n}{L_{l_n}}
\right\rceil.
\end{equation}
The effective rate on the second block is
\begin{equation}
r_\theta
:=
\lim_{n\to\infty}
\frac1{k_n}\log K_{k_n}
=
\frac{r-\theta r_0}{1-\theta}.
\label{eq:rtheta-high}
\end{equation}
Since $r>d_0$ and $r_0<d_0$, we have
\begin{equation}
r_\theta>r>d_0.
\end{equation}
Moreover, for $\theta>0$ sufficiently small,
\begin{equation}
r_\theta<D(\omega\Vert\mc I).
\end{equation}

We next determine the exponent of the relaxed error on the second
block. Since
\begin{equation}
K_{k_n}
=
2^{k_nr_\theta+o(k_n)},
\end{equation}
for every sufficiently small $\delta>0$ and all sufficiently large
$n$,
\begin{equation}
2^{-k_n(r_\theta+\delta)}
\leq
\frac1{K_{k_n}}
\leq
2^{-k_n(r_\theta-\delta)}.
\end{equation}
By Lemma~\ref{lem:op2} and the monotonicity of the hypothesis testing
type-II error in the allowed type-I error,
\begin{align}
H_{r_\theta+\delta}^\infty(\omega\Vert\mc I)
&\leq
\liminf_{n\to\infty}
-\frac1{k_n}
\log
\varepsilon_{\leq}
\left(
\omega^{\ox k_n},\frac1{K_{k_n}}
\right)
\nonumber\\
&\leq
\limsup_{n\to\infty}
-\frac1{k_n}
\log
\varepsilon_{\leq}
\left(
\omega^{\ox k_n},\frac1{K_{k_n}}
\right)
\nonumber\\
&\leq
H_{r_\theta-\delta}^\infty(\omega\Vert\mc I).
\label{eq:rtheta-H-sandwich}
\end{align}
Because $r_\theta>d_0$, the regularized Hoeffding divergence is finite
and continuous at $r_\theta$. Letting $\delta$ tend to $0$ in
Equation~\eqref{eq:rtheta-H-sandwich} gives
\begin{equation}
\lim_{n\to\infty}
-\frac1{k_n}
\log
\varepsilon_{\leq}
\left(
\omega^{\ox k_n},\frac1{K_{k_n}}
\right)
=
H_{r_\theta}^\infty(\omega\Vert\mc I).
\label{eq:relaxed-rtheta-exponent}
\end{equation}

Theorem~\ref{them:relax seed} and the zero error of the first block
therefore imply
\begin{equation}
\varepsilon_{=}
\left(
\omega^{\ox n},\frac1{J_n}
\right)
\leq
\varepsilon_{\leq}
\left(
\omega^{\ox k_n},\frac1{K_{k_n}}
\right),
\end{equation}
and consequently
\begin{align}
\liminf_{n\to\infty}
-\frac1n
\log
\varepsilon_{=}
\left(
\omega^{\ox n},\frac1{J_n}
\right)
\geq
(1-\theta)
H_{r_\theta}^\infty(\omega\Vert\mc I).
\label{eq:full-regime-lower-theta}
\end{align}

Now let $\theta$ tend to $0$. From
Equation~\eqref{eq:rtheta-high},
\begin{equation}
r_\theta\longrightarrow r.
\end{equation}
Since $r>d_0$, $H_s^\infty(\omega\Vert\mc I)$ is continuous at $r$.
Therefore,
\begin{equation}
\lim_{\theta\to0}
(1-\theta)
H_{r_\theta}^\infty(\omega\Vert\mc I)
=
H_r^\infty(\omega\Vert\mc I),
\end{equation}
and Equation~\eqref{eq:full-regime-lower-theta} yields
\begin{equation}
\liminf_{n\to\infty}
-\frac1n
\log
\varepsilon_{=}
\left(
\omega^{\ox n},\frac1{J_n}
\right)
\geq
H_r^\infty(\omega\Vert\mc I).
\label{eq:full-regime-lower-final}
\end{equation}

For the converse, the relaxed problem satisfies
\begin{equation}
\varepsilon_{\leq}
\left(
\omega^{\ox n},\frac1{J_n}
\right)
\leq
\varepsilon_{=}
\left(
\omega^{\ox n},\frac1{J_n}
\right).
\end{equation}
Since $r>d_0$, the same threshold-sandwich argument used above gives
\begin{equation}
\lim_{n\to\infty}
-\frac1n
\log
\varepsilon_{\leq}
\left(
\omega^{\ox n},\frac1{J_n}
\right)
=
H_r^\infty(\omega\Vert\mc I).
\end{equation}
It follows that
\begin{equation}
\limsup_{n\to\infty}
-\frac1n
\log
\varepsilon_{=}
\left(
\omega^{\ox n},\frac1{J_n}
\right)
\leq
H_r^\infty(\omega\Vert\mc I).
\label{eq:full-regime-upper-final}
\end{equation}
Combining
Equations~\eqref{eq:full-regime-lower-final}
and~\eqref{eq:full-regime-upper-final} gives
\begin{equation}
\lim_{n\to\infty}
-\frac1n
\log
\varepsilon_{=}
\left(
\omega^{\ox n},\frac1{J_n}
\right)
=
H_r^\infty(\omega\Vert\mc I).
\end{equation}

Finally, additivity of the Petz R\'enyi relative entropy of coherence
gives
\begin{equation}
H_r^\infty(\omega\Vert\mc I)
=
\sup_{\alpha\in(0,1)}
\frac{\alpha-1}{\alpha}
\left(
r-D_{\Petz,\alpha}(\omega\Vert\mc I)
\right),
\end{equation}
which proves Equation~\eqref{eq:coherence-rf-full-regime}.
\end{proof}

\subsection{Zero-Error and One-Shot Distillation Rates}
\label{sec:Exact Zero-Error Distillation}
In this section, we focus on the distillation rate of coherence under MIO and DIO. 
In particular, we provide a structural characterization of the states with a positive zero-error distillation rate, and a closed-form expression for the zero-error distillation rate for all states.
Finally, we provide dimension-independent upper and lower bounds on the one-shot distillation rate in terms of the hypothesis testing relative entropy of coherence, which give a direct proof for the asymptotic distillation rate, whereas the existing proof (particularly the lower bound) relies on the operational relation with incoherent operations (IO)~\cite{PhysRevLett.113.140401,winter2016operational}.

For a quantum state $\omega$ on $\mathcal{H}$, when the free operations are $\mc O$, the \emph{zero-error distillation rate} $C_{d,0}(\omega)$ is defined as the supremum of all $r \ge 0$ for which there exist integers $K_n$ such that 
\begin{equation}
  \liminf_{n\to\infty}\frac{1}{n}\log K_n\geq r,~~~
  \fid(\omega^{\otimes n}\xrightarrow{\mc O}\Psi_{K_n})=1
\end{equation}
for all sufficiently large $n$. From the equivalence relation in Theorem~\ref{them:op1}, when the free operations are MIO or DIO, the constraint is equivalent to $\F_{=}\!\left(\omega^{\otimes n},K_n\right)=1$ for all sufficiently large $n$.

First, we still use the relaxed optimization problem as an intermediate tool, and examine the relaxed zero-error asymptotic rate induced by the relaxed optimization problem. This rate is defined as
supremum of all $r\geq 0$ for which there exist integers $K_n$ such that
\begin{equation}
  \liminf_{n\to\infty}\frac{1}{n}\log K_n\geq r,~~~
  \F_{\leq}\!\left(\omega^{\otimes n},K_n\right)=1
\end{equation}
for all sufficiently large $n$.

\begin{lemma}
\label{lem:upper bound for zero-error rate}
Let $\omega$ be a state on the Hilbert space $\mathcal{H}$, and let $P_\omega$ be the projection onto the support of $\omega$. 
Let
\begin{equation}
  p_*:=\max_i \langle i|P_\omega|i\rangle.
\end{equation}
Then, for every $n$ and integer $K_n$,
\begin{equation}
  \F_{\leq}\!\left(\omega^{\otimes n},K_n\right)=1
  \quad\Longleftrightarrow\quad
  K_n\leq p_*^{-n}.
\end{equation}
Consequently, the relaxed zero-error asymptotic rate is
\[
   -\log p_*=D_{\Petz,0}(\omega\|\mc I).
\]
\end{lemma}

\begin{proof}
Put $P:=P_\omega$. If $K_n\leq p_*^{-n}$, then
\begin{equation}
  \max_{i^n}
  \langle i^n|P^{\otimes n}|i^n\rangle
  =
  p_*^n
  \leq
  \frac{1}{K_n}.
\end{equation}
Thus $Q=P^{\otimes n}$ is relaxed feasible and has unit expectation
under $\omega^{\otimes n}$.

Conversely, suppose a relaxed feasible $Q$ has unit expectation. Since
$I-Q\geq 0$ and $\omega^{\otimes n}$ is strictly positive on the range of
$P^{\otimes n}$,
\begin{equation}
  P^{\otimes n}(I-Q)P^{\otimes n}=0.
\end{equation}
For a positive operator, a vanishing principal compression forces the
associated off-diagonal blocks to vanish. Hence $I-Q$ vanishes on the
range of $P^{\otimes n}$ and $Q\geq P^{\otimes n}$. It follows that
\begin{equation}
  \frac{1}{K_n}
  \geq
  \max_{i^n}\langle i^n|Q|i^n\rangle
  \geq
  \max_{i^n}
  \langle i^n|P^{\otimes n}|i^n\rangle
  =
  p_*^n.
\end{equation}
This is equivalent to $K_n\leq p_*^{-n}$. Finally, $D_{\Petz,0}(\omega\|\mc I)=-\log p_*$ is given by~\cite{Zhu_2017}.
\end{proof}

We next characterize which states admit a positive zero-error
distillation rate. This characterization, together with
Lemma~\ref{lem:upper bound for zero-error rate}, will allow us to give a complete characterization of the zero-error distillation rate for every finite-dimensional state.

\begin{theorem}
\label{them:coherence structure of state}
Let $\omega$ be a finite-dimensional state. Under MIO or DIO,
\begin{equation}
C_{d,0}(\omega)>0
\end{equation}
if and only if $\operatorname{supp}(\omega)$ contains no incoherent
basis vector.
\end{theorem}

\begin{proof}
Let $P_\omega$ be the support projection of $\omega$ and recall that
\begin{equation}
p_*:=\max_i\langle i|P_\omega|i\rangle.
\end{equation}
Since $P_\omega$ is a projection,
\begin{equation}
\langle i|P_\omega|i\rangle=1
\quad\Longleftrightarrow\quad
|i\rangle\in\operatorname{supp}(\omega).
\label{eq:support-basis-equivalence}
\end{equation}

Suppose first that $\operatorname{supp}(\omega)$ contains an incoherent
basis vector. Then $p_*=1$. By
Lemma~\ref{lem:upper bound for zero-error rate}, the relaxed zero-error
rate is
\begin{equation}
-\log p_*=0.
\end{equation}
Since the exact optimization problem is more restrictive than the
relaxed one, it follows that
\begin{equation}
C_{d,0}(\omega)=0.
\end{equation}

Conversely, suppose that $\operatorname{supp}(\omega)$ contains no
incoherent basis vector. By
Equation~\eqref{eq:support-basis-equivalence} and finite dimensionality,
\begin{equation}
p_*<1.
\end{equation}
Choose $n$ sufficiently large such that
\begin{equation}
p_*^n<\frac14,
\end{equation}
and set
\begin{equation}
Q:=I-P_\omega^{\otimes n}.
\end{equation}
We show that there exists $0\leq H\leq I$ satisfying
\begin{equation}
\Delta(H)=\frac12 I,
\qquad
HP_\omega^{\otimes n}=0.
\label{eq:H-required}
\end{equation}

For each incoherent product basis vector $|x\rangle$, we have
\begin{equation}
\langle x|P_\omega^{\otimes n}|x\rangle
\leq p_*^n<\frac14.
\end{equation}
Consider the linear equations
\begin{equation}
\sum_y
|\langle x|Q|y\rangle|^2d_y
=
\frac12 ,
\qquad\forall x.
\label{eq:diagonal-completion-linear}
\end{equation}
Since $Q$ is a projection,
\begin{align}
|\langle x|Q|x\rangle|^2
&=
\left(
1-\langle x|P_\omega^{\otimes n}|x\rangle
\right)^2,
\\
\sum_{y\neq x}|\langle x|Q|y\rangle|^2
&=
\langle x|P_\omega^{\otimes n}|x\rangle
\left(
1-\langle x|P_\omega^{\otimes n}|x\rangle
\right).
\end{align}
Hence the coefficient matrix in
Equation~\eqref{eq:diagonal-completion-linear} is strictly diagonally
dominant, since
\begin{align}
&|\langle x|Q|x\rangle|^2
-
\sum_{y\neq x}|\langle x|Q|y\rangle|^2
\nonumber\\
&\qquad=
\left(
1-\langle x|P_\omega^{\otimes n}|x\rangle
\right)
\left(
1-2\langle x|P_\omega^{\otimes n}|x\rangle
\right)
>0.
\end{align}
Therefore, Equation~\eqref{eq:diagonal-completion-linear} has a unique
real solution.

Writing $d_y=\frac12+e_y$,
Equation~\eqref{eq:diagonal-completion-linear} becomes
\begin{equation}
\sum_y|\langle x|Q|y\rangle|^2e_y
=
\frac12\langle x|P_\omega^{\otimes n}|x\rangle.
\end{equation}
Choose $x$ such that $|e_x|=\max_y|e_y|$. The preceding diagonal
dominance gives
\begin{align}
\frac12 p_*^n
&\geq
\frac12\langle x|P_\omega^{\otimes n}|x\rangle
\nonumber\\
&=
\left|
\sum_y|\langle x|Q|y\rangle|^2e_y
\right|
\nonumber\\
&\geq
(1-p_*^n)(1-2p_*^n)\max_y|e_y|.
\end{align}
Thus
\begin{equation}
\max_y|e_y|
\leq
\frac{p_*^n}
{2(1-p_*^n)(1-2p_*^n)}
<\frac12,
\end{equation}
where the last inequality follows from $p_*^n<1/4$. Hence
\begin{equation}
0<d_x<1
\qquad\forall x.
\end{equation}

Define
\begin{equation}
H
:=
Q\left(\sum_xd_x|x\rangle\langle x|\right)Q.
\end{equation}
Since $0\leq d_x\leq1$,
\begin{equation}
0\leq H\leq Q\leq I,
\end{equation}
and Equation~\eqref{eq:diagonal-completion-linear} gives
\begin{equation}
\Delta(H)=\frac12 I.
\end{equation}
Moreover, $QP_\omega^{\otimes n}=0$, and hence
\begin{equation}
HP_\omega^{\otimes n}=0.
\end{equation}
Therefore $G:=I-H$ satisfies
\begin{equation}
0\leq G\leq I,
\qquad
\Delta(G)=\frac12 I,
\qquad
GP_\omega^{\otimes n}=P_\omega^{\otimes n}.
\end{equation}
Since $\omega^{\otimes n}$ is supported on $P_\omega^{\otimes n}$,
\begin{equation}
\operatorname{Tr}(\omega^{\otimes n}G)=1.
\end{equation}
By Theorem~\ref{them:op1},
\begin{equation}
\fid(\omega^{\otimes n}\xrightarrow{\mc O}\Psi_2)=1.
\end{equation}
Repeating this exact protocol independently over blocks of $n$ copies
yields one coherent bit per block, and therefore
\begin{equation}
C_{d,0}(\omega)\geq\frac1n>0.
\end{equation}
This proves the converse implication.
\end{proof}

We are now ready to determine the zero-error distillation rate exactly.
Remarkably, the relaxed upper bound of
Lemma~\ref{lem:upper bound for zero-error rate} is asymptotically
achievable for every state.

\begin{theorem}[Zero-error coherence distillation rate]
\label{them:coherence zero-error rate}
Let $\omega$ be a finite-dimensional state. Under MIO or DIO, the
zero-error distillation rate is
\begin{equation}
C_{d,0}(\omega)
=
D_{\Petz,0}(\omega\Vert\mc I).
\label{eq:exact-zero-error-coherence-rate}
\end{equation}
\end{theorem}

\begin{proof}
By Lemma~\ref{lem:upper bound for zero-error rate}, the relaxed
zero-error asymptotic rate is
\begin{equation}
D_{\Petz,0}(\omega\Vert\mc I)
=
-\log p_*.
\end{equation}
Since the exact optimization problem is more restrictive than the
relaxed one, we immediately obtain
\begin{equation}
C_{d,0}(\omega)
\leq
D_{\Petz,0}(\omega\Vert\mc I).
\label{eq:zero-error-coherence-upper}
\end{equation}

It remains to prove the reverse inequality. We distinguish two cases.

If
\begin{equation}
D_{\Petz,0}(\omega\Vert\mc I)=0,
\end{equation}
then Equation~\eqref{eq:zero-error-coherence-upper}, together with the
non-negativity of $C_{d,0}(\omega)$, immediately gives
\begin{equation}
C_{d,0}(\omega)
=
D_{\Petz,0}(\omega\Vert\mc I)
=
0.
\end{equation}

Suppose now that
\begin{equation}
D_{\Petz,0}(\omega\Vert\mc I)>0.
\end{equation}
Since
\begin{equation}
D_{\Petz,0}(\omega\Vert\mc I)
=
-\log p_*,
\end{equation}
we have $p_*<1$. By Theorem~\ref{them:coherence structure of state},
this implies
\begin{equation}
C_{d,0}(\omega)>0.
\end{equation}
Fix
\begin{equation}
0<r_0<C_{d,0}(\omega).
\end{equation}
We show that every
\begin{equation}
r<D_{\Petz,0}(\omega\Vert\mc I)
\end{equation}
is achievable with zero error.

If $r\leq r_0$, this follows directly from the achievability of
$r_0$. Suppose therefore that
\begin{equation}
r_0<r<D_{\Petz,0}(\omega\Vert\mc I).
\end{equation}
Choose
\begin{equation}
r<r_2<D_{\Petz,0}(\omega\Vert\mc I)
\end{equation}
and set
\begin{equation}
\theta
:=
\frac{r_2-r}{r_2-r_0}
\in(0,1),
\qquad
r=\theta r_0+(1-\theta)r_2.
\label{eq:show r is right}
\end{equation}

Split $n$ copies into two blocks of lengths
\begin{equation}
m_n=\lfloor\theta n\rfloor,
\qquad
\ell_n=n-m_n.
\end{equation}
Since $r_0<C_{d,0}(\omega)$, the first block admits an exact
zero-error distillation protocol with target dimension
$L_{m_n}$ satisfying
\begin{equation}
\frac{1}{m_n}\log L_{m_n}
\geq r_0-o(1).
\label{eq:first-exact-block-rate}
\end{equation}

On the second block, choose integers $K_{\ell_n}$ satisfying
\begin{equation}
\frac{1}{\ell_n}\log K_{\ell_n}
=
r_2-o(1).
\end{equation}
Since
\begin{equation}
r_2<D_{\Petz,0}(\omega\Vert\mc I)
=-\log p_*,
\end{equation}
we have, for all sufficiently large $n$,
\begin{equation}
K_{\ell_n}
\leq
p_*^{-\ell_n}.
\end{equation}
Lemma~\ref{lem:upper bound for zero-error rate} therefore provides a
relaxed zero-error test on the second block:
\begin{equation}
\F_{\leq}
\left(
\omega^{\otimes\ell_n},
K_{\ell_n}
\right)
=1.
\end{equation}

Applying Theorem~\ref{them:relax seed} to the exact test on the first
block and the relaxed test on the second block yields an exact
zero-error protocol on all $n$ copies with target dimension
$L_{m_n}K_{\ell_n}$. Its asymptotic rate satisfies
\begin{align}
\liminf_{n\to\infty}
\frac1n\log(L_{m_n}K_{\ell_n})
&\geq
\theta r_0+(1-\theta)r_2
=
r,
\end{align}
where the last equality follows from
Equation~\eqref{eq:show r is right}. Thus every
$r<D_{\Petz,0}(\omega\Vert\mc I)$ is achievable, and hence
\begin{equation}
C_{d,0}(\omega)
\geq
D_{\Petz,0}(\omega\Vert\mc I).
\end{equation}
Together with Equation~\eqref{eq:zero-error-coherence-upper}, this
proves Equation~\eqref{eq:exact-zero-error-coherence-rate}.
\end{proof}

Combining Theorems~\ref{them:coherence structure of state}
and~\ref{them:coherence zero-error rate}, we obtain the equivalent
characterizations
\begin{align}
&C_{d,0}(\omega)>0\nonumber\\
&\iff
D_{\Petz,0}(\omega\Vert\mc I)>0
\nonumber\\
&\iff
\operatorname{supp}(\omega)
\text{ contains no incoherent basis vector}.\nonumber
\end{align}
Thus the zero-error distillation rate is completely determined by the
position of the support of $\omega$ relative to the incoherent basis.
In particular, every full-rank state has zero zero-error distillable
coherence, while rank deficiency alone is not sufficient for a positive
zero-error rate: a rank-deficient state still has
$C_{d,0}(\omega)=0$ whenever its support contains an incoherent basis
vector. Conversely, every state whose support contains no incoherent
basis vector has a strictly positive zero-error rate; in particular,
this includes every coherent pure state.

Finally, we examine the one-shot distillation rate of coherence distillation. In particular, we obtain dimension-independent upper and lower bounds on the one-shot distillation rate in terms of the hypothesis testing relative entropy of coherence. These bounds directly yield a proof of the asymptotic distillation rate, whereas the existing proof (especially the lower bound) relies on the operational relation with IO~\cite{PhysRevLett.113.140401,winter2016operational}.

Let $\omega$ be a state on the Hilbert space $\mathcal{H}$. For $\mathcal O\in\{\mathrm{MIO},\mathrm{DIO}\}$, define the
one-shot distillation rate by
\begin{equation}
  C_{d,\mathcal O}^{(1),\varepsilon}(\omega)
  :=
  \log\max\left\{
    K\in\mathbb N:\err(\omega\!\!\xrightarrow{\mathcal{O}}\! \Psi_K)
    \leq\varepsilon
  \right\}.
  \label{eq:operational-one-shot-rate}
\end{equation}

To investigate the one-shot distillation rate, we still resort to the two optimization problems
\begin{equation}
  \beta_{=}^{\varepsilon}(\omega)
  :=
  \min\left\{
    c\in[0,1]:
    \begin{array}{c}
      0\leq G\leq I,\\
      \operatorname{Tr}(\omega G)\geq 1-\varepsilon,\\
      \Delta(G)=cI
    \end{array}
  \right\},
  \label{eq:beta-exact}
\end{equation}
\begin{equation}
  \beta_{\leq}^{\varepsilon}(\omega)
  :=
  \min_{\substack{0\leq Q\leq I\\
                   \operatorname{Tr}(\omega Q)\geq 1-\varepsilon}}
  \lVert\Delta(Q)\rVert_{\infty}.
  \label{eq:beta-relaxed}
\end{equation}
They are strictly positive for
$\varepsilon<1$.  We may therefore set
\begin{align}
  \widetilde{C}_{\mathrm H}^{\varepsilon}(\omega)
  &:= -\log\beta_{=}^{\varepsilon}(\omega),
  \label{eq:C-H-tilde}\\
  C_{\mathrm H}^{\varepsilon}(\omega)
  &:= -\log\beta_{\leq}^{\varepsilon}(\omega).
  \label{eq:C-H}
\end{align}

For a positive semidefinite alternative $\sigma$, recall that
\begin{equation}
  D_{\mathrm H}^{\varepsilon}(\omega\Vert\sigma)
  :=
  -\log
  \min_{\substack{0\leq Q\leq I\\
                   \operatorname{Tr}(\omega Q)\geq 1-\varepsilon}}
  \operatorname{Tr}(Q\sigma).
  \label{eq:D-H}
\end{equation}
The second expression is the standard hypothesis testing relative entropy of
coherence,
\begin{equation}
  C_{\mathrm H}^{\varepsilon}(\omega)
  =
  \min_{\delta\in\mathcal I}
  D_{\mathrm H}^{\varepsilon}(\omega\Vert\delta).
  \label{eq:C-H-composite}
\end{equation}
For completeness, if
\begin{equation}
  \mathcal T_{\varepsilon}(\omega)
  :=
  \left\{
    Q:0\leq Q\leq I,\ 
    \operatorname{Tr}(\omega Q)\geq 1-\varepsilon
  \right\},
  \label{eq:T-epsilon}
\end{equation}
then the minimax theorem gives
\begin{align}
  2^{-C_{\mathrm H}^{\varepsilon}(\omega)}
  &=
  \max_{\delta\in\mathcal I}
  \min_{Q\in\mathcal T_{\varepsilon}(\omega)}
  \operatorname{Tr}(Q\delta)
  \notag\\
  &=
  \min_{Q\in\mathcal T_{\varepsilon}(\omega)}
  \max_{\delta\in\mathcal I}
  \operatorname{Tr}(Q\delta)
  \notag\\
  &=
  \min_{Q\in\mathcal T_{\varepsilon}(\omega)}
  \lVert\Delta(Q)\rVert_{\infty}
  =
  \beta_{\leq}^{\varepsilon}(\omega).
  \label{eq:minimax-beta}
\end{align}
Here both domains are compact and convex, the objective is bilinear, and the
third equality uses
$\max_{\delta\in\mathcal I}\operatorname{Tr}(Q\delta)
=\max_x\langle x|Q|x\rangle$.

\begin{theorem}
\label{them:one-shot-coherence}
Let $\omega$ be a finite-dimensional state and
$0\leq\varepsilon'<\varepsilon<1$. Define
\begin{equation}
  s:=\sqrt{\varepsilon}-\sqrt{\varepsilon'},~~~\kappa(\varepsilon,\varepsilon')
  :=
  -\log\bigl[s(2-s)\bigr].
  \label{eq:smoothing-gap}
\end{equation}
Then
\begin{equation}
    \max\left\{
      0,
      C_{\mathrm H}^{\varepsilon'}(\omega)
      -\kappa(\varepsilon,\varepsilon')
    \right\}
    \leq
    \widetilde{C}_{\mathrm H}^{\varepsilon}(\omega)
    \leq
    C_{\mathrm H}^{\varepsilon}(\omega).
  \label{eq:one-shot-Hermitian-PSD}
\end{equation}
Furthermore,
\begin{equation}
  \max\left\{
    0,
    C_{\mathrm H}^{\varepsilon'}(\omega)
    -\kappa(\varepsilon,\varepsilon')-1
  \right\}
  \leq
  C_{d,\mathcal O}^{(1),\varepsilon}(\omega)
  \leq
  C_{\mathrm H}^{\varepsilon}(\omega).
  \label{eq:operational-one-shot-comparison}
\end{equation}
\end{theorem}
\begin{proof}
Let us first prove Equation~\eqref{eq:one-shot-Hermitian-PSD}. The upper bound is straightforward.
Since any test with $\Delta(G)=cI$ satisfies $\|\Delta(G)\|_\infty=c$, minimizing over the respective test sets yields
\begin{equation}
  \beta_{\leq}^{\varepsilon}(\omega) \le \beta_{=}^{\varepsilon}(\omega),
  \label{eq:beta-relaxation-order}
\end{equation}
and hence $\widetilde{C}_{\mathrm H}^{\varepsilon}(\omega) \le C_{\mathrm H}^{\varepsilon}(\omega)$.

Next, we prove the lower bound.
Set
\begin{equation}
  r:=C_{\mathrm H}^{\varepsilon'}(\omega),
  \qquad
  c:=2^{-r}=\beta_{\leq}^{\varepsilon'}(\omega).
  \label{eq:r-and-c}
\end{equation}
Since $\varepsilon'<1$, we have $c>0$. The case $c=1$ gives $r=0$, and because $\kappa(\varepsilon,\varepsilon')>0$, the lower bound in Equation~\eqref{eq:one-shot-Hermitian-PSD} is trivially zero. Hence we may assume $0<c<1$.

Let $Q$ attain the minimum in Equation~\eqref{eq:beta-relaxed}. Then
\begin{equation}
  0\leq Q\leq I,
  \qquad
  \operatorname{Tr}(\omega Q)\geq1-\varepsilon',
  \qquad
  \Delta(Q)\leq cI.
  \label{eq:optimal-Q}
\end{equation}
For $H:=I-Q$, Equation~\eqref{eq:optimal-Q} is equivalent to
\begin{equation}
  0\leq H\leq I,
  \qquad
  \operatorname{Tr}(\omega H)\leq\varepsilon',
  \qquad
  \Delta(H)\geq(1-c)I.
  \label{eq:H-constraints}
\end{equation}
It follows directly from Equation~\eqref{eq:to find H} that
\begin{equation}
  \varepsilon_{\leq}(\omega,c)\leq\varepsilon'.
  \label{eq:epsilon-geq-bound}
\end{equation}
Define
\begin{equation}
  \eta(c):=1-\sqrt{1-c}.
  \label{eq:eta}
\end{equation}
Lemma~\ref{lem: lower and up bound coherence} and Equation~\eqref{eq:epsilon-geq-bound} give
\begin{equation}
  \varepsilon_{=}(\omega,c)
  \leq
  \left(\sqrt{\varepsilon'}+\eta(c)\right)^2.
  \label{eq:epsilon-exact-bound}
\end{equation}
More explicitly, there is an operator $H_c$ satisfying
\begin{equation}
  0\leq H_c\leq I,~
  \Delta(H_c)=(1-c)I,~
  \operatorname{Tr}(\omega H_c)
  \leq
  \left(\sqrt{\varepsilon'}+\eta(c)\right)^2.
  \label{eq:Hc}
\end{equation}
Its complement $G_c:=I-H_c$ obeys
\begin{equation}
  0\leq G_c\leq I,~\Delta(G_c)=cI,~\operatorname{Tr}(\omega G_c)
  \geq 1-\left(\sqrt{\varepsilon'}+\eta(c)\right)^2.
  \label{eq:Gc}
\end{equation}

Suppose first that $r\geq\kappa(\varepsilon,\varepsilon')$. Then
\begin{equation}
  c=2^{-r}
  \leq
  2^{-\kappa(\varepsilon,\varepsilon')}
  =s(2-s).
  \label{eq:c-smoothing-bound}
\end{equation}
The function $\eta(c)=1-\sqrt{1-c}$ is increasing on $[0,1]$, so
\begin{align}
  \eta(c)
  &\leq
  1-\sqrt{1-s(2-s)}
  \notag\\
  &=
  1-\sqrt{(1-s)^2}
  =s.
  \label{eq:eta-smoothing-bound}
\end{align}
Equations~\eqref{eq:epsilon-exact-bound} and~\eqref{eq:eta-smoothing-bound} imply
\begin{equation}
  \varepsilon_{=}(\omega,c)
  \leq
  \left(\sqrt{\varepsilon'}+s\right)^2
  =\varepsilon.
  \label{eq:epsilon-prescribed}
\end{equation}
Taking complements as in Equation~\eqref{eq:Gc}, or equivalently using the definition of
$\beta_{=}^{\varepsilon}$, now gives
\begin{equation}
  \beta_{=}^{\varepsilon}(\omega)\leq c,
  \qquad
  \widetilde{C}_{\mathrm H}^{\varepsilon}(\omega)
  \geq -\log c=r.
  \label{eq:beta-exact-c}
\end{equation}
If $r\geq\kappa(\varepsilon,\varepsilon')$, then Equation~\eqref{eq:beta-exact-c} in particular
yields
\begin{equation}
  \widetilde{C}_{\mathrm H}^{\varepsilon}(\omega)
  \geq
  r-\kappa(\varepsilon,\varepsilon').
  \label{eq:C-H-lower-case-one}
\end{equation}
If $r<\kappa(\varepsilon,\varepsilon')$, the same conclusion with a maximum at
zero follows from
\begin{equation}
  \widetilde{C}_{\mathrm H}^{\varepsilon}(\omega)\geq0.
  \label{eq:C-H-nonnegative}
\end{equation}
The two cases prove
the lower bound in Equation~\eqref{eq:one-shot-Hermitian-PSD}.

Finally, we prove Equation~\eqref{eq:operational-one-shot-comparison}.
Fix $0\leq\varepsilon<1$ and write
\begin{equation}
  b:=\beta_{=}^{\varepsilon}(\omega).
  \label{eq:b}
\end{equation}
We have $0 < b \le 1 - \varepsilon \le 1$. If $b = 1$, then no $K \ge 2$ is achievable, because such a target would require an exact test with diagonal entries $1/K$, implying $b \le 1/K < 1$, a contradiction. Hence we may assume $b < 1$, and let $G$ be an optimizer with $\Delta(G) = b I$.

For any $c'\in[b,1]$, define
\begin{equation}
  \lambda:=\frac{1-c'}{1-b},
  \qquad
  G':=\lambda G+(1-\lambda)I.
  \label{eq:G-prime}
\end{equation}
Since $0\leq\lambda\leq1$, we have $0\leq G'\leq I$, and a direct calculation
gives
\begin{align}
  \Delta(G')
  &=
  \bigl[\lambda b+(1-\lambda)\bigr]I
  =c'I,
  \label{eq:G-prime-diagonal}\\
  \operatorname{Tr}(\omega G')
  &\geq
  \lambda(1-\varepsilon)+(1-\lambda)
  =1-\lambda\varepsilon
  \geq1-\varepsilon.
  \label{eq:G-prime-acceptance}
\end{align}
Thus every $c'\geq b$ is feasible in the exact testing problem. Conversely,
the definition of $b$ shows that no $c'<b$ is feasible. By Theorem~\ref{them:op1}, a target dimension $K$ is
therefore achievable with error at most $\varepsilon$ if and only if
\begin{equation}
  \frac{1}{K}\geq b.
  \label{eq:target-iff-b}
\end{equation}
Consequently,
\begin{equation}
  \max\left\{
    K\in\mathbb N:
    \varepsilon_{\mathcal O}(\omega,1/K)\leq\varepsilon
  \right\}
  =
  \left\lfloor\frac{1}{b}\right\rfloor
  =
  \left\lfloor
    2^{\widetilde{C}_{\mathrm H}^{\varepsilon}(\omega)}
  \right\rfloor.
  \label{eq:integer-rounding}
\end{equation}

Finally, for every $x\geq0$, one has
$\lfloor2^x\rfloor\geq2^{x-1}$ and $\lfloor2^x\rfloor\geq1$, whence
\begin{equation}
  \log\lfloor2^x\rfloor\geq\max\{0,x-1\}.
  \label{eq:floor-bound}
\end{equation}
Combining Equations~\eqref{eq:integer-rounding},~\eqref{eq:floor-bound}, and~\eqref{eq:one-shot-Hermitian-PSD} proves Equation~\eqref{eq:operational-one-shot-comparison}.
\end{proof}

The preceding one-shot bounds immediately recover the asymptotic distillation rate of coherence under MIO and DIO. Recall that, for $\mathcal O\in\{\mathrm{MIO},\mathrm{DIO}\}$, the asymptotic distillation rate is defined by
\begin{equation}
C_{d,\mathcal O}(\omega)
:=
\lim_{\varepsilon\to 0}
\liminf_{n\to\infty}
\frac{1}{n}
C_{d,\mathcal O}^{(1),\varepsilon}
\bigl(\omega^{\otimes n}\bigr),
\label{eq:asymptotic-distillation-rate}
\end{equation}
which characterizes the optimal number of coherent bits that can be distilled per copy from i.i.d.\ copies of $\omega$ as the distillation error vanishes asymptotically.

\begin{corollary}
\label{cor:asymptotic-coherence-distillation}
Let $\omega$ be a finite-dimensional state and $\mathcal O\in\{\mathrm{MIO},\mathrm{DIO}\}$. Then, for every $0<\varepsilon<1$,
\begin{equation}
\lim_{n\to\infty}
\frac{1}{n}
C_{d,\mathcal O}^{(1),\varepsilon}
\bigl(\omega^{\otimes n}\bigr)
=
D(\omega\|\mathcal I).
\label{eq:asymptotic-coherence-distillation}
\end{equation}
Consequently, the asymptotic distillable coherence under both MIO and DIO is given by the relative entropy of coherence,
\begin{equation}
C_{d,\mathcal O}(\omega)
=
D(\omega\|\mathcal I).
\label{eq:asymptotic-distillable-coherence}
\end{equation}
\end{corollary}

\begin{proof}
Fix $0<\varepsilon'<\varepsilon<1$. Applying Theorem~\ref{them:one-shot-coherence} to $\omega^{\otimes n}$ yields
\begin{align}
&\max\Bigl\{0,\,C_{\mathrm H}^{\varepsilon'}\bigl(\omega^{\otimes n}\bigr)-\kappa(\varepsilon,\varepsilon')-1\Bigr\}\nonumber\\
&\qquad\leq C_{d,\mathcal O}^{(1),\varepsilon}\bigl(\omega^{\otimes n}\bigr)\leq C_{\mathrm H}^{\varepsilon}\bigl(\omega^{\otimes n}\bigr).
\label{eq:one-shot-bound-iid}
\end{align}
Dividing by $n$, we obtain
\begin{align}
&\max\Bigl\{0,\,\frac{1}{n}C_{\mathrm H}^{\varepsilon'}\bigl(\omega^{\otimes n}\bigr)-\frac{\kappa(\varepsilon,\varepsilon')+1}{n}\Bigr\}\nonumber\\
&\qquad\leq \frac{1}{n}C_{d,\mathcal O}^{(1),\varepsilon}\bigl(\omega^{\otimes n}\bigr)\leq \frac{1}{n}C_{\mathrm H}^{\varepsilon}\bigl(\omega^{\otimes n}\bigr).
\label{eq:normalized-one-shot-bound}
\end{align}

For any fixed $0<\eta<1$, the asymptotic equipartition property for the hypothesis testing relative entropy of coherence~\cite{berta2021composite} gives
\begin{equation}
\lim_{n\to\infty}
\frac{1}{n}
C_{\mathrm H}^{\eta}\bigl(\omega^{\otimes n}\bigr)
=
D(\omega\|\mathcal I).
\label{eq:generalized-AEP-coherence}
\end{equation}
Since $\kappa(\varepsilon,\varepsilon')$ is independent of $n$, we have
\begin{equation}
\lim_{n\to\infty}
\frac{\kappa(\varepsilon,\varepsilon')+1}{n}
=0.
\end{equation}
Applying these limits to Eq.~\eqref{eq:normalized-one-shot-bound} yields
\begin{align}
\max\bigl\{0,\,D(\omega\|\mathcal I)\bigr\}
\leq \lim_{n\to\infty}\frac{1}{n}
C_{d,\mathcal O}^{(1),\varepsilon}\bigl(\omega^{\otimes n}\bigr)
\leq D(\omega\|\mathcal I).
\end{align}

Because $D(\omega\|\mathcal I)\ge 0$, the maximum simplifies to $D(\omega\|\mathcal I)$. Hence,
\begin{equation}
\lim_{n\to\infty}
\frac{1}{n}
C_{d,\mathcal O}^{(1),\varepsilon}\bigl(\omega^{\otimes n}\bigr)
=
D(\omega\|\mathcal I),
\end{equation}
which establishes Eq.~\eqref{eq:asymptotic-coherence-distillation}.

Since the right-hand side is independent of the fixed error $0<\varepsilon<1$, taking $\varepsilon\to0$ immediately recovers the usual vanishing-error asymptotic distillation rate,
\begin{equation}
C_{d,\mathcal O}(\omega)
=
D(\omega\|\mathcal I),
\end{equation}
for both $\mathcal O=\mathrm{MIO}$ and $\mathcal O=\mathrm{DIO}$. This completes the proof.
\end{proof}

\section{Conclusion and Discussion}
\label{sec: conclusion}
We have presented a comprehensive framework for universal entanglement and coherence distillation in the black-box setting, which marks a crucial step towards modelling realistic scenarios that account for both environmental uncertainty and non-i.i.d. correlations.
As a preliminary contribution, we resolved the inequivalence of two regularized generalized quantum Hoeffding divergences by explicit counterexamples and established fundamental properties of the Hoeffding divergence and anti-divergence.

For entanglement distillation, we establish exact relationships between the distillation error of universal entanglement distillation under different free operations and the optimal type‑II error of composite correlated hypothesis testing in different settings, in both the finite‑blocklength and asymptotic exponent analyses. Under non-entangling and PPT-preserving operations, the reliability function is characterized by the regularized quantum Hoeffding divergence; under dually free operations, the zero-rate error exponent and distillable entanglement are characterized by the regularized relative entropy under restricted measurements --- endowing these quantities with new operational interpretations. Our findings complement or recover the corresponding results previously established in the literature.
Additionally, we proposed five concrete noise models whose induced black-box sequences naturally satisfy the required properties, thereby bridging abstract theory and physical reality for both entanglement and coherence.

For coherence distillation, MIO and DIO share an exact one-shot semidefinite characterization. By introducing a relaxed optimization problem, we establish a connection between the distillation error and composite hypothesis testing, and derive the reliability function above the critical rate. Using a seed-assisted construction, the reliability function is further extended to the full range of distillation rates. Additionally, the positive zero-error distillation rate is fully characterized, and dimension-independent upper and lower bounds on the one-shot distillation rate are provided within the framework of hypothesis testing relative entropy, thereby directly proving the asymptotic distillation rate.

Investigating the continuity of the regularized Petz–Rényi relative entropy at $\alpha=1$ is an important problem, as it relates to recovering the zero-rate error exponent --- the final piece in the complete picture for entanglement distillation under non-entangling and PPT-preserving operations.
Although~\cite{fang2025errorexponentsquantumstate} provides counterexamples demonstrating the breakdown of continuity in the general case, no explicit counterexample has been identified for scenarios where the first component is the set of free states, which correspond precisely to the aforementioned problem.
Notably, the antisymmetric Werner state---the universal counterexample in entanglement theory---remains continuous in this setting~\cite{fang2025errorexponentsquantumstate}. Furthermore, when the set of free states appears as the second component, Grinko and Lami~\cite{grinko2026samplecomplexityquantumresource} recently proved the continuity.
Additionally, the error exponent and strong converse exponent of hypothesis testing under restricted measurements remain interesting open problems.
Finally, extending the present results to the distillation of multipartite entanglement and other resources offers a promising direction for future work.

\section*{Acknowledgment}

Z.L. is supported by the National Natural Science Foundation of China (Grants No. 62571166, No. 12031004, No. 12371138 and No. W2441002).
K.L. is supported by the National Natural Science Foundation of China (Grants No. 62571166 and No. 12031004).
K.F. is supported in part by the National Natural Science Foundation of China (Grants No. 92470113 and No. 12404569), the Shenzhen Science and Technology Program (Grants No. QNXMB20250701091826036 and No. JCYJ20240813113519025).

{\appendix

\begin{proof}[Proof of Proposition~\ref{prop: Hoeffding}]
For convenience, for every $\alpha\in(0,1)$ define
\begin{align}
h_{\alpha,r}(\rho_n\Vert\sigma_n)
:=
\frac{\alpha-1}{\alpha}
\left(
nr-D_{\Petz,\alpha}(\rho_n\Vert\sigma_n)
\right),
\end{align}
so that
\begin{equation}
H_{n,r}(\rho_n\Vert\sigma_n)
=
\sup_{\alpha\in(0,1)}
h_{\alpha,r}(\rho_n\Vert\sigma_n).
\end{equation}

\emph{1) Data-processing inequality.}
For every $\alpha\in(0,1)$, the Petz R\'enyi divergence satisfies the
data-processing inequality
\begin{equation}
D_{\Petz,\alpha}
\bigl(
\Lambda(\rho_n)\Vert\Lambda(\sigma_n)
\bigr)
\leq
D_{\Petz,\alpha}(\rho_n\Vert\sigma_n).
\end{equation}
Since $(1-\alpha)/\alpha>0$, it follows that
\begin{equation}
h_{\alpha,r}
\bigl(
\Lambda(\rho_n)\Vert\Lambda(\sigma_n)
\bigr)
\leq
h_{\alpha,r}(\rho_n\Vert\sigma_n).
\end{equation}
Taking the supremum over $\alpha\in(0,1)$ proves
\begin{equation}
H_{n,r}
\bigl(
\Lambda(\rho_n)\Vert\Lambda(\sigma_n)
\bigr)
\leq
H_{n,r}(\rho_n\Vert\sigma_n).
\end{equation}

\emph{2) Joint convexity.}
For $0<\alpha<1$, let
\begin{equation}
Q_\alpha(\rho\Vert\sigma)
:=
\tr(\rho^\alpha\sigma^{1-\alpha}).
\end{equation}
By the joint concavity of $Q_\alpha$,
\begin{equation}
Q_\alpha
\left(
\sum_i p_i\rho_n^i
\Big\Vert
\sum_i p_i\sigma_n^i
\right)
\geq
\sum_i p_i
Q_\alpha(\rho_n^i\Vert\sigma_n^i).
\end{equation}
Since $-\log x$ is convex and decreasing,
\begin{align}
D_{\Petz,\alpha}&
\left(
\sum_i p_i\rho_n^i
\Big\Vert
\sum_i p_i\sigma_n^i
\right)\\
&=
-\frac{1}{1-\alpha}
\log
Q_\alpha
\left(
\sum_i p_i\rho_n^i
\Big\Vert
\sum_i p_i\sigma_n^i
\right)
\nonumber\\
&\leq
-\frac{1}{1-\alpha}
\log
\left(
\sum_i p_i
Q_\alpha(\rho_n^i\Vert\sigma_n^i)
\right)
\nonumber\\
&\leq
\sum_i p_i
D_{\Petz,\alpha}(\rho_n^i\Vert\sigma_n^i).
\end{align}
Hence
\begin{align}
h_{\alpha,r}
\left(
\sum_i p_i\rho_n^i
\Big\Vert
\sum_i p_i\sigma_n^i
\right)
\leq
\sum_i p_i
h_{\alpha,r}(\rho_n^i\Vert\sigma_n^i).
\end{align}
Taking the supremum over $\alpha$ and using
$\sup_\alpha\sum_i p_i a_{i,\alpha}
\leq
\sum_i p_i\sup_\alpha a_{i,\alpha}$ gives
\begin{align}
H_{n,r}
\left(
\sum_i p_i\rho_n^i
\Big\Vert
\sum_i p_i\sigma_n^i
\right)
\leq
\sum_i p_i
H_{n,r}(\rho_n^i\Vert\sigma_n^i).
\end{align}

\emph{3) Continuity, convexity, and monotonicity in $r$.}
Set
\begin{equation}
d_0:=D_{\Petz,0}(\rho_n\Vert\sigma_n).
\end{equation}
Since
\begin{equation}
\lim_{\alpha\to0}
D_{\Petz,\alpha}(\rho_n\Vert\sigma_n)
=
d_0,
\end{equation}
if $nr<d_0$, then there exist $\delta>0$ and $\alpha_0>0$ such that
\begin{equation}
D_{\Petz,\alpha}(\rho_n\Vert\sigma_n)-nr
\geq\delta,
\qquad
0<\alpha<\alpha_0.
\end{equation}
Consequently,
\begin{align}
h_{\alpha,r}(\rho_n\Vert\sigma_n)
&=
\frac{1-\alpha}{\alpha}
\left(
D_{\Petz,\alpha}(\rho_n\Vert\sigma_n)-nr
\right)
\nonumber\\
&\geq
\frac{1-\alpha}{\alpha}\delta
\longrightarrow+\infty
\end{align}
as $\alpha\to0$. Hence
\begin{equation}
H_{n,r}(\rho_n\Vert\sigma_n)=+\infty.
\end{equation}

Suppose now that $nr>d_0$. By the continuity of
$D_{\Petz,\alpha}(\rho_n\Vert\sigma_n)$ at $\alpha=0$, there exists
$\alpha_0>0$ such that
\begin{equation}
D_{\Petz,\alpha}(\rho_n\Vert\sigma_n)<nr,
\qquad
0<\alpha<\alpha_0.
\end{equation}
Thus
\begin{equation}
h_{\alpha,r}(\rho_n\Vert\sigma_n)\leq0,
\qquad
0<\alpha<\alpha_0.
\end{equation}
On the other hand,
$h_{\alpha,r}(\rho_n\Vert\sigma_n)$ is continuous in $\alpha$ on
$[\alpha_0,1)$ and satisfies
\begin{equation}
\lim_{\alpha\to1}
h_{\alpha,r}(\rho_n\Vert\sigma_n)=0.
\end{equation}
It is therefore bounded from above on $[\alpha_0,1)$, and hence
\begin{equation}
H_{n,r}(\rho_n\Vert\sigma_n)<\infty.
\end{equation}

We next consider the dependence on $r$. For every fixed
$\alpha\in(0,1)$,
\begin{equation}
h_{\alpha,r}(\rho_n\Vert\sigma_n)
=
\frac{1-\alpha}{\alpha}
D_{\Petz,\alpha}(\rho_n\Vert\sigma_n)
-
n\frac{1-\alpha}{\alpha}r
\end{equation}
is an affine function of $r$, and
\begin{equation}
\frac{\partial}{\partial r}
h_{\alpha,r}(\rho_n\Vert\sigma_n)
=
-n\frac{1-\alpha}{\alpha}<0.
\end{equation}
Hence, for every fixed $\alpha$, the function
$r\mapsto h_{\alpha,r}(\rho_n\Vert\sigma_n)$ is strictly decreasing.
Taking the supremum over $\alpha\in(0,1)$ therefore shows that
\begin{equation}
r\longmapsto H_{n,r}(\rho_n\Vert\sigma_n)
\end{equation}
is non-increasing.

Moreover, since a pointwise supremum of affine functions is convex,
for any $r_1,r_2>0$ and $\lambda\in[0,1]$,
\begin{align}
&H_{n,\lambda r_1+(1-\lambda)r_2}
(\rho_n\Vert\sigma_n)
\nonumber\\
&\qquad\leq
\lambda H_{n,r_1}(\rho_n\Vert\sigma_n)
+
(1-\lambda)
H_{n,r_2}(\rho_n\Vert\sigma_n).
\end{align}
We have already shown that $H_{n,r}(\rho_n\Vert\sigma_n)$ is finite
for every
\begin{equation}
r>\frac1nD_{\Petz,0}(\rho_n\Vert\sigma_n).
\end{equation}
Since every finite convex function is continuous on the interior of
its domain, it follows that
\begin{equation}
r\longmapsto H_{n,r}(\rho_n\Vert\sigma_n)
\end{equation}
is continuous on
\begin{equation}
\left(
\frac1nD_{\Petz,0}(\rho_n\Vert\sigma_n),
\infty
\right).
\end{equation}
This proves all the assertions in Item~3.

\emph{4) Subadditivity.}
The Petz R\'enyi divergence is additive under tensor products:
\begin{align}
&D_{\Petz,\alpha}
(\rho_m\otimes\rho_n\Vert\sigma_m\otimes\sigma_n)
\nonumber\\
&\qquad=
D_{\Petz,\alpha}(\rho_m\Vert\sigma_m)
+
D_{\Petz,\alpha}(\rho_n\Vert\sigma_n).
\end{align}
Hence
\begin{align}
&H_{m+n,r}
(\rho_m\otimes\rho_n\Vert\sigma_m\otimes\sigma_n)
\nonumber\\
&=
\sup_{\alpha\in(0,1)}
\Bigl[
h_{\alpha,r}(\rho_m\Vert\sigma_m)
+
h_{\alpha,r}(\rho_n\Vert\sigma_n)
\Bigr]
\nonumber\\
&\leq
H_{m,r}(\rho_m\Vert\sigma_m)
+
H_{n,r}(\rho_n\Vert\sigma_n).
\end{align}

For two stable sequences $\ms A$ and $\ms B$, stability gives
\begin{equation}
\ms A_m\otimes\ms A_n\subseteq\ms A_{m+n},
\qquad
\ms B_m\otimes\ms B_n\subseteq\ms B_{m+n}.
\end{equation}
Therefore,
\begin{align}
&H_{m+n,r}(\ms A_{m+n}\Vert\ms B_{m+n})
\nonumber\\
&\leq
\inf_{\substack{\rho_m\in\ms A_m,\rho_n\in\ms A_n\\
                 \sigma_m\in\ms B_m,\sigma_n\in\ms B_n}}
H_{m+n,r}
(\rho_m\otimes\rho_n\Vert\sigma_m\otimes\sigma_n)
\nonumber\\
&\leq
H_{m,r}(\ms A_m\Vert\ms B_m)
+
H_{n,r}(\ms A_n\Vert\ms B_n),
\end{align}
which proves the stated set-valued subadditivity.

\emph{5) Positivity and faithfulness.}
The Petz R\'enyi divergence is non-decreasing in its order and
\begin{equation}
\lim_{\alpha\to1}
D_{\Petz,\alpha}(\rho_n\Vert\sigma_n)
=
D(\rho_n\Vert\sigma_n).
\end{equation}
If
\begin{equation}
nr\geq D(\rho_n\Vert\sigma_n),
\end{equation}
then
\begin{equation}
D_{\Petz,\alpha}(\rho_n\Vert\sigma_n)
\leq
D(\rho_n\Vert\sigma_n)
\leq nr
\end{equation}
for all $\alpha\in(0,1)$, and hence
$h_{\alpha,r}\leq0$. On the other hand,
$h_{\alpha,r}\to0$ as $\alpha\to1$, so
\begin{equation}
H_{n,r}(\rho_n\Vert\sigma_n)=0.
\end{equation}

Conversely, if
\begin{equation}
nr<D(\rho_n\Vert\sigma_n),
\end{equation}
then by the convergence
$D_{\Petz,\alpha}\to D$ as $\alpha\to1$, there exists
$\alpha\in(0,1)$ such that
\begin{equation}
D_{\Petz,\alpha}(\rho_n\Vert\sigma_n)>nr.
\end{equation}
Therefore
\begin{equation}
h_{\alpha,r}(\rho_n\Vert\sigma_n)>0,
\end{equation}
and hence
\begin{equation}
H_{n,r}(\rho_n\Vert\sigma_n)>0.
\end{equation}
The same conclusion holds when
$D(\rho_n\Vert\sigma_n)=+\infty$, since
$D_{\Petz,\alpha}$ becomes arbitrarily large as $\alpha\to1$.
Thus
\begin{equation}
H_{n,r}(\rho_n\Vert\sigma_n)\geq0,
\end{equation}
with equality if and only if
$nr\geq D(\rho_n\Vert\sigma_n)$.

\emph{6) Isometric invariance.}
For every isometry $V$ and every $\alpha\in(0,1)$,
\begin{equation}
D_{\Petz,\alpha}(\rho_n\Vert\sigma_n)
=
D_{\Petz,\alpha}
(V\rho_nV^\dagger\Vert V\sigma_nV^\dagger).
\end{equation}
Substituting this identity into the definition and taking the supremum
over $\alpha$ immediately gives
\begin{equation}
H_{n,r}(\rho_n\Vert\sigma_n)
=
H_{n,r}
(V\rho_nV^\dagger\Vert V\sigma_nV^\dagger).
\end{equation}
This completes the proof.
\end{proof}

\begin{proof}[Proof of Proposition~\ref{prop: Hoeffding anti}]
For every $\alpha>1$, define
\begin{equation}
\widetilde{h}_{\alpha,r}(\rho_n\Vert\sigma_n)
:=
\frac{\alpha-1}{\alpha}
\left(
nr-D_{\Sand,\alpha}(\rho_n\Vert\sigma_n)
\right),
\end{equation}
so that
\begin{equation}
\widetilde{H}_{n,r}(\rho_n\Vert\sigma_n)
=
\sup_{\alpha>1}
\widetilde{h}_{\alpha,r}(\rho_n\Vert\sigma_n).
\end{equation}

\emph{1) Anti-data-processing inequality.}
For every $\alpha>1$, the sandwiched R\'enyi divergence satisfies
\begin{equation}
D_{\Sand,\alpha}
\bigl(
\Lambda(\rho_n)\Vert\Lambda(\sigma_n)
\bigr)
\leq
D_{\Sand,\alpha}(\rho_n\Vert\sigma_n).
\end{equation}
Since $(\alpha-1)/\alpha>0$, this implies
\begin{equation}
\widetilde{h}_{\alpha,r}
\bigl(
\Lambda(\rho_n)\Vert\Lambda(\sigma_n)
\bigr)
\geq
\widetilde{h}_{\alpha,r}(\rho_n\Vert\sigma_n).
\end{equation}
Taking the supremum over $\alpha>1$ proves
\begin{equation}
\widetilde{H}_{n,r}
\bigl(
\Lambda(\rho_n)\Vert\Lambda(\sigma_n)
\bigr)
\geq
\widetilde{H}_{n,r}(\rho_n\Vert\sigma_n).
\end{equation}

\emph{2) Mixture Lower Bound.}
Let
\begin{equation}
\bar\rho_n:=\sum_i p_i\rho_n^i,
\qquad
\bar\sigma_n:=\sum_i p_i\sigma_n^i.
\end{equation}
Introduce the block-diagonal states
\begin{equation}
\widehat{\rho}_n
:=
\sum_i p_i|i\rangle\langle i|\otimes\rho_n^i,
\qquad
\widehat{\sigma}_n
:=
\sum_i p_i|i\rangle\langle i|\otimes\sigma_n^i.
\end{equation}
Tracing out the classical register gives
$\bar\rho_n$ and $\bar\sigma_n$. Hence, by data processing,
\begin{equation}
D_{\Sand,\alpha}
(\bar\rho_n\Vert\bar\sigma_n)
\leq
D_{\Sand,\alpha}
(\widehat{\rho}_n\Vert\widehat{\sigma}_n).
\end{equation}
For the block-diagonal states,
\begin{align}
&D_{\Sand,\alpha}
(\widehat{\rho}_n\Vert\widehat{\sigma}_n)
\nonumber\\
&=
\frac{1}{\alpha-1}
\log
\left[
\sum_i p_i
2^{(\alpha-1)
D_{\Sand,\alpha}(\rho_n^i\Vert\sigma_n^i)}
\right]
\nonumber\\
&\leq
\max_i
D_{\Sand,\alpha}(\rho_n^i\Vert\sigma_n^i).
\end{align}
Therefore,
\begin{equation}
D_{\Sand,\alpha}
(\bar\rho_n\Vert\bar\sigma_n)
\leq
\max_i
D_{\Sand,\alpha}(\rho_n^i\Vert\sigma_n^i).
\end{equation}
It follows that
\begin{align}
\widetilde{H}_{n,r}
(\bar\rho_n\Vert\bar\sigma_n)
&=
\sup_{\alpha>1}
\frac{\alpha-1}{\alpha}
\left(
nr-
D_{\Sand,\alpha}
(\bar\rho_n\Vert\bar\sigma_n)
\right)
\nonumber\\
&\geq
\sup_{\alpha>1}
\frac{\alpha-1}{\alpha}
\left(
nr-
\max_i
D_{\Sand,\alpha}(\rho_n^i\Vert\sigma_n^i)
\right),
\end{align}
which is the claimed mixture lower bound.

\emph{3) Continuity, convexity, and monotonicity in $r$.}
For every fixed $\alpha>1$,
$r\mapsto\widetilde{h}_{\alpha,r}$ is affine with slope
\begin{equation}
n\frac{\alpha-1}{\alpha}\in(0,n).
\end{equation}
Therefore
$r\mapsto\widetilde{H}_{n,r}$, being a supremum of affine functions,
is convex and non-decreasing.

Suppose first that
\begin{equation}
\operatorname{supp}(\rho_n)
\subseteq
\operatorname{supp}(\sigma_n).
\end{equation}
Then $D_{\Sand,\alpha}(\rho_n\Vert\sigma_n)<\infty$ for every
$\alpha>1$. Moreover, for $r_2>r_1$,
\begin{align}
\widetilde{h}_{\alpha,r_2}
&=
\widetilde{h}_{\alpha,r_1}
+
n\frac{\alpha-1}{\alpha}(r_2-r_1)
\nonumber\\
&\leq
\widetilde{h}_{\alpha,r_1}
+n(r_2-r_1).
\end{align}
Taking the supremum over $\alpha>1$ gives
\begin{equation}
0
\leq
\widetilde{H}_{n,r_2}
-
\widetilde{H}_{n,r_1}
\leq
n(r_2-r_1).
\end{equation}
Thus $\widetilde{H}_{n,r}$ is $n$-Lipschitz, and hence continuous in
$r$.

If
$\operatorname{supp}(\rho_n)\nsubseteq\operatorname{supp}(\sigma_n)$,
then
$D_{\Sand,\alpha}(\rho_n\Vert\sigma_n)=+\infty$ for every
$\alpha>1$, so
\begin{equation}
\widetilde{H}_{n,r}(\rho_n\Vert\sigma_n)=-\infty
\end{equation}
for all $r$. Thus it is constant in $r$ in the extended-real sense.
This proves Item~3.

\emph{4) Subadditivity.}
The sandwiched R\'enyi divergence is additive under tensor products:
\begin{align}
&D_{\Sand,\alpha}
(\rho_m\otimes\rho_n\Vert\sigma_m\otimes\sigma_n)
\nonumber\\
&\qquad=
D_{\Sand,\alpha}(\rho_m\Vert\sigma_m)
+
D_{\Sand,\alpha}(\rho_n\Vert\sigma_n).
\end{align}
Hence
\begin{align}
&\widetilde{H}_{m+n,r}
(\rho_m\otimes\rho_n\Vert\sigma_m\otimes\sigma_n)
\nonumber\\
&=
\sup_{\alpha>1}
\Bigl[
\widetilde{h}_{\alpha,r}(\rho_m\Vert\sigma_m)
+
\widetilde{h}_{\alpha,r}(\rho_n\Vert\sigma_n)
\Bigr]
\nonumber\\
&\leq
\widetilde{H}_{m,r}(\rho_m\Vert\sigma_m)
+
\widetilde{H}_{n,r}(\rho_n\Vert\sigma_n),
\end{align}
which proves subadditivity.

\emph{5) Positivity and faithfulness.}
Suppose first that
\begin{equation}
\operatorname{supp}(\rho_n)
\subseteq
\operatorname{supp}(\sigma_n).
\end{equation}
The sandwiched R\'enyi divergence is non-decreasing in $\alpha$ and
satisfies
\begin{equation}
\lim_{\alpha\to1}
D_{\Sand,\alpha}(\rho_n\Vert\sigma_n)
=
D(\rho_n\Vert\sigma_n).
\end{equation}
If
\begin{equation}
nr\leq D(\rho_n\Vert\sigma_n),
\end{equation}
then for every $\alpha>1$,
\begin{equation}
D_{\Sand,\alpha}(\rho_n\Vert\sigma_n)
\geq
D(\rho_n\Vert\sigma_n)
\geq nr.
\end{equation}
Hence
\begin{equation}
\widetilde{h}_{\alpha,r}(\rho_n\Vert\sigma_n)\leq0.
\end{equation}
On the other hand,
\begin{equation}
\lim_{\alpha\to1}
\widetilde{h}_{\alpha,r}(\rho_n\Vert\sigma_n)
=0,
\end{equation}
and therefore
\begin{equation}
\widetilde{H}_{n,r}(\rho_n\Vert\sigma_n)=0.
\end{equation}

Conversely, if
\begin{equation}
nr>D(\rho_n\Vert\sigma_n),
\end{equation}
then, by continuity of the sandwiched R\'enyi divergence at
$\alpha=1$, there exists some $\alpha>1$ sufficiently close to $1$
such that
\begin{equation}
D_{\Sand,\alpha}(\rho_n\Vert\sigma_n)<nr.
\end{equation}
Hence
\begin{equation}
\widetilde{h}_{\alpha,r}(\rho_n\Vert\sigma_n)>0,
\end{equation}
so
\begin{equation}
\widetilde{H}_{n,r}(\rho_n\Vert\sigma_n)>0.
\end{equation}
Thus, when the support condition holds,
\begin{equation}
\widetilde{H}_{n,r}(\rho_n\Vert\sigma_n)\geq0,
\end{equation}
with equality if and only if
$nr\leq D(\rho_n\Vert\sigma_n)$.

If instead
\begin{equation}
\operatorname{supp}(\rho_n)
\nsubseteq
\operatorname{supp}(\sigma_n),
\end{equation}
then
\begin{equation}
D_{\Sand,\alpha}(\rho_n\Vert\sigma_n)=+\infty,
\qquad
\forall\alpha>1,
\end{equation}
and consequently
\begin{equation}
\widetilde{H}_{n,r}(\rho_n\Vert\sigma_n)=-\infty.
\end{equation}

\emph{6) Isometric invariance.}
For every isometry $V$ and every $\alpha>1$,
\begin{equation}
D_{\Sand,\alpha}(\rho_n\Vert\sigma_n)
=
D_{\Sand,\alpha}
(V\rho_nV^\dagger\Vert V\sigma_nV^\dagger).
\end{equation}
Substituting this identity into the definition of the quantum Hoeffding
anti-divergence and taking the supremum over $\alpha>1$ gives
\begin{equation}
\widetilde{H}_{n,r}(\rho_n\Vert\sigma_n)
=
\widetilde{H}_{n,r}
(V\rho_nV^\dagger\Vert V\sigma_nV^\dagger).
\end{equation}
This completes the proof.
\end{proof}

}

\end{document}